\documentclass[12pt,english,table]{article}
\usepackage{lmodern}

\usepackage[T1]{fontenc}
\usepackage[utf8]{inputenc}
\usepackage{color}
\usepackage{verbatim}
\usepackage{amsmath}
\usepackage{amsthm}
\usepackage{amssymb}
\usepackage{graphicx}
\usepackage{geometry}
\usepackage{setspace}
\usepackage[authoryear]{natbib}
\makeatletter
\providecommand{\propositionname}{Proposition}
\providecommand{\corollaryname}{Corollary}
\providecommand{\theoremname}{Theorem}
\providecommand{\lemmaname}{Lemma}

\theoremstyle{plain}
\newtheorem{proposition}{\propositionname}
\newtheorem{corollary}{\corollaryname}
\newtheorem{theorem}{\theoremname}
\newtheorem{lemma}{\lemmaname}
\theoremstyle{definition}
\newtheorem{assumption}{Assumption}
\setcitestyle{round}
\usepackage{lscape}
\usepackage{longtable}
\usepackage{xcolor}
\usepackage{rotating}

\usepackage{array}
\usepackage{arydshln}
\usepackage{tabularx}
\usepackage{multirow}
\usepackage{booktabs}
\usepackage{ragged2e}\newcolumntype{C}[1]{>{\centering\arraybackslash}p{#1}}
\usepackage{rotfloat}
\newcolumntype{J}[1]{>{\justify\arraybackslash}p{#1}}
\newcolumntype{R}[1]{>{\RaggedLeft\arraybackslash}p{#1}}
\newcolumntype{Q}[1]{>{\columncolor{Gray}\RaggedLeft\arraybackslash}p{#1}}
\newcolumntype{L}[1]{>{\RaggedRight\arraybackslash}p{#1}}
\newcolumntype{G}{@{\extracolsep{0.5cm}}l@{\extracolsep{0pt}}}%
\newcolumntype{P}[1]{>{\centering\arraybackslash}p{#1}}
\newcolumntype{Y}{>{\centering\arraybackslash}X}
\newcommand{\nhphantom}[1]{\sbox0{#1}\hspace{-\the\wd0}} % command to add negative space, like off-set a minussign

\AtBeginDocument{

}

\usepackage[colorlinks,
            linkcolor=blue,
            anchorcolor=blue,
            citecolor=blue]{hyperref}
\usepackage{babel}
\usepackage{titlesec}

\makeatletter
\renewcommand*{\fps@figure}{htb}
\makeatother

\makeatother

\usepackage{babel}
\providecommand{\corollaryname}{Corollary}
\providecommand{\lemmaname}{Lemma}
\providecommand{\propositionname}{Proposition}
\providecommand{\theoremname}{Theorem}

\begin{document}

\def\spacingset#1{\renewcommand{\baselinestretch}%
{#1}\small\normalsize} \spacingset{1}

%%%%%%%%%%%%%%%%%%%%%%%%%%%%%%%%%%%%%%%%%%%%%%%%%%%%%%%%%%%%%%%%%%%%%%%%%%%%%%

 \title{\bf Exact Likelihood Inference and Robust Filtering for Gauss-Cauchy Convolution Models}
  \author{Peter Reinhard Hansen\thanks{Correspondence: Peter Reinhard Hansen (hansen@unc.edu). We thank Yacine A\"{\i}t-Sahalia for pointing out the connection to discretely sampled L\'evy processes, George Tauchen for pointing out the connection to Mills ratio, Mitchell Watt for tracing the reference to \citet{WeibullMattssonVoorneveld:2007}, Jihyun Kim for discussing the paper in Toulouse, and participants at the 2026 Triangle Conference at UNC, the 2026 Financial Econometrics Conference at the Toulouse School of Economics, and the 2026 Ca' Foscari conference on score-driven and nonlinear time series models for valuable comments. Chen Tong acknowledges financial support from the Youth Fund of the National Natural Science Foundation of China (72301227) and the Fujian Provincial Natural Science Foundation of China (2025J08008).}\hspace{.2cm}\\
   {\small Department of Economics, University of North Carolina at Chapel Hill}\\
     \\
    Chen Tong \\
     {\small School of Economics, Xiamen University}}
  \maketitle

\begin{abstract}
The convolution of a Gaussian and a Cauchy distribution, known as the Voigt distribution, is widely used in spectroscopy and provides a natural framework for modeling heavy-tailed measurement noise. We derive analytical expressions for its density, score, Hessian, Fisher information, and conditional moments using the scaled complementary error function, enabling stable maximum likelihood estimation without numerical convolution, finite-difference derivatives, or pseudo-Voigt approximations. The conditional expectation of the latent Gaussian component is governed by a redescending location score, so extreme observations are automatically discounted rather than propagated. This structure leads to the Gauss-Cauchy Convolution (GCC) filter for state-space models with Gaussian latent dynamics and Voigt measurement errors, where the Masreliez Gaussian prediction approximation preserves a Voigt prediction-error density. In an application to log realized volatility for the Technology Select Sector SPDR Fund, the GCC filter separates persistent latent variation from transient measurement noise and attains the highest implemented prediction-error criterion among the Gaussian, Student-$t$, Huber, and related filtering specifications considered.
\end{abstract}

\noindent%
{\it Keywords:} Voigt profile, Convolutions, Heavy Tails, Robust Filtering, Kalman Filter
\vfill

\newpage
\spacingset{1.5} 

\everymath{\setlength{\abovedisplayskip}{11pt} \setlength{\belowdisplayskip}{11pt}} 
\everydisplay{\setlength{\abovedisplayskip}{11pt} \setlength{\belowdisplayskip}{11pt}} 

\section{Introduction}

The convolution of a Gaussian and a Cauchy distribution, known as the Voigt distribution, arises naturally when finite-variance background noise is combined with heavy-tailed disturbances. The Voigt profile is widely used in spectroscopy, where Gaussian broadening is associated with Doppler effects and Cauchy broadening with pressure or lifetime effects. It is also a natural
statistical model for measurement errors that combine ordinary noise with
occasional extreme observations.

Although an expression for the Voigt density has been known since \citet{Kendall_D:1938}, its use as a likelihood-based statistical model has been limited by the perception that exact inference is computationally
difficult. Because the density is not an elementary function, applications often rely on numerical convolution, pseudo-Voigt approximations, nonlinear
least squares, or simulation-based methods. We show that these compromises are unnecessary for likelihood inference. The scaled complementary error function provides a stable representation of the Voigt density. Its differential identities are known special-function results, but their statistical implication is important: they yield closed-form expressions for Voigt likelihood derivatives, Fisher information, and conditional moments. The
contribution is therefore not the density representation alone, but the statistical calculus it enables: the Gaussian-Cauchy convolution is an unusual
heavy-tailed Gaussian convolution model for which likelihood evaluation, differentiation, conditional moments, and Masreliez-type filtering remain
analytically tractable.

This paper develops likelihood inference and filtering methods that exploit this tractable heavy-tailed Gaussian convolution structure. We first derive
closed-form expressions for the density, score, Hessian, Fisher information, and conditional moments of the Voigt distribution. These expressions reduce
inference to evaluations of a standard special function and algebraic operations. We then establish the asymptotic properties of the maximum likelihood estimator and document its numerical performance in simulations. 

As a further implication, this tractability clarifies a special case of likelihood inference for discretely sampled L\'evy processes. \citet{AitSahaliaJacod:2008} study Fisher information for models of the form $X_t=\sigma W_t+\theta Y_t$, where $W$ is a stable process and $Y$ is an independent
L\'evy perturbation. In their Brownian-Cauchy example, increments of $X$ are exactly Gaussian-Cauchy convolutions, with Gaussian scale proportional to $\sqrt{\Delta}$ and Cauchy scale proportional to $\Delta$. Their numerical Fisher-information calculations for this case therefore require repeated evaluation of the convolution density and its derivatives.
Our formulas replace this convolution step by closed-form Voigt evaluations and algebraic score and Hessian expressions. Thus the Gauss-Cauchy convolution provides a finite-sample likelihood building block for the Brownian-plus-Cauchy L\'evy model, complementing the
high-frequency Fisher-information limits in \citet{AitSahaliaJacod:2008}.

The same analytical structure is useful for robust signal extraction. If an observation is the sum of a latent Gaussian signal and Cauchy measurement noise, the conditional expectation of the Gaussian component is nonlinear in the observation. The associated location score is approximately linear near the center of the distribution but redescends in the tails. Hence, extreme observations are automatically discounted rather than propagated into the latent
signal. 
This differs from the Kalman update, which is linear in the prediction error, and from Student-$t$ or Huber-type procedures, where tail attenuation is introduced through a heavy-tailed observation density or a robust loss function. In the Gauss-Cauchy convolution, the redescending update is
instead an implication of signal extraction: it is the exact conditional mean of
the latent Gaussian component when the observation is contaminated by an additive Cauchy component. Thus, robustness is not imposed as an influence
function but derived from the convolution structure itself. Robustness has several meanings in filtering. In the sense of \citet{CalvetCzellarRonchetti:2015}, it concerns stability of the inferred state distribution under small departures from the assumed model. Our robustness claim is narrower and operational, in the Masreliez-Martin influence-function tradition: the GCC conditional-mean update has redescending influence, so sufficiently large prediction errors are assigned mainly to the Cauchy measurement component rather than to the persistent state.

We use this result to construct the Gauss-Cauchy Convolution (GCC) filter for linear state-space models with Gaussian latent dynamics and Voigt measurement
errors. 
Following the Masreliez Gaussian prediction approximation \citep{Masreliez:1975, MasreliezMartin:1977}, the key step is that Tweedie's formula becomes applicable to the state-prediction error, while the Gaussian-Cauchy convolution structure keeps the prediction-error density inside the Voigt family. Specifically, if the conditional state-prediction error is approximated as Gaussian with variance $h_{t|t-1}$, then
$$
N(0,h_{t|t-1})+\mathcal V(0,\sigma,\gamma) = \mathcal V(0,\delta_t,\gamma), \qquad \delta_t^2=h_{t|t-1}+\sigma^2.
$$
Thus the state update is the Tweedie conditional-mean correction associated
with the Voigt prediction-error score. The filter nests the Kalman filter as
the special case $\gamma=0$ and the pure Cauchy measurement-error filter when  $\sigma=0$. We assess the Masreliez approximation directly by
comparing the GCC recursion with an exact benchmark filter based on numerical
density propagation. The approximation error is small at both the
predictive-density level and the one-step-correction level, especially in the
empirically relevant region where the estimated Cauchy component is small
relative to the Gaussian component.

We illustrate the method using daily log realized volatility for the Technology Select Sector SPDR Fund (XLK). The realized-volatility series contains both persistent volatility movements and transient extreme observations associated with market stress, liquidity disruptions, and microstructure effects. The GCC filter separates these components by retaining a Gaussian core for ordinary measurement variation while using the Cauchy component to absorb occasional large deviations.
In the empirical application, the GCC specification attains the highest implemented prediction-error criterion among the competing filtering specifications considered, including Gaussian, Cauchy, Normal-Laplace, Student-$t$, and Huber alternatives.

The paper is related to several strands of literature. It contributes to work on convolution-based distributions and Voigt-profile estimation by showing that the special-function representation yields a tractable likelihood theory. 
It is also related to robust filtering and score-driven dynamics
\citep{CrealKoopmanLucas:2013, Harvey:2013}, where heavy-tailed observation densities are used to reduce the influence of outliers.
It also connects to work on the optimality of score-driven updates under information-theoretic criteria \citep{BlasquesKoopmanLucas:2015,
GorgiLauriaLuati:2024}. Recent contributions using Student-$t$ scores for robust score-driven dynamics include the univariate signal-extraction
filter of \citet{HarveyLuati:2014}, the multivariate dynamic-location filter of \citet{DInnocenzoLuatiMazzocchi:2023}, and the spatio-temporal model of \citet{GasperoniLuatiPaciDInnocenzo:2023}. 
The GCC filter differs by deriving the redescending update as the conditional mean correction
implied by the Gaussian-Cauchy convolution structure, rather than by specifying a heavy-tailed observation density or mixture structure directly.

The redescending conditional expectation also connects the GCC model to recent theoretical work on information aggregation with thin- and heavy-tailed signals. In global games, \citet{MorrisYildiz:2019} show that large shocks can trigger equilibrium shifts when agents rationally attribute extreme observations to a common heavy-tailed component rather than to idiosyncratic noise. A related ``too good to be true'' logic appears in \citet{WeibullMattssonVoorneveld:2007}, where sufficiently extreme signals are discounted because they are more likely to reflect noise than fundamentals. \citet{HautschHessMuller:2012} extend this exact mechanism to financial price discovery, demonstrating that fat-tailed noise causes rational agents to optimally discount extreme signals. The Gauss-Cauchy convolution provides a tractable parametric realization of this behavior: the conditional expectation of the latent Gaussian component is explicitly redescending, and the parameters $(\sigma,\gamma)$ dictate precisely when extreme observations are treated as valid signals and when they are discarded as noise.

The remainder of the paper is organized as follows. Section
\ref{sec:Voigt} defines the Gauss-Cauchy convolution and derives its conditional
moment structure. Section \ref{sec:MLE} develops likelihood inference for the
Voigt distribution. Section \ref{sec:GCCfilter} introduces the GCC filter.
Section \ref{sec:Simulations} evaluates the Masreliez approximation using exact
benchmark filtering. Section \ref{sec:Empirical} applies the filter to log
realized volatility. Section \ref{sec:Conclusion} concludes. Proofs, additional results, and details on the competing filters are collected in the Supplement.

\section{Gauss-Cauchy Convolution Distributions}\label{sec:Voigt}

Let $Z\sim \mathcal{N}(0,\sigma^{2})$ and $X\sim\operatorname{Cauchy}(0,\gamma)$ be independent, and consider the convolution
$$
Y = \mu + Z + X,
$$
where $\sigma,\gamma>0$. We denote the resulting distribution by $\mathcal{V}(\mu,\sigma,\gamma)$ because it is known as the \emph{Voigt profile} in the field
of spectroscopy, and let $f_Y(y;\theta)$ denote its density, where $\theta=(\mu,\sigma,\gamma)^\prime$. 
This density can be expressed in terms of the scaled complementary error function, $\mathsf{e}(w)
\equiv
\operatorname{erfcx}(w)$, where $\mathsf{e}(-iw)$ is the Faddeeva function. For $\operatorname{Re}(w)>0$, we have
$$
\mathsf{e}(w)
=
\frac{1}{\pi}
\int_{-\infty}^{\infty}
\frac{\exp(-t^2)}{w+it}dt,
$$
and for $Z\sim \mathcal{N}(0,1)$ we have 
$\Pr(|Z|>r)=\mathsf{e}(\tfrac{r}{\sqrt{2}})e^{-r^2/2}$ for $r>0$. 

We are particularly interested in evaluating $\mathsf{e}(w)$ on the vertical
\emph{Voigt line} in $\mathbb{C}$ traced out by  $y\mapsto w_{y,\theta} = \tfrac{\gamma+i(y-\mu)}{\sigma\sqrt{2}}$, and we introduce notation for real and imaginary parts of $\mathsf{e}(w_{y,\theta})$
$$
\mathsf{u}(y;\theta)
\equiv
\operatorname{Re}[\mathsf{e}(w_{y,\theta})],
\qquad
\mathsf{v}(y;\theta)
\equiv
\operatorname{Im}[\mathsf{e}(w_{y,\theta})],
$$
so that $\mathsf{e}(w_{y,\theta})=
\mathsf{u}(y;\theta)+i\mathsf{v}(y;\theta)$.

\begin{proposition}[Kendall, 1938]
\label{prop:Voigt}The density of the Gauss-Cauchy convolution, $\mathcal{V}(\mu,\sigma,\gamma)$, can be expressed
\begin{equation}
f_Y(y;\theta)=\frac{1}{\sqrt{2\pi\sigma^{2}}}\mathsf{u}(y;\theta),\qquad \theta=(\mu,\sigma,\gamma)^\prime.\label{eq:Kendall}
\end{equation}
\end{proposition}
The expression for this density was first obtained by \citet{Kendall_D:1938}. The proof of Proposition \ref{prop:Voigt} is given in Supplement~\ref{sec:SupplementA} and is based on the characteristic function $\varphi_{Z+X}(s)=e^{-\sigma^2s^2/2-\gamma|s|}$. Figure \ref{fig:DenCauchyNormal} shows the density of $Y$ for $\mu=0$ and $\sigma=\gamma=1$, together with the best approximating Student-$t$ distribution, as defined by the Kullback-Leibler discrepancy. The visible differences highlight the specification error from using Student-$t$ distributions as proxies.\footnote{The Voigt density can be viewed as a limiting case of the convolution of two Student-$t$ distributions, a class of convolution problems known to be analytically difficult. For recent progress in this area, see \citet{Nason:2006}, \citet{Forchini:2008}, \citet{BergVignat:2010}, and \citet{HansenTong:2026}.} A second common approximation is the pseudo-Voigt profile, a two-component mixture of a Gaussian density and a Cauchy density with the same location.\footnote{The pseudo-Voigt density takes the form $(1-\eta_p)\phi(y;\mu,\sigma_p)+\eta_p c(y;\mu,\gamma_p)$, with $0\leq\eta_p\leq1$, where $\phi$ is the Gaussian density and $c$ is the Cauchy density.} It provides a closer density approximation in this example, but it is not a convolution and therefore does not preserve the additive-noise interpretation used below.
The Voigt profile uniquely captures the superposition of finite-variance noise and infinite-variance jumps, a structure that generic heavy-tailed approximations fail to reproduce.
\begin{figure}[tbh]
\spacingset{1.15}
\begin{centering}
\includegraphics[width=0.49\textwidth]{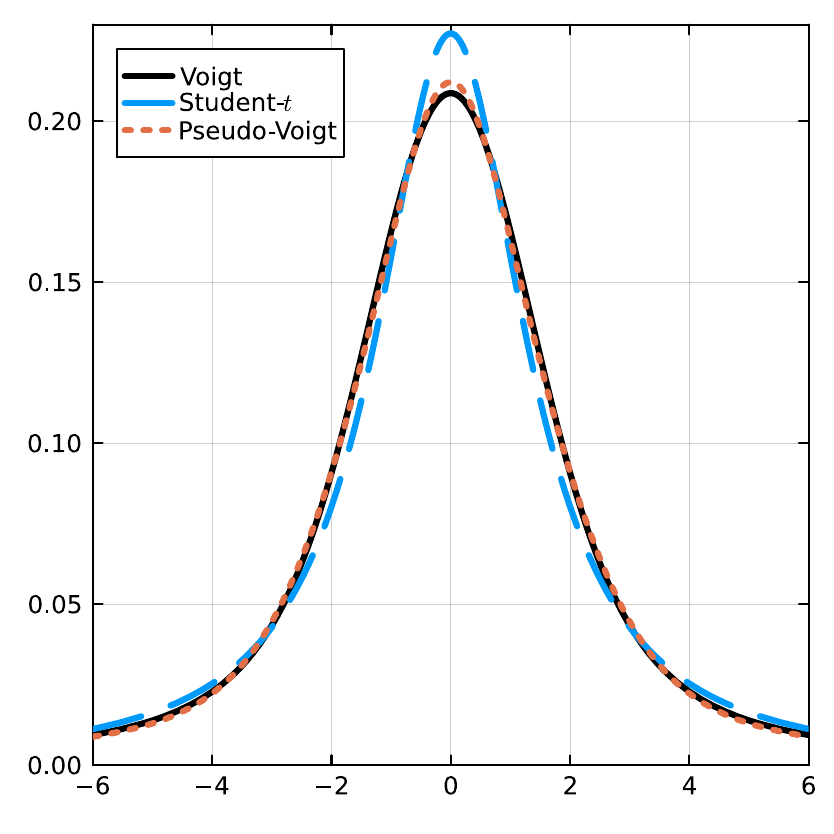}
\includegraphics[width=0.5\textwidth]{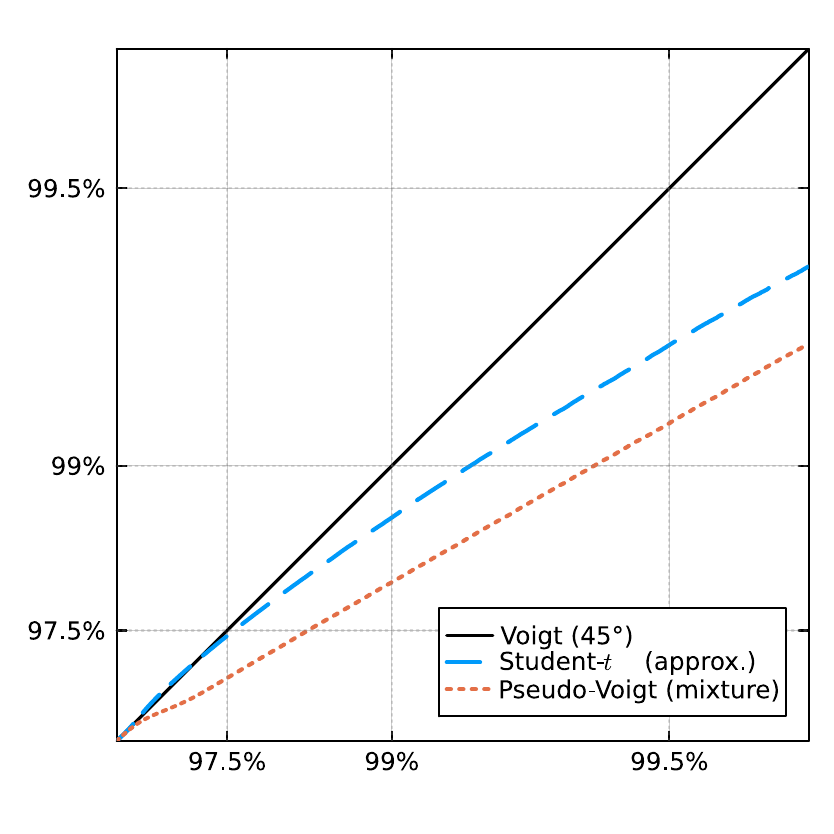}
\par\end{centering}
\caption{{\footnotesize The left panel shows the Voigt density, $\mathcal{V}(0,1,1)$, together with the best approximating Student-$t$ distribution, $(\sigma_t,\nu_t)\approx(1.45,1.22)$, and pseudo-Voigt mixture, $(\eta_p,\sigma_p,\gamma_p)\approx(0.65,1.62,1.65)$. The right panel shows the corresponding Q-Q plots. The approximation parameters are chosen to minimize the Kullback-Leibler discrepancy, but visible discrepancies remain.\label{fig:DenCauchyNormal}}}
\end{figure}

Interestingly, the density (\ref{eq:Kendall}) can be expressed from the Mills ratio for a standard normal random variable, $m(t)=\Phi(-t)/\phi(t)$,
where $\phi$ and $\Phi$ are the density and cumulative distribution function for $\mathcal{N}(0,1)$, respectively. This representation highlights
the connection between the Voigt profile and the hazard rate of the normal distribution, providing another avenue for numerical evaluation.

\begin{corollary}[Voigt density and Mills ratio]\label{cor:MillsRatio}
The Voigt density can be expressed in terms of the complex Mills ratio for a standard normal variable: $f_{Y}(y;\theta)=\frac{1}{\pi\sigma}\operatorname{Re}\left[m(\tfrac{\gamma+i(y-\mu)}{\sigma})\right]$, $\theta=(\mu,\sigma,\gamma)^\prime$, where $f_{Y}$ is the density in (\ref{eq:Kendall}), and $m(t)=\Phi(-t)/\phi(t)$.
\end{corollary}

\begin{lemma}\label{lem:erfcx}
Let $\sigma>0$ and $\gamma>0$. Then: (i) $\mathsf{e}$ is entire and satisfies
$
\mathsf{e}^\prime(w) = 2w\mathsf{e}(w)-\frac{2}{\sqrt{\pi}}.
$
Moreover, for every integer $n\geq0$, there exist polynomials $p_n$ and
$q_n$ such that
$
\mathsf{e}^{(n)}(w)=p_n(w)\mathsf{e}(w)+q_n(w).
$
(ii) The real and imaginary parts satisfy the symmetry relations
$\mathsf{u}(\mu+r;\theta)=\mathsf{u}(\mu-r;\theta)$,
$\mathsf{v}(\mu+r;\theta)=-\mathsf{v}(\mu-r;\theta)$,
and $\mathsf{u}(y;\theta)>0$ for all real $y$. 
\end{lemma}

Lemma \ref{lem:erfcx} collects the special-function identities that make the Voigt likelihood tractable. Part (i) states that $\mathsf{e}=\operatorname{erfcx}$ is closed under differentiation in the sense that every derivative is an algebraic function of $w$ and $\mathsf{e}(w)$.\footnote{Special-function representations of the Voigt profile are discussed by \citet{PagniniSaxena:2008}, and reliable numerical algorithms for the Voigt/Faddeeva function are well developed; see, for example, \citet{Humlicek:1982} and \citet{ZaghloulAli:2011}. The derivative identities for $\operatorname{erfcx}$ and the Faddeeva function are known special-function results. The contribution here is statistical: once the Voigt density is written in terms of $\operatorname{erfcx}$, the likelihood derivatives and conditional moments for the Gauss-Cauchy convolution follow algebraically from the same complex-function evaluation used for the density.} Part (ii) gives the positivity and symmetry properties needed for the real and imaginary components along the Voigt line.

The positivity of $\mathsf{u}(y;\theta)$ ensures that the density, log-density, score, Hessian, and conditional moment formulas involving ratios such as ${\mathsf{v}(y;\theta)}/{\mathsf{u}(y;\theta)} $ are well defined for all real $y$. The symmetry relations explain the corresponding even-odd structure of the density and the conditional moments derived next.

The main likelihood implication of Lemma \ref{lem:erfcx} is the following algebraic closure property.

\begin{corollary}[Derivatives of Voigt density]\label{cor:VoigtAlgebraicClosure}
All finite-order derivatives of the Voigt density,
$f_Y(y;\theta)$, with respect to $y$, $\mu$, $\sigma$, and $\gamma$ are
algebraic functions of $w_{y,\theta}$, $\mathsf{u}(y;\theta)$,
$\mathsf{v}(y;\theta)$, and the parameters $\theta=(\mu,\sigma,\gamma)^\prime$. Consequently, all finite-order derivatives of the log-density $\log f_Y(y;\theta)$ can be expressed as rational functions of the same quantities, with denominators involving only powers of $\mathsf{u}(y;\theta)$ and $\sigma$. Since $\mathsf{u}(y;\theta)>0$ and $\sigma>0$, these expressions are well defined for all real $y$.
\end{corollary}

Corollary \ref{cor:VoigtAlgebraicClosure} is the central tractability result for likelihood inference. Although Voigt density is a Gaussian-Cauchy convolution, the scaled complementary error function absorbs the unstable exponential term, reducing density differentiation to algebraic operations on a single complex-valued function. Thus, the likelihood, score, Hessian, and higher-order derivatives can be evaluated directly from $\mathsf{u}(y;\theta)$ and $\mathsf{v}(y;\theta)$. This bypasses numerical convolution, quadrature, finite-difference, or pseudo-Voigt approximations, making exact maximum likelihood estimation far more direct than the convolution representation suggests.

\subsection{Conditional Distribution, Expectation, and Variance}

We now consider the signal-extraction problem implied by the convolution $
Y=\mu+Z+X$. The objective is to infer the latent Gaussian component $Z$ from the observed
realization of $Y$, treating the Cauchy component $X$ as measurement noise.
This formulation is directly relevant for signal extraction with redescending update influence, where $Z$ represents the latent signal of interest and $X$
captures transient heavy-tailed contamination. The following theorem gives the closed form conditional
density of $Z$ given $Y=y$.
\begin{theorem}
\label{thm:CondDensity}Suppose that $Z\sim \mathcal{N}(0,\sigma^{2})$ and
$X\sim {\operatorname{Cauchy}(0,\gamma)}$ are independent
and define $Y=\mu+Z+X$. Then the conditional density of $Z$ given $Y=y$
is given by 
$$
f_{Z|Y}(z|y)=\frac{1}{\gamma\pi }\frac{1}{\mathsf{u}(y;\theta)}\frac{\exp(-\frac{1}{2}z^{2}/\sigma^{2})}{1+(\frac{y-\mu-z}{\gamma})^{2}},
$$
where $f_{Z|Y}(z|y)\rightarrow f_{Z}(z)=e^{-z^{2}/(2\sigma^{2})}/\sqrt{2\pi\sigma^{2}}$
as $(y-\mu)\rightarrow\pm\infty$. 
\end{theorem}

The conditional moments of $Z$ given $Y$ can be obtained from Tweedie’s formula for Gaussian convolutions, originally attributed to \citet{Robbins:1956}; see \citet{Efron:2011} for a modern treatment. The key point is that the conditional mean of an additive Gaussian component is determined by the score of the marginal density, and the conditional variance by its derivative.

\begin{proposition}[Tweedie's formula for Gaussian convolutions]\label{prop:TweedieGaussian}
Let $Z\sim\mathcal{N}(0,\sigma^2)$, with $\sigma>0$, be independent of a proper real-valued random variable $X$, and define $Y=\mu+Z+X$.
Then $Y$ has a strictly positive and smooth density $f_Y$. Moreover, for all real $y$,
\begin{gather*}
    \mathbb{E}[Z\mid Y=y]
    =
    -\sigma^2\frac{\partial}{\partial y}\log f_Y(y),  \\
    \mathbb{V}(Z\mid Y=y)
    =
    \sigma^2+\sigma^4\frac{\partial^2}{\partial y^2}\log f_Y(y)
    =
    \sigma^2\left(1-\frac{\partial}{\partial y}\mathbb{E}[Z\mid Y=y]\right).
\end{gather*}

\end{proposition}
In the Gauss-Cauchy convolution, Proposition~\ref{prop:Voigt} gives the marginal density $f_Y$ explicitly in terms of $\mathsf{u}(y;\theta)$. Therefore Proposition~\ref{prop:TweedieGaussian} reduces the conditional-moment calculation to differentiating the Voigt log-density. Using the score and Hessian identities in Lemma~\ref{lem:ScoreHessian} gives the following closed-form expressions.
\begin{corollary}
\label{cor:CondExpect}In the Gauss-Cauchy convolution model, the conditional expectation and conditional variance of the latent Gaussian component $Z$ given $Y=y$ are
\begin{eqnarray}
\mathbb{E}\left[Z|Y=y\right]&=&(y-\mu) +\gamma\frac{\mathsf{v}(y;\theta)}{\mathsf{u}(y;\theta)}\label{eq:CondExpect}\\
\mathbb{V}(Z|Y=y) &=& \sqrt{2/\pi}\frac{\sigma\gamma}{\mathsf{u}(y;\theta)} - \gamma^2\left(1 + \frac{\mathsf{v}^{2}(y;\theta)}{\mathsf{u}^{2}(y;\theta)}\right)
\end{eqnarray}
and $\mathbb{E}\left[Z|Y=y\right]$ reaches its maximum/minimum at $y^{*}$ satisfying $\mathbb{V}(Z|Y=y^\ast)=\sigma^2$. 
% $$
% (\gamma^{2}+\sigma^{2})\mathsf{u}^{2}(y^{\ast};\theta)+\gamma^{2}\mathsf{v}^{2}(y^{\ast};\theta)-\sqrt{2/\pi}\sigma\gamma \mathsf{u}(y^{\ast};\theta)=0.
% $$
\end{corollary}
The conditional expectation of $Z$ given
$Y$ is antisymmetric about $\mu$ and, conversely, the conditional variance is symmetric about $\mu$, driven by the fact that it depends on the prediction error strictly through the even function $\mathsf{u}(y)$ and the squared odd function $\mathsf{v}(y)^2$.
The conditional expectation and variance are plotted for $(\mu,\sigma,\gamma)=(0,1,1)$ 
in Figure \ref{fig:Conditional-expectation}. For this case, the maximum conditional expectation is $\max_{y}\mathbb{E}\left[Z|Y=y\right]\approx0.7486$
for $y\approx2.4637$ and the conditional variance ranges between $\mathbb{V}(Z|Y=0)\approx 0.5251$ and $\mathbb{V}(Z|Y=y)\approx 1.1603$ at $y\approx \pm 3.6621$.\footnote{That 
$\mathbb{V}(Z|Y=y)$ can exceed $\mathbb{V}(Z)$ is not inconsistent with the law of total variance, because the latter implies only that
$\mathbb{E}[\mathbb{V}(Z|Y)]\leq \mathbb{V}(Z)$, not that
$\mathbb{V}(Z|Y=y)\leq \mathbb{V}(Z)$ for every $y$. When $\sigma=\gamma=1$ and $\mu=0$, the thresholds for maximum signal extraction ($y\approx2.4637$) and maximum signal confusion ($y\approx3.6621$) are the positive roots of the second and third derivatives of the Voigt log-likelihood, respectively.}
As $y\rightarrow\pm\infty$
the conditional moments converge to the unconditional moments,
$\mathbb{E}\left[Z|Y=y\right]\rightarrow 0$
and $\mathbb{V}\left[Z|Y=y\right]\rightarrow 1$. The reason is that the conditional distribution
converges to the unconditional distribution as $y$ increases in absolute value, i.e. $f_{Z|Y}(z|y)\rightarrow f_{Z}(z)$ (almost everywhere) as $y\rightarrow\pm\infty$. 
The intuition for this is simply that the tail of a Gaussian vanishes much faster than that of a Cauchy, such that a very extreme realization of $y$ is almost fully attributed to the Cauchy component.

Figure \ref{fig:Conditional-expectation} provides valuable insight into the filter we propose later. The shape of (\ref{eq:CondExpect}) is approximately linear near the origin, similar to a standard Kalman filter. However, for large values of $|y|$ the filter enters a nonlinear regime, and the expectation does not merely saturate but redescends toward zero. This non-monotonicity is key and implies that the filter identifies observations that are ``too large to be true signals'' and isolates them as pure noise, protecting the latent trend estimate from contamination.
Since $Y=\mu + Z+X$ we may define 
$\mathbb{E}[X|Y=y]=-\gamma\frac{\mathsf{v}(y)}{\mathsf{u}(y)}$, despite $X$ being Cauchy distributed, with  $\mathbb{E}|X|=\infty$.

The conditional mean in Figure \ref{fig:Conditional-expectation} bears a striking resemblance to the score function of the Student-$t$ distribution in \citet[Figure 1]{HarveyLuati:2014}. This is noteworthy because two objects arise from different constructions. In score-driven filters, the redescending shape is typically induced by specifying a heavy-tailed observation density, such as Student-$t$, whose
score limits the influence of large observations
\citep{HarveyLuati:2014, GasperoniLuatiPaciDInnocenzo:2023}. In contrast, the redescending shape in Figure \ref{fig:Conditional-expectation} emerges endogenously as the exact conditional mean of a latent Gaussian signal observed with Cauchy noise. Thus, the Gauss-Cauchy convolution provides a probabilistic interpretation of redescending update functions of the type used in score-driven models.
\begin{figure}[tbh]
\spacingset{1.15}
\begin{centering}
\includegraphics[width=0.95\textwidth]{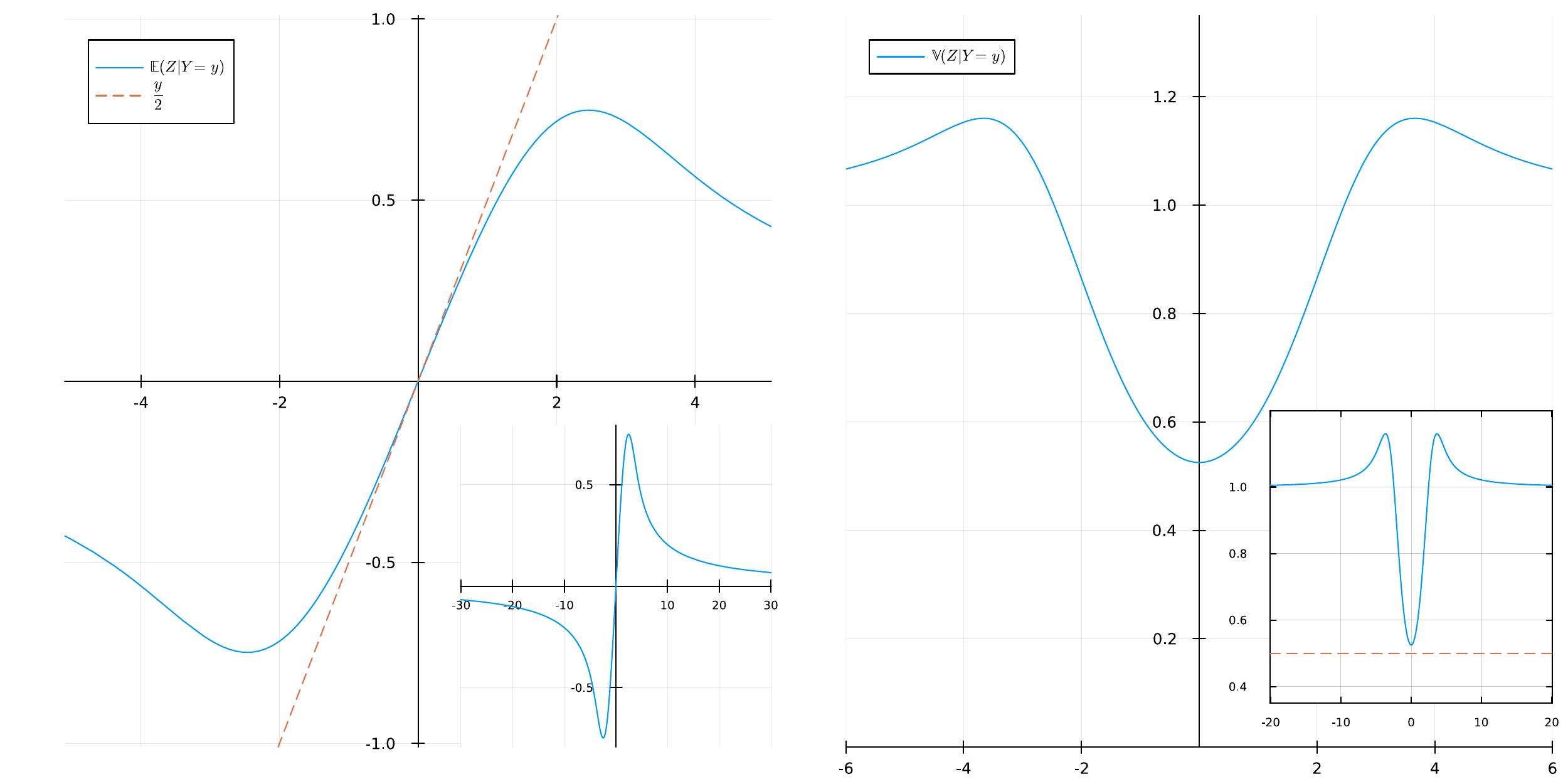}
\par\end{centering}
\caption{{\footnotesize The left panel shows the conditional expectation of $Z$ given $Y=y$, where $Y=Z+X$ with $Z\sim\mathcal N(0,1)$ and $X\sim\operatorname{Cauchy}(0,1)$ independent. The right panel shows the corresponding conditional variance. The insets present the same quantities over a wider range of values of $y$.\label{fig:Conditional-expectation}}}
\end{figure}

The non-monotonicity of the conditional variance is inextricably linked to the redescending nature of the conditional expectation. By differentiating Tweedie's formula for conditional mean, $\mathbb{E}[Z | Y=y] = -\sigma^2 \frac{\partial \log f_Y(y)}{\partial y}$, with respect to $y$, we obtain $\frac{\partial}{\partial y}\mathbb{E}[Z| Y=y] = -\sigma^2 \frac{\partial^2 \log f_Y(y)}{\partial y^2}$. Substituting this into the conditional variance equation yields $\mathbb{V}(Z| Y=y) = \sigma^2 \left( 1 - \frac{\partial}{\partial y} \mathbb{E}[Z| Y=y] \right)$. This establishes three distinct filtering regimes based on the slope of the conditional expectation. Near zero, where $\mathbb{E}(Z|Y=y)$ is increasing in $y$, the observation is informative about the latent Gaussian component, and the conditional variance is strictly less than the unconditional variance $\sigma^2$, reflecting a reduction in uncertainty. Then, at the local extrema of the conditional expectation ($y \approx \pm 2.4637$ for the standard case), the slope is zero, and the conditional variance exactly equals $\sigma^2$. These points mark the transition at which additional signal magnitude no longer increases the inferred Gaussian component. Further into the tails, the expectation redescends and the conditional variance strictly exceeds $\sigma^2$. The increased uncertainty occurs because the filter cannot confidently attribute the moderate outlier to either the Gaussian signal or the Cauchy noise. As $|y| \to \infty$, $\frac{\partial}{\partial y}\mathbb{E}(Z|Y=y) \to 0$ and the conditional variance smoothly asymptotes back to the unconditional variance as the observation is entirely discounted.

More generally, the same Gaussian-convolution identity yields all conditional cumulants of $Z|Y=y$ as derivatives of the Voigt log-density, so conditional skewness, kurtosis, and higher conditional moments of the latent Gaussian component can be evaluated without numerical integration; see Supplement Section \ref{sec:higher-order-cumulants}.

\section{Likelihood-Based Inference for Voigt Distributions}\label{sec:MLE}

Having established the exact density and conditional properties of the Voigt distribution, we now turn to the problem of parameter estimation. While the Voigt distribution is fundamental in spectroscopy and physics, rigorous statistical inference for its parameters has historically been underdeveloped. The closed-form expression in (\ref{eq:Kendall}) allows us to proceed directly with Maximum Likelihood Estimation (MLE), avoiding the approximations commonly used in the existing literature. These results are obtained by leveraging the differential properties of the scaled complementary error function, which leads to exact analytical expressions for the score vector and Hessian matrix.

\begin{lemma}[Score and Hessian of the Voigt log-likelihood]
\label{lem:ScoreHessian}Let $Y\sim\mathcal{V}(\mu,\sigma,\gamma)$,
then the score and Hessian matrix are given by
\begin{align*}
s_{\mu} & =\tfrac{1}{\sigma^2}\left(y-\mu+\gamma\tfrac{v}{u}\right)=\tfrac{1}{\sigma^2}\mathbb{E}(Z|Y=y)\\
s_{\sigma} & =\tfrac{1}{\sigma^{3}u}\left((\tilde{y}^{2}-\gamma^{2}-\sigma^{2})u+2\gamma\tilde{y}v+\sqrt{\tfrac{2}{\pi}}\sigma\gamma\right)\\
s_{\gamma} & =\tfrac{1}{\sigma^{2}u}\left(\gamma u-\tilde{y}v-\sqrt{\tfrac{2}{\pi}}\sigma\right),
\end{align*}
where $\tilde{y}=y-\mu$, $u=\mathsf{u}(y;\mu,\sigma,\gamma)$ and $v=\mathsf{v}(y;\mu,\sigma,\gamma)$ and the elements of the Hessian,
$ H_\theta $
\begin{align*}
H_{\mu\mu} &=  \tfrac{1}{\sigma}s_{\sigma}-s_{\mu}^{2}&\quad&
H_{\mu\sigma}  =-\tfrac{1}{\sigma}\left(s_{\mu}+\gamma H_{\mu\gamma}-\tilde{y}H_{\mu\mu}\right) \\
H_{\gamma\gamma}& =  {-\tfrac{1}{\sigma}s_{\sigma}-s_{\gamma}^{2}}&\quad&
H_{\gamma\sigma} = -\tfrac{1}{\sigma}\left(s_{\gamma}+\gamma H_{\gamma\gamma}-\tilde{y}H_{\mu\gamma}\right)\\
H_{\mu\gamma} &=  \tfrac{1}{\sigma^{2}}\left(\tilde{y}s_{\gamma}+\gamma s_{\mu}+\tfrac{v}{u}\right)-s_{\mu}s_{\gamma}&\quad&
H_{\sigma\sigma} = -\tfrac{1}{\sigma}\left(s_{\sigma}+\gamma H_{\gamma\sigma}-\tilde{y}H_{\mu\sigma}\right).
\end{align*}
\end{lemma}

One significance of Lemma~\ref{lem:ScoreHessian} is computational. The perceived intractability
of the Voigt integral has led researchers to rely on computationally expensive methods or
approximation-based substitutes.\footnote{\citet{CannasPiras:2025} discuss the computational burden of direct numerical integration and pseudo-Voigt approximations in statistical Voigt estimation. Lemma \ref{lem:ScoreHessian} shows that these complications are not intrinsic to likelihood-based inference for the Voigt model: once the density is expressed through $\operatorname{erfcx}$, the score and Hessian follow algebraically and require neither numerical convolution nor finite-difference derivatives.}
However, once the density is represented through the scaled complementary error function,
the likelihood, score, and Hessian are all computed from the same quantities $u(y;\theta)$
and $v(y;\theta)$. Exact maximum likelihood estimation therefore reduces to repeated
evaluation of a standard special function and algebraic operations, without numerical
convolution, finite-difference derivatives, or pseudo-Voigt approximations.

The same chain-rule structure also gives the finite-$\Delta$ likelihood contribution for
a Brownian-Cauchy case of the discretely sampled L\'evy-process model studied by
\citet{AitSahaliaJacod:2008}. Let
$X_t=\sigma W_t+\theta Y_t$, 
where $W$ is Brownian motion and $Y$ is an independent standard symmetric Cauchy
process. Under the Cauchy scale convention used here,
$$
X_{t+\Delta}-X_t\sim \mathcal V (0,\sigma\sqrt{\Delta},|\theta|\Delta).
$$
Thus an increment has a Voigt density with Gaussian scale $\sigma\sqrt{\Delta}$ and
Cauchy scale $|\theta|\Delta$. Its likelihood contribution, score, Hessian, and
Fisher-information entries follow from Lemma~\ref{lem:ScoreHessian} by the chain rule.
For example, when $\theta>0$,
$$
\begin{aligned}
\partial_\sigma\log p_\Delta(x;\sigma,\theta)
&=\sqrt{\Delta}s_\sigma^V(x;0,\sigma\sqrt{\Delta},\theta\Delta),\\
\partial_\theta\log p_\Delta(x;\sigma,\theta)
&=\Delta s_\gamma^V(x;0,\sigma\sqrt{\Delta},\theta\Delta),
\end{aligned}
$$
where $s_\sigma^V$ and $s_\gamma^V$ denote Voigt scale scores. Hence the
Brownian-Cauchy convolution entering finite-$\Delta$ Fisher-information calculations can
be evaluated without numerical convolution or finite-difference derivatives, complementing
the high-frequency limits in \citet{AitSahaliaJacod:2008}.

A natural concern is whether standard likelihood-based inference remains valid
for a model whose distribution has Cauchy tails. Since the Cauchy component dominates the tails of $Y$, the distribution has no finite positive integer moments; in particular, even the mean is undefined. One might therefore expect conventional asymptotic normality to fail. The relevant issue for maximum likelihood estimation, however, is not whether $Y$ has finite moments, but whether the score, Hessian, and Fisher information are well defined. In the Gauss-Cauchy convolution model, they are, so exact maximum likelihood remains a regular estimation problem despite the absence of moments for $Y$. This is analogous to cases that arise in GARCH models, where asymptotic normality can hold even when the observed process has limited or no finite moments; see \citet{JensenRahbek:2004}. Theorem \ref{thm:MLE-consistent-asN} establishes consistency and asymptotic normality of the MLE for the Voigt model.

\begin{assumption}[Compactness]\label{assu:Compact}
For $\theta=(\mu,\sigma,\gamma)^\prime$, the parameter set $\theta\in\Theta$ is compact
$$
\Theta
=
\left\{
(\mu,\sigma,\gamma)\in\mathbb{R}^{3}:
|\mu|\leq\mu_{\max},
0<\sigma_{\min}\leq\sigma\leq\sigma_{\max},
0<\gamma_{\min}\leq\gamma\leq\gamma_{\max}
\right\}.
$$
\end{assumption}

Assumption \ref{assu:Compact} restricts attention to the regular Voigt case in which both scale parameters are strictly positive. This excludes the boundary cases $\sigma=0$ and $\gamma=0$, corresponding to the pure Cauchy and pure Gaussian limits, respectively. These limiting cases are mathematically meaningful, but they require separate treatment because the analytical expressions below involve ratios and derivatives evaluated under $\sigma>0$ and $\gamma>0$.

The lower bound on $\gamma$ is also useful for the uniform integrability condition used in the consistency proof. If the true distribution has Cauchy tails but the parameter space includes the pure Gaussian boundary $\gamma=0$, the Gaussian tail can make $|\log f_Y(Y;\theta)|$ grow quadratically in $Y$ for some admissible parameter values, while $Y$ itself has Cauchy tails. Bounding $\gamma$ away from zero avoids this difficulty and yields a common logarithmic envelope for the likelihood. The lower bound on $\sigma$ similarly keeps the analysis within the interior of the regular Voigt family, where the score and Hessian formulas are ordinary derivatives with respect to $(\mu,\sigma,\gamma)$.

\begin{theorem}[Consistency and Asymptotic Normality of the Voigt MLE]
\label{thm:MLE-consistent-asN}
Let $Y_1,\ldots,Y_n$ be iid with
$\mathcal{V}(\mu_0,\sigma_0,\gamma_0)$, and let
$\theta_0=(\mu_0,\sigma_0,\gamma_0)^\prime\in\Theta$, where $\Theta$ is the
compact parameter space defined in Assumption \ref{assu:Compact}. Define the
maximum likelihood estimator by $
\hat\theta_n = \arg\max_{\theta\in\Theta}\ell_n(\theta)$, $\ell_n(\theta)=
\frac{1}{n}\sum_{i=1}^n\log f_Y(Y_i;\theta)$. 
Then $\hat\theta_n\xrightarrow{p}\theta_0$. 

If, in addition, $\theta_0$ is an interior point of $\Theta$, then the
information matrix equality holds: $
\mathcal{I}_{\theta_0} = \mathbb{E}_{\theta_0}\left[s(Y;\theta_0)s(Y;\theta_0)^\prime\right]
= -\mathbb{E}_{\theta_0}\left[H(Y;\theta_0)\right]$, 
where  $s(y;\theta)=
\frac{\partial\log f_Y(y;\theta)}{\partial\theta}$, $H(y;\theta)
= \frac{\partial^2\log f_Y(y;\theta)}
{\partial\theta\partial\theta^\prime}$. Moreover, $\mathcal{I}_{\theta_0}$ is nonsingular and
$\sqrt{n}(\hat\theta_n-\theta_0)
\xrightarrow{d}
\mathcal{N}(0,\mathcal{I}_{\theta_0}^{-1})$. Finally, by symmetry of the Voigt density, $\mathcal{I}_{\theta_0}$ is
block diagonal with
$\mathcal{I}_{\mu_0\sigma_0}=\mathcal{I}_{\mu_0\gamma_0}=0$.
\end{theorem}

Despite the lack of moments, $\mathbb{E}|Y|=\infty$,  
Theorem \ref{thm:MLE-consistent-asN} shows that likelihood-based inference is quite standard in the Gauss-Cauchy convolution model. Convergence in distribution is at the conventional root-$n$ rate and the limit distribution is Gaussian, such that standard errors and $t$-statistics have their usual interpretation. The underlying reason is that the likelihood function and its derivatives are well-behaved. Moreover, the information-matrix equality holds, and the simplest way to estimate $\mathcal{I}_{\theta_{0}}$ appears to be to evaluate $\mathbb{E}_{\theta_{0}}[s_{\theta_{0}}s_{\theta_{0}}^{\prime}]$ by numerical integration. 
The asymptotic independence of $\hat{\mu}$ and $(\hat{\sigma},\hat{\gamma})$ is a consequence of $s_\mu$ being an odd function of the $\tilde{y}=y-\mu$, whereas  $s_\gamma$ and $s_\sigma$ are both even functions. Thus, ($\sigma,\gamma$) will have the same  asymptotic variances in the simpler model where $\mu$ is known.

The Fisher-information matrix also implies a scale-invariant comparison of the asymptotic precision of the Gaussian and Cauchy scale estimates. Because the asymptotic covariance matrix for $(\sigma,\gamma)$ is homogeneous of degree two, the ratio of their asymptotic standard deviations depends only on $\lambda=\gamma/\sigma$. Supplement~\ref{app:relative_precision} reports this ratio and shows that, over the range considered, the Cauchy scale is asymptotically more precisely estimated than the Gaussian scale. This reflects the fact that $\gamma$ is strongly identified by extreme observations, whereas $\sigma$ is identified primarily from the central part of the distribution, where the two scale parameters can partially offset one another.

\subsection{Finite-Sample Behavior of the Voigt MLE}

We examine the finite-sample behavior of the Voigt maximum likelihood estimator through a Monte Carlo study. The purpose is to assess how accurately the asymptotic distribution in Theorem \ref{thm:MLE-consistent-asN} approximates the sampling distribution of the MLE. The simulations use $N=100,000$ replications, sample sizes from $n=100$ to $n=10,000$, and three values of the Cauchy-to-Gaussian scale ratio, $\lambda=\gamma_0/\sigma_0\in\{0.01,0.1,1\}$.

The detailed results are reported in Supplement~\ref{app:voigt_mle_mc}. Across the three designs, the MLE is well centered and the empirical standard deviations are close to the Fisher-information standard deviations. The main finite-sample distortions occur for the scale parameters in small samples, especially when $\gamma_0$ is close to the boundary or when $\sigma_0$ and $\gamma_0$ are difficult to separate. These distortions diminish rapidly as $n$ increases.

\section{The Gauss-Cauchy Convolution (GCC) Filter}\label{sec:GCCfilter}

We introduce the Gauss-Cauchy Convolution (GCC) filter, a Masreliez-type approximate filter for extracting a latent state observed through Voigt measurement error. Built on the Voigt prediction-error density, the filter updates the state using the conditional-mean correction implied by the Gaussian-Cauchy convolution. The resulting update is approximately linear for moderate prediction errors but redescends toward zero for large errors. Thus, its robustness is operational: extreme observations have small update influence and are attributed primarily to the Cauchy measurement component rather than to persistent latent state movements.

A major advantage of the GCC filter is analytical tractability. Conditional on a Gaussian prediction density, the prediction-error density is Voigt and available in closed form. This contrasts with Student-$t$ measurement-error models, which are not closed under convolution with Gaussian state uncertainty and therefore require an additional approximation to obtain a tractable prediction-error density. The GCC filter nests the pure Cauchy filter of \citet{McCabeGualdoni:2024}, extending it to accommodate Gaussian background noise. Finally, filter parameters are estimated by quasi-maximum likelihood using analytical derivatives, rather than by selecting tuning constants.

\subsection{Filtering Setup}\label{sec:FilteringSetup}
Let $x_t$ denote a latent state variable that follows a Gaussian AR(1) process, and $y_t$ denote the observed variable, measured with additive error:
\begin{alignat*}{2}
x_t &= (1-\phi)\mu+\phi x_{t-1}+\varepsilon_t,
&\qquad \varepsilon_t &\sim \mathcal{N}(0,\tau^2), \\
y_t &= x_t+\eta_t,
&\qquad \eta_t &\sim \mathcal{V}(0,\sigma,\gamma).
\end{alignat*}
The state innovations, Gaussian measurement-error components, and Cauchy
measurement-error components are mutually independent and independent over
time. The initial state is independent of all subsequent innovations.
Thus the latent state is Gaussian, while the measurement
error follows a Voigt distribution. Equivalently, the measurement error contains
a Gaussian component with scale $\sigma$ and a Cauchy component with scale
$\gamma$.

We use standard state-space notation \citep[see, e.g.,][]{Harvey:1989, DurbinKoopman:2012},
with lowercase letters for the state and observation variables to avoid
conflict with the convolution variables used above. Let
$\mathcal{F}_t=\sigma(y_t,y_{t-1},\ldots)$ denote the natural filtration, and
$$
x_{t|t-1}
=
\mathbb{E}[x_t|\mathcal{F}_{t-1}],
\quad\text{and}\quad 
h_{t|t-1}
=
\operatorname{var}(x_t|\mathcal{F}_{t-1}).
$$
Here $h_{t|t-1}$ denotes the one-step-ahead prediction variance of the latent
state. The linear Gaussian state transition gives the usual prediction
equations:
$$
x_{t|t-1}  = (1-\phi)\mu+\phi x_{t-1|t-1}\quad\text{and}\quad 
h_{t|t-1} = \phi^2 h_{t-1|t-1}+\tau^2.
$$

The nonlinear part of the GCC filter enters through the update step. Following
the approximation introduced by \citet{Masreliez:1975}, we approximate the
one-step-ahead prediction density of the state by a Gaussian distribution.

\begin{assumption}[Masreliez approximation]\label{ass:Masreliez}
Conditionally on $\mathcal{F}_{t-1}$, 
$$\xi_t=x_t-x_{t|t-1}\sim\mathcal{N}(0,h_{t|t-1}).$$
\end{assumption}
This approximation was originally introduced by \citet{Masreliez:1975} for
robust radar tracking in the presence of so-called glint noise, a form of
sporadic heavy-tailed measurement contamination. In the present setting, the
approximation is useful for a more specific reason: it turns the prediction
error, $e_t$, into a Gaussian convolution. 
$$
e_t=y_t-x_{t|t-1}=(y_t-x_t)+(x_t-x_{t|t-1})=\eta_t+\xi_t,
$$
Under Assumption \ref{ass:Masreliez}, conditionally on $\mathcal F_{t-1}$,
$\xi_t$ is Gaussian and independent of $\eta_t$. Proposition
\ref{prop:TweedieGaussian} can therefore be applied conditionally to obtain the
state correction from the score of the one-step-ahead prediction-error density: $\mathbb{E}[\xi_t|e_t,\mathcal F_{t-1}]
=-h_{t|t-1}\frac{\partial}{\partial e_t}\log f_{t-1}(e_t)$, where $f_{t-1}$ denotes the conditional density of $e_t$ given
$\mathcal F_{t-1}$.

The second ingredient is specific to the GCC model. Since
$\eta_t\sim\mathcal V(0,\sigma,\gamma)$, we may write
$\eta_t=Z_t+C_t$, where $Z_t\sim\mathcal N(0,\sigma^2)$ and
$C_t\sim\operatorname{Cauchy}(0,\gamma)$. Hence
$ e_t=(\xi_t+Z_t)+C_t$, $\xi_t+Z_t|\mathcal F_{t-1}\sim\mathcal N(0,h_{t|t-1}+\sigma^2)$,  so the prediction-error density remains Voigt, specifically
$e_t|\mathcal F_{t-1}\sim\mathcal V(0,\delta_t,\gamma)$ with $\delta_t^2=h_{t|t-1}+\sigma^2$.  Thus Assumption \ref{ass:Masreliez} and the Gaussian-Cauchy convolution
structure work together: Tweedie's formula gives the moment update, and the
Voigt closure property makes the required prediction-error score available in
closed form.

\begin{theorem}[GCC filter]\label{thm:GCCfilter}
Under Assumption \ref{ass:Masreliez}, conditionally on $\mathcal{F}_{t-1}$, the prediction error
$e_t=y_t-x_{t|t-1}=x_t-x_{t|t-1}+\eta_t$ satisfies
$$
e_t|\mathcal{F}_{t-1}\sim\mathcal{V}(0,\delta_t,\gamma),
\qquad
\delta_t^2=h_{t|t-1}+\sigma^2.
$$
The conditional mean and variance,
$x_{t|t}=\mathbb{E}[x_t | \mathcal{F}_t]$ and
$h_{t|t}=\operatorname{var}(x_t | \mathcal{F}_t)$, are updated by
$$
x_{t|t} = x_{t|t-1}+h_{t|t-1}\psi_t
\quad\text{and}\quad 
h_{t|t} = h_{t|t-1}-h_{t|t-1}^2\psi_t^\prime,
$$
where
$$
\psi_t = \frac{1}{\delta_t^2}
\left( e_t+\gamma\frac{\mathsf{v}_t}{\mathsf{u}_t}
\right)
\quad\text{and}\quad
\psi_t^\prime
= \frac{1}{\delta_t^4} \left[\delta_t^2+ \gamma^2
\left(1+\frac{\mathsf{v}_t^2}{\mathsf{u}_t^2}
\right)
-\sqrt{\frac{2}{\pi}}\frac{\gamma\delta_t}{\mathsf{u}_t}\right],
$$
with $\psi_t^\prime=\partial\psi_t/\partial e_t$,
$\mathsf{u}_t=\mathsf{u}(e_t;0,\delta_t,\gamma)$, and
$\mathsf{v}_t=\mathsf{v}(e_t;0,\delta_t,\gamma)$.
\end{theorem}
The update has the same structure as the Kalman filter, except that the
linear Gaussian score $e_t/\delta_t^2$ is replaced by the Voigt location score
$\psi_t$. For moderate prediction errors, the update is approximately linear
and resembles the Kalman update. For large prediction errors, however, the
Voigt score redescends toward zero, so the filter discounts observations that
are more plausibly attributed to the Cauchy component than to latent
Gaussian state.

Figure \ref{fig:score_functions} illustrates this mechanism. The left panel
shows how the Voigt score changes with the Cauchy-to-Gaussian ratio
$\lambda=\gamma/\sigma$. Larger values of $\lambda$ generate stronger
redescending behavior. The right panel compares the GCC score with Gaussian,
Huber, and Student-$t$ scores. The Gaussian score is linear, the Huber score is
bounded, and the Student-$t$ score is redescending. The GCC score is also
redescending, but its shape is tied directly to the Gaussian-Cauchy convolution
rather than imposed as a robustification device.

\begin{figure}[t]
\spacingset{1.15}
\begin{centering}
\includegraphics[width=1\textwidth]{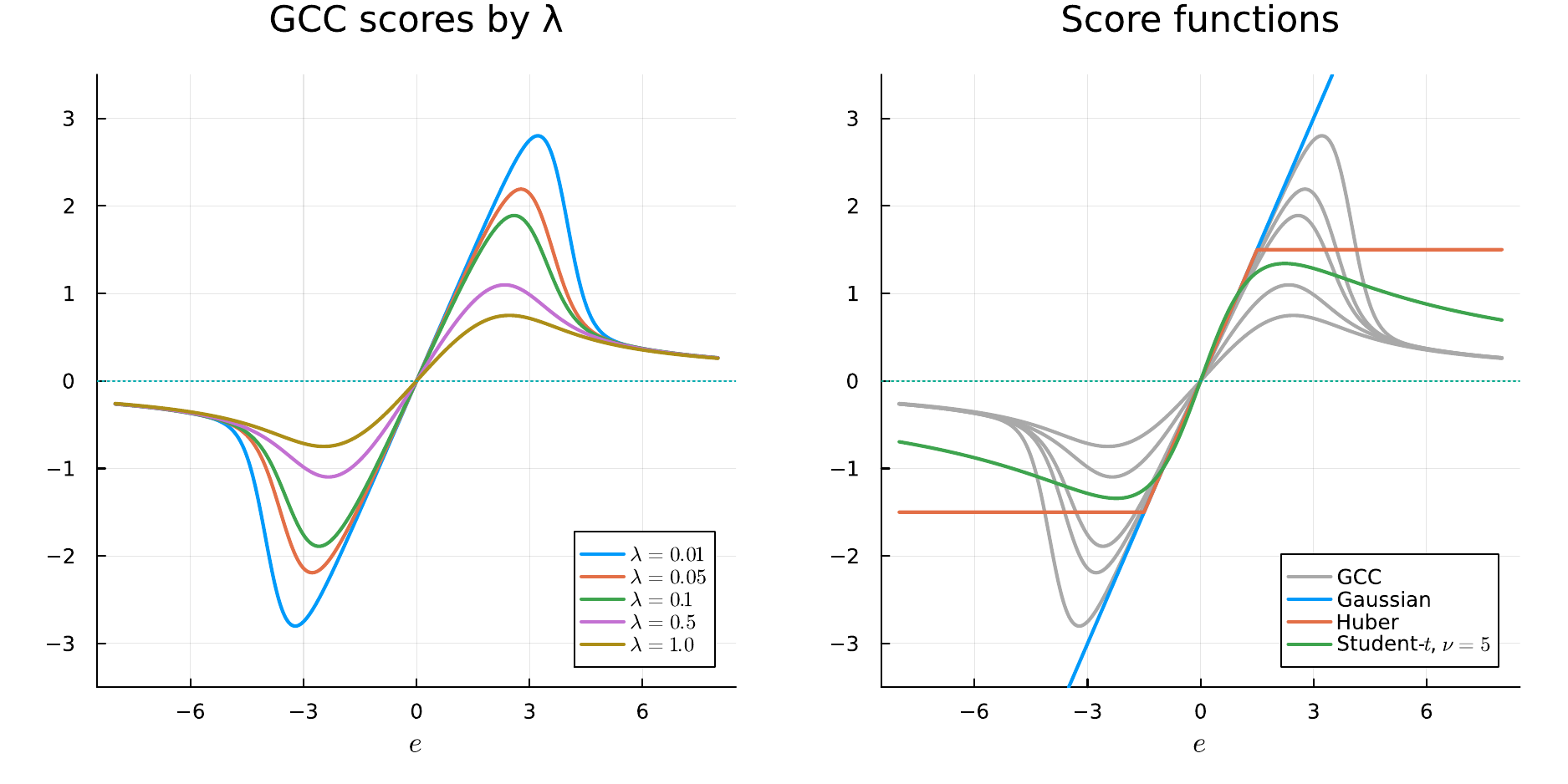}
\par \end{centering}
\caption{\footnotesize Score functions. The left panel shows the GCC location score, $\psi(e)$,
for different values of $\lambda=\gamma/\sigma$. The right panel compares the
GCC score with Gaussian, Huber, and Student-$t$ scores. The GCC score is
approximately linear near the origin and redescends in the tails, implying that
extreme prediction errors receive little weight in the state update.}
\label{fig:score_functions}
\end{figure}

Under the Gaussian prediction approximation in Assumption \ref{ass:Masreliez},
the recursions in Theorem \ref{thm:GCCfilter} coincide with the approximate
conditional mean filter of \citet{Masreliez:1975}. The update is therefore
optimal within the class of filters that represent the one-step-ahead state
density by a Gaussian approximation and update its first two conditional
moments using the prediction-error density. The next lemma records that the variance recursion remains positive under the Gaussian prediction approximation.
\begin{lemma}[Positivity of the GCC variance recursion]\label{lem:GCCvariance}
Suppose Assumption \ref{ass:Masreliez} holds and $h_{t|t-1}>0$. Then the GCC
variance update satisfies $h_{t|t}=h_{t|t-1}-h_{t|t-1}^2\psi_t^\prime>0$.
Moreover,
$h_{t+1|t} = \phi^2h_{t|t}+\tau^2 \geq \tau^2$, 
with strict inequality whenever $\phi\neq0$. Consequently,
$\delta_{t+1}^2=h_{t+1|t}+\sigma^2>0$, so the Voigt prediction-error density is always well defined.
\end{lemma}

\subsection{GCC Smoothing Algorithm}

The GCC filter produces filtered estimates $x_{t|t}$ and $h_{t|t}$. For
applications in which a full-sample estimate of the latent state is desired, we
construct smoothed estimates using a fixed-interval recursion (the Rauch-Tung-Striebel smoother) based on the Gaussian state transition and the moment approximation that underlies the GCC
filter. Starting from $x_{T|T}$ and $h_{T|T}$, the smoothed estimates are obtained
backward by $x_{t|T}=x_{t|t}+c_t(
x_{t+1|T}-x_{t+1|t})$, $c_t={\phi h_{t|t}}/{h_{t+1|t}}$ for $t=T-1,\ldots,1$. The smoothed variance is $h_{t|T}=h_{t|t}+c_t^2(h_{t+1|T}-h_{t+1|t})$. Here $x_{t|t}$, $h_{t|t}$, $x_{t+1|t}$, and $h_{t+1|t}$ are produced by the forward GCC filter.

\subsection{Filtering of the Cauchy Components}

As established in Section \ref{sec:FilteringSetup}, the total Gaussian component of the prediction error, $e_t = \xi_t + Z_t + C_t$, is $\xi_t + Z_t |\mathcal{F}_{t-1} \sim \mathcal{N}(0, \delta_t^2)$. Applying Corollary \ref{cor:CondExpect}, its conditional expectation given $\mathcal{F}_{t}$ is given by $\mathbb{E}[\xi_t + Z_t  |  \mathcal{F}_t] = \delta_t^2 \psi(e_t; 0, \delta_t, \gamma) = e_t + \gamma {\mathsf{v}(e_t; 0, \delta_t, \gamma)}/{\mathsf{u}(e_t; 0, \delta_t, \gamma)}$. The filtered Cauchy component is then isolated residually:
\begin{equation}
\mathbb{E}[C_t  |  \mathcal{F}_t] = e_t - \mathbb{E}[\xi_t + Z_t  |  \mathcal{F}_t] = -\gamma \frac{\mathsf{v}(e_t; 0, \delta_t, \gamma)}{\mathsf{u}(e_t; 0, \delta_t, \gamma)}.
\label{eq:filtered_cauchy}
\end{equation}
The total Gaussian component is subsequently allocated between the state update and the measurement noise in proportion to their prior variances:
\begin{equation*}
\mathbb{E}[\xi_t  |  \mathcal{F}_t] = h_{t|t-1}\psi(e_t; 0, \delta_t, \gamma), \qquad \mathbb{E}[Z_t  |  \mathcal{F}_t] = \sigma^2\psi(e_t; 0, \delta_t, \gamma).
\end{equation*}

\subsection{QMLE of the Gauss-Cauchy Convolution Filter}

The parameters are estimated by maximizing the quasi log-likelihood implied by
the Masreliez prediction approximation. Let
$\vartheta=(\mu,\sigma,\gamma,\phi,\tau)^\prime$ and define, recursively,
$e_t(\vartheta)=y_t-x_{t|t-1}(\vartheta)$ and
$\delta_t^2(\vartheta)=h_{t|t-1}(\vartheta)+\sigma^2$. The one-step
quasi log-likelihood contribution is
$$
\ell_t(\vartheta) = -\tfrac{1}{2}\log(2\pi)
-\log\delta_t(\vartheta)
+ \log\mathsf{u}(e_t(\vartheta);0,\delta_t(\vartheta),\gamma),
$$
and the full-sample objective is
$\ell_T(\vartheta)=\sum_{t=1}^T\ell_t(\vartheta)$.

We estimate $\vartheta$ over a compact parameter space $\Theta$ satisfying
$|\mu|\leq\mu_{\max}$,
$0<\sigma_{\min}\leq\sigma\leq\sigma_{\max}$,
$0<\gamma_{\min}\leq\gamma\leq\gamma_{\max}$,
$|\phi|\leq\phi_{\max}<1$, and
$0<\tau_{\min}\leq\tau\leq\tau_{\max}$. The QMLE is
$\hat\vartheta_T=\operatorname{argmax}_{\vartheta\in\Theta}\ell_T(\vartheta)$.
The recursion is initialized at the stationary moments,
$x_{1|0}=\mu$ and $h_{1|0}=\tau^2/(1-\phi^2)$.

The objective is a quasi-likelihood because the prediction-error density is
evaluated under the Gaussian prediction approximation in Assumption
\ref{ass:Masreliez}. Conditional on the recursively generated prediction
moments, each likelihood contribution is the exact Voigt density of the
prediction error under the approximating Gaussian prediction law. Thus the
criterion is exact for the approximating recursion, but not the exact likelihood
of the non-Gaussian state-space model. The analytical derivatives in Lemma
\ref{lem:ScoreHessian} and the algebraic closure property in Corollary
\ref{cor:VoigtAlgebraicClosure} make it possible to compute the objective and
its derivatives without numerical convolution or finite differences.

Let
$$
Q(\vartheta)=\lim_{T\to\infty}T^{-1}
\sum_{t=1}^T\mathbb{E}_{\vartheta_0}[\ell_t(\vartheta)]
$$
when this limit exists, and let
$\vartheta_\star=\operatorname{argmax}_{\vartheta\in\Theta}Q(\vartheta)$
denote the pseudo-true parameter, when this maximizer is unique. We use the
standard recursive M-estimation expansion to guide inference. Specifically, if
the effect of initialization is asymptotically negligible, the score satisfies a
central limit theorem, and the Hessian converges to a nonsingular limit at
$\vartheta_\star$, then
$\sqrt{T}(\hat{\vartheta}_T-\vartheta_\star)
\xrightarrow{d}
\mathcal{N}\left(0,\mathcal{J}_{\vartheta_\star}^{-1}
\mathcal{I}_{\vartheta_\star}
\mathcal{J}_{\vartheta_\star}^{-1}\right)$, where
$$
\mathcal{I}_{\vartheta_\star}
=
\lim_{T\to\infty}
\operatorname{var}\left(
T^{-1/2}\sum_{t=1}^T
\frac{\partial\ell_t(\vartheta_\star)}{\partial\vartheta}
\right),
\qquad
\mathcal{J}_{\vartheta_\star}
=
-\operatorname{plim}_{T\to\infty}
T^{-1}\sum_{t=1}^T
\frac{\partial^2\ell_t(\vartheta_\star)}
{\partial\vartheta\partial\vartheta^\prime}.
$$

If the recursively generated Gaussian prediction law coincides with the true conditional prediction law, then the quasi-likelihood is correctly specified,
$\vartheta_\star=\vartheta_0$, and $\mathcal{I}_{\vartheta_0}=\mathcal{J}_{\vartheta_0}$. For the exact Gauss-Cauchy state-space model, however, the filtering distribution is not Gaussian after conditioning on past observations, so these recursive regularity conditions are not proved here. The Monte Carlo evidence below therefore evaluates whether the inverse-information approximation provides an accurate finite-sample description of the QMLE in the designs considered.

\subsection{QMLE Monte Carlo}
We assess the finite-sample behavior of the GCC QMLE in two designs: a small-Cauchy design with $\gamma_0/\sigma_0=0.1$ and a balanced design with $\gamma_0/\sigma_0=1$. The detailed Monte Carlo results are reported in Supplement~\ref{app:gcc_qmle_mc}. Across both designs, the QMLE is well centered and inverse-information standard errors closely approximate sampling variability. The main finite-sample distortions are a small downward bias in the persistence parameter and, in the balanced design, a short-sample downward bias in $\hat\sigma$. These discrepancies shrink with sample size, indicating that the pseudo-true parameter is close to the true parameter and that the information loss induced by the Masreliez approximation is small in the designs considered.

\section{Assessing Masreliez Approximation in GCC Filter}
\label{sec:Simulations}

The GCC filter relies on Assumption \ref{ass:Masreliez}, where the Gaussian approximation to the predictive state density enables a Tweedie moment correction, and the measurement-error structure keeps the prediction-error density Voigt. Since the latter step is exact given the approximation, this section assesses the numerical cost of the former by holding parameters fixed at their true values and evaluating the consequences of replacing the exact latent state predictive distribution with the GCC Gaussian approximation.
\begin{table}[htbp]
\small
\caption{Kullback-Leibler diagnostics for the Gaussian prediction approximation}
\label{tab:gcc_density_diagnostics_main}
\spacingset{1.15}
\begin{centering}
\begin{small}
\begin{tabularx}{\textwidth}{@{}YYYYYY @{}}
\toprule
    \midrule
$\lambda$ & $\overline{\mathrm{KL}}^{x,\operatorname{shape}}$ & $\overline{\mathrm{KL}}^{x,\operatorname{op}}$ & $\overline{\mathrm{KL}}^{y,\operatorname{shape}}$ & $\overline{\mathrm{KL}}^{y,\operatorname{op}}$ & $\max\mathrm{KL}^{x,\operatorname{op}}$ \\
\midrule
$0.00$ & $<10^{-15}$          & $<10^{-15}$          & $<10^{-15}$          & $<10^{-15}$          & $<10^{-15}$          \\
$0.01$ & $7.05\times 10^{-5}$ & $9.10\times 10^{-5}$ & $1.34\times 10^{-5}$ & $2.17\times 10^{-5}$ & $2.78\times 10^{-2}$ \\
$0.05$ & $1.78\times 10^{-4}$ & $2.33\times 10^{-4}$ & $2.82\times 10^{-5}$ & $4.76\times 10^{-5}$ & $2.13\times 10^{-2}$ \\
$0.10$ & $2.94\times 10^{-4}$ & $4.11\times 10^{-4}$ & $3.93\times 10^{-5}$ & $8.38\times 10^{-5}$ & $3.50\times 10^{-2}$ \\
$0.50$ & $7.30\times 10^{-4}$ & $1.03\times 10^{-3}$ & $4.57\times 10^{-5}$ & $1.31\times 10^{-4}$ & $6.47\times 10^{-2}$ \\
$1.00$ & $8.44\times 10^{-4}$ & $1.24\times 10^{-3}$ & $2.88\times 10^{-5}$ & $1.06\times 10^{-4}$ & $4.11\times 10^{-2}$ \\
\midrule
\bottomrule
\end{tabularx}
\end{small}
\par\end{centering}
{\footnotesize \textit{Note:} Density-level Kullback-Leibler diagnostics for the Masreliez-GCC approximation, averaged over the nine designs with $\phi\in\{0.90,0.97,0.99\}$ and $\tau/\sigma\in\{0.25,0.50,1.00\}$ for each $\lambda=\gamma/\sigma$. Superscripts $x$ and $y$ denote latent-state and observation predictive densities. ``Shape'' uses the moment-matched Gaussian approximation; ``op'' uses the operational GCC approximation. The final column is the maximum $\mathrm{KL}^{x,\operatorname{op}}$ over all dates and designs in the row.}
\end{table}

We evaluate this against an exact benchmark filter from numerical density propagation by \citet{Kitagawa:1987}, which updates the predictive density via Bayes' rule without imposing Gaussianity (details in Supplement~\ref{app:simulation-details}). This exact benchmark is compared with two Gaussian approximations: a moment-matched version using exact predictive moments (isolating the shape error from Gaussianizing the exact density) and the operational GCC approximation using recursively generated moments (measuring the full implementation error). The approximation is evaluated by computing Kullback-Leibler (KL) discrepancies for the latent-state and observation predictive densities, and by comparing the exact versus GCC one-step corrections. By scale equivariance, the simulation design uses dimensionless ratios $\lambda=\gamma/\sigma$ (including the empirically relevant $\lambda=0.10$) and $\tau/\sigma$ alongside the persistence parameter $\phi$. 

Density-level diagnostics (Table \ref{tab:gcc_density_diagnostics_main}) reveal high accuracy: for $\lambda \leq 0.10$, average KL discrepancies remain on the order of $10^{-4}$ or smaller for the latent state, and an order of magnitude smaller for the observation density due to Voigt convolution smoothing. Correspondingly, Table \ref{tab:gcc_correction_diagnostics_main} shows these differences have minimal impact on the filter update. For $\lambda \leq 0.10$, mean absolute operational distortions range from $10^{-3}$ to $5 \times 10^{-3}$, with root mean square errors below $1.5 \times 10^{-2}$. Operational diagnostics only modestly exceed shape diagnostics, indicating the GCC recursion does not materially amplify approximation errors. While larger discrepancies arise in stress designs ($\lambda=0.50$ or 1.00), distortions remain moderate, suggesting violations of Assumption~\ref{ass:Masreliez} have limited practical consequences for the GCC filter, particularly when the Cauchy component is small relative to the Gaussian noise.

\begin{table}[htbp]
\caption{Correction distortion by Cauchy-to-Gaussian measurement-noise ratio}
\label{tab:gcc_correction_diagnostics_main}
\spacingset{1.15}
\begin{centering}
\begin{small}
\begin{tabularx}{\textwidth}{@{}YYYYY@{}}
\toprule
    \midrule
$\lambda$
& $\operatorname{MAE}^{\operatorname{shape}}$
& $\operatorname{MAE}^{\operatorname{op}}$
& $\operatorname{RMSE}^{\operatorname{op}}$
& $q_{.95}|D^{\operatorname{op}}|$ \\
\midrule
$0.00$ & $<10^{-15}$          & $<10^{-15}$          & $<10^{-15}$          & $<10^{-15}$          \\
$0.01$ & $6.53\times 10^{-4}$ & $1.10\times 10^{-3}$ & $4.72\times 10^{-3}$ & $3.83\times 10^{-3}$ \\
$0.05$ & $1.84\times 10^{-3}$ & $2.80\times 10^{-3}$ & $7.71\times 10^{-3}$ & $1.06\times 10^{-2}$ \\
$0.10$ & $3.41\times 10^{-3}$ & $4.94\times 10^{-3}$ & $1.28\times 10^{-2}$ & $1.96\times 10^{-2}$ \\
$0.50$ & $8.24\times 10^{-3}$ & $1.14\times 10^{-2}$ & $2.10\times 10^{-2}$ & $3.81\times 10^{-2}$ \\
$1.00$ & $9.18\times 10^{-3}$ & $1.32\times 10^{-2}$ & $2.17\times 10^{-2}$ & $4.34\times 10^{-2}$ \\
\midrule
\bottomrule
\end{tabularx}
\end{small}
\par\end{centering}
{\footnotesize \textit{Note:} Correction-level diagnostics for the Masreliez-GCC approximation, averaged over the nine designs with $\phi\in\{0.90,0.97,0.99\}$ and $\tau/\sigma\in\{0.25,0.50,1.00\}$ for each $\lambda=\gamma/\sigma$. The distortion is $D_t(y_t)=\Delta_t^\ast(y_t)-\Delta_t(y_t)$, where $\Delta_t^\ast$ is the exact correction and $\Delta_t$ is the operational GCC correction. The shape diagnostic uses the moment-matched Gaussian approximation.}
\end{table}
%%%%%%%%%%%%%%%%%%%%%%%%%%%%%%%%%%%%%%%%%%%%%%%%%%%%%%%%
\section{Empirical Application}\label{sec:Empirical}

We now apply the GCC filter to realized volatility data. The objective is to
evaluate whether the Gauss-Cauchy measurement-error specification provides a
useful empirical decomposition of realized volatility into a persistent latent
component and transitory measurement noise. We compare the GCC filter with
several alternative filters that use the same Gaussian AR(1) state equation but
differ in their measurement-error distributions.

\subsection{Data}

We apply the GCC filter to the daily logarithm of realized volatility for the
Technology Select Sector SPDR Fund (XLK). The sample runs from December 22,
1998 to November 7, 2025, and the realized-volatility series is obtained from
the Dacheng Xiu Risk Lab. Realized measures are constructed from high-frequency
intraday returns and are therefore informative about daily variation in
volatility, but they may also contain extreme observations associated with
market stress, liquidity disruptions, and market microstructure effects.

We work with log realized volatility. This transformation reduces the
right-skewness of realized volatility and makes a symmetric measurement-error
specification more plausible. In the XLK sample, the skewness of the log
realized-volatility series is 0.084.

\subsection{Competing filters}

We benchmark the GCC filter against alternatives that retain the same Gaussian AR(1) state equation but differ in the measurement-error distribution. The comparison includes the Gaussian Kalman filter, the pure Cauchy filter, the Normal-Laplace filter, and two standard robust benchmarks based on Student-$t$ and Huber measurement errors. The Gaussian and Cauchy specifications are limiting cases of the GCC construction, while the Normal-Laplace model provides a light-tailed robust alternative whose Gaussian convolution also has a closed form.

A key distinction is whether the prediction-error density remains analytically tractable after convolution with Gaussian state uncertainty. This is the case for the Gaussian, Cauchy, GCC, and Normal-Laplace specifications, but not for the Student-$t$ and Huber filters. The latter are therefore implemented as operational filtering rules based on same-family prediction-error approximations. The full specification table and additional parameterization details are given in Supplement~\ref{app:filter_details}.

\subsection{Estimation results}

Table \ref{tab:DFest-block} reports estimation results for the competing filters. The GCC specification attains the largest value of the implemented prediction-error criterion, $\ell=-1,306$, compared with
$\ell=-2,465$ for the Gaussian benchmark. This large difference indicates that a purely Gaussian measurement-error specification is too restrictive for log
realized volatility.

Allowing for heavy tails improves the criterion across all alternatives. The Student-$t$ filter reaches $\ell=-1,349$ ($\nu\approx 5.25$), while the pure Cauchy yields $\ell=-1,536$. Because the Student-$t$ and Huber filters rely on same-family prediction-error density approximations, their reported pseudo-log likelihoods strictly reflect a comparison of operational filtering rules. With this qualification, the implemented criteria nevertheless place the GCC
filter at the top of the empirical ranking.

The Normal-Laplace estimate is also informative: its Gaussian scale is essentially zero, so the fitted specification is close to a pure-Laplace measurement-error model. Thus, within the Normal-Laplace family, the data do not support a separately identified Gaussian measurement-error component once the Laplace component is included.

The GCC estimates decompose measurement error into a dominant Gaussian component and a small but non-negligible Cauchy component ($\sigma=0.1817$, $\gamma=0.0199$, $\lambda=\gamma/\sigma\approx 0.11$). As Section \ref{sec:Simulations} demonstrates, the Gaussian prediction approximation is highly accurate in this range. Thus, the fit improves not by replacing Gaussian noise with a purely heavy-tailed specification, but by retaining a Gaussian core that accommodates occasional large deviations. In contrast, the Gaussian filter absorbs these extremes by inflating the measurement variance ($\sigma=0.2980$), yielding a noisier extraction of the latent state.

Beyond empirical fit, the GCC also has an analytical advantage.
As indicated in Table \ref{tab:DFest-block}, the prediction-error density is available in closed form for the GCC, Gaussian, Cauchy, and Normal-Laplace models under the Masreliez approximation, which allows the likelihood contribution and score to be evaluated without numerical convolution. By contrast, the Student-$t$ and Huber filters are not closed under convolution with Gaussian state uncertainty, so their implementations rely on approximations to the prediction-error density.

Finally, all models imply highly persistent latent volatility dynamics, with estimates of $\phi$ close to $0.97$. The exception is the pure Cauchy
specification, which yields a somewhat lower persistence estimate
($\phi\approx 0.936$), reflecting its greater reliance on the measurement-error component to absorb large fluctuations.

\begin{table}[htbp]
\caption{Estimation Results under Different Filtering Methods}
\spacingset{1.15}
\begin{centering}
\begin{small}
\begin{tabularx}{\textwidth}{p{2.2cm}YYYYYY}
\toprule
    \midrule
\multicolumn{1}{c}{Filter}
& \begin{tabular}{c}GCC\\[-0.5mm](Voigt)\end{tabular}
& \begin{tabular}{c}Kalman\\[-0.5mm](Gaussian)\end{tabular}
& Cauchy
& \begin{tabular}{c}Normal-\\[-0.5mm]Laplace\end{tabular}
& Student-$t$
& Huber \\
\midrule      
    \multicolumn{1}{c}{$\sigma$} & 0.1817   & 0.2980 &        & 0.0001 & 0.1696 & 0.1749 \\
                                 & {\it (0.0053)} & {\it (0.0109)}  &        & {\it (0.0001)}  & {\it (0.0044)}  & {\it (0.0134)} \\
    \multicolumn{1}{c}{$\gamma$} & 0.0199   &        & 0.0553 & 0.1786 &        &  \\
                                 & {\it (0.0021)}  &        & {\it (0.0037)}  & {\it (0.0042)}  &        &  \\
    \multicolumn{1}{c}{$\nu$}    &        &        &        &        & 5.2545 &  \\
                                 &        &        &        &        & {\it (0.3288)}  &  \\
    \multicolumn{1}{c}{$k$ }     &        &        &        &        &        & 1.3141 \\
                                &        &        &        &        &        & {\it (0.1846)} \\
    \multicolumn{1}{c}{$\mu$ }   & -1.9420  & -1.9492 & -1.9529 & -1.9600   & -1.9667 & -1.9437 \\
                                 & {\it (0.0502)}  & {\it (0.0535)}  & {\it (0.0351)}  & {\it (0.0490)}  & {\it (0.0461)}  & {\it (0.0622)} \\
    \multicolumn{1}{c}{$\phi$}   & 0.9716   & 0.9768   & 0.9361   & 0.9735 & 0.9699 & 0.9702 \\
                                & {\it (0.0040)}  & {\it (0.0037)}  & {\it (0.0070)}  & {\it (0.0034)}  & {\it (0.0039)}  & {\it (0.0058)} \\
    \multicolumn{1}{c}{$\tau$}      & 0.1138 & 0.1010 & 0.1813        & 0.1087 & 0.1145 & 0.1180 \\
                                     & {\it (0.0055)}  & {\it (0.0086)}  & {\it (0.0072)}  & {\it (0.0042)}  & {\it (0.0053)}  & {\it (0.0052)} \\
          &        &        &        &        &        &  \\
    \multicolumn{1}{c}{$\operatorname{var}(x_t)$}    & 0.2316 & 0.2236 & 0.2659 & 0.2257 & 0.2209 & 0.2373 \\
\\[-0.2cm]    
    \multicolumn{1}{c}{$\ell$}    & \textbf{-1,306} & -2,465 & -1,536 & -1,474 & -1,349 & -1,445 \\
          &        &        &        &        &        &  \\
    Closed-form & Yes    & Yes    & Yes    & Yes    & No     & No \\
\midrule
\bottomrule
\end{tabularx}
\end{small}

\par\end{centering}
{\footnotesize \textit{Note:} Estimation results under different filtering methods for the logarithm of daily realized volatility for the Technology Select Sector SPDR Fund (XLK). Sample: December 22, 1998 to November 7, 2025. Column headers indicate the filter name, with the measurement-error distribution in parentheses when it differs from the filter name. The parameter $k$ denotes the standardized threshold in the Huber density. $\operatorname{var}(x_t)=\tau^2/(1-\phi^2)$ is the unconditional variance of the latent state. We also report the maximized full-sample prediction-error criterion, $\ell$, and indicate the largest value in bold. For GCC, Kalman, Cauchy, and Normal-Laplace, this criterion is the log-likelihood based on the closed-form prediction-error density under the Gaussian prediction approximation. For Student-$t$ and Huber, it is a pseudo-log likelihood based on the standard same-family approximation to the prediction-error density. The last row indicates whether the filtering method has a closed-form expression for the prediction-error density. Robust sandwich standard errors are provided in parentheses. Huber filter standard errors are obtained via bootstrap due to a discontinuity in the second derivative of the log-likelihood. \label{tab:DFest-block}}
\end{table}

\subsection{Filtered volatility and residual decomposition}

Figure \ref{fig:logrv} plots the observed log realized-volatility series alongside the filtered GCC estimate of the latent volatility component. The observed data exhibit large spikes during market stress periods, including the Dot-Com bubble burst, the aftermath of September 11, 2001, the May 6, 2010 Flash Crash, and the August 24, 2015 market disruption. Conversely, unusually low measurements in the early sample are plausibly driven by market microstructure effects and thin trading.

The filtered state is substantially smoother than the observed series and is
not pulled strongly toward isolated extreme observations. This behavior is a
direct consequence of the redescending location score of the Voigt
prediction-error density, for which $\psi_t\to0$ as $|e_t|\to\infty$. Large
prediction errors are therefore downweighted in the state update and are instead
attributed primarily to the heavy-tailed component of the measurement error.

The right panel of Figure \ref{fig:logrv}, labeled ``Measurement errors,'' illustrates the filtered conditional-mean decomposition of the measurement
error. The Gaussian component captures regular, small-to-moderate measurement variation around the latent state, while the Cauchy component absorbs the most extreme deviations.
Thus, the GCC filter preserves sensitivity to ordinary daily variation in realized volatility while limiting the influence of observations that are more plausibly interpreted as transitory measurement disturbances.

\begin{figure}[t]
\spacingset{1.15}
\begin{centering}
\includegraphics[width=1\textwidth]{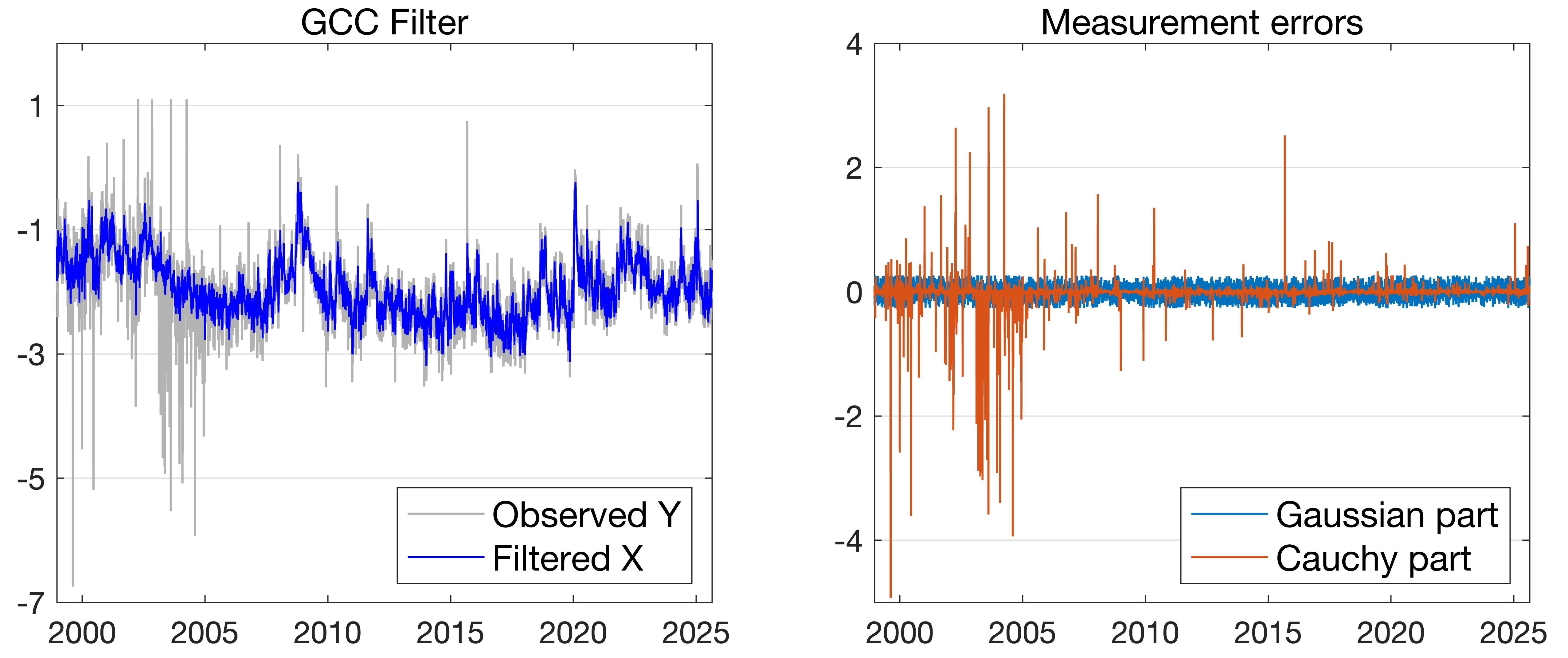}
\par
\end{centering}
\caption{\footnotesize GCC filtering results for the daily logarithm of realized
volatility of the Technology Select Sector SPDR Fund (XLK), December 22, 1998
to November 7, 2025. The left panel shows the observed series $y_t$ together
with the filtered state estimate. The right panel presents the filtered conditional means of the Gaussian and Cauchy measurement-error components using (\ref{eq:filtered_cauchy}).}
\label{fig:logrv}
\end{figure}

\section{Conclusion}\label{sec:Conclusion}

This paper develops likelihood inference and filtering for the Gauss-Cauchy convolution
model. Although familiar in spectroscopy, the Voigt density has been less used as a statistical
likelihood because exact evaluation is viewed as difficult. We show this is avoidable: the
scaled complementary error function gives an algebraically tractable Voigt representation.
Its identities yield closed-form likelihood, score, Hessian, and conditional-moment formulas,
so MLE requires neither numerical convolution nor finite-difference derivatives.

These formulas also provide an exact finite-sample likelihood and Fisher-information
building block for Brownian-Cauchy L\'evy increments, complementing the high-frequency
analysis of \citet{AitSahaliaJacod:2008}. In that setting, the Gaussian scale is proportional
to $\sqrt{\Delta}$, while the Cauchy scale is proportional to $\Delta$, so the same Voigt
representation evaluates the finite-$\Delta$ convolution density and its derivatives directly.

The same structure gives a Masreliez-type state-space filter with a redescending conditional-mean update. Under the Masreliez Gaussian prediction approximation, the GCC prediction error is Voigt and the state update is driven by its location score. This score is nearly linear centrally and redescends in the tails, so extreme observations have small update influence and are not treated as persistent latent state movements. In this operational sense, robustness follows from the conditional expectation implied by the convolution
model. Furthermore, an exact benchmark filter shows that the Masreliez approximation is accurate; the error is small for predictive densities and the one-step correction, and negligible in the empirically relevant region with a small Cauchy-to-Gaussian ratio.

The application to log volatility of XLK illustrates the method's practical value. The GCC filter attains the highest implemented prediction-error criterion among the specifications considered by combining a dominant Gaussian measurement-error component with a small Cauchy component that absorbs transient extremes, keeping the resulting state smooth without being driven by isolated spikes.

These results complement recent machine-learning approaches to Voigt fitting by identifying the likelihood target, score, and Fisher information that fast algorithms may approximate. Extensions include multivariate GCC models, richer latent dynamics, hybrid likelihood-machine-learning estimation, and applications where persistent Gaussian variation is observed through intermittent heavy-tailed noise.

\section*{Data and Code Availability Statement}

Replication code and data for the numerical experiments and empirical analysis is
available on GitHub at
\url{https://github.com/reinhardhansen/GaussCauchyConvolution-Replication}.
The repository contains the scripts used to reproduce the simulation results,
empirical estimates, tables, and figures reported in this paper.

{\spacingset{1.2}
\bibliographystyle{apalike}
\bibliography{prh}
}

% \phantomsection\label{supplementary-material}
% \bigskip

% \begin{center}

% \newpage

% {\large\bf SUPPLEMENTARY MATERIAL}

% \end{center}

% \begin{description}
% \item[Title:]
% Brief description. (file type)
% \item[R-package for MYNEW routine:]
% R-package MYNEW containing code to perform the diagnostic methods
% described in the article. The package also contains all datasets used as
% examples in the article. (GNU zipped tar file)
% \item[HIV data set:]
% Data set used in the illustration of MYNEW method in
% Section~\ref{sec-verify} (.txt file).
% \end{description}

% \section{BibTeX}\label{bibtex}

% We encourage you to use BibTeX. If you have, please feel free to use the
% package natbib with any bibliography style you're comfortable with. The
% .bst file agsm has been included here for your convenience. 

\newpage

%\phantomsection\label{supplementary-material}
%\bigskip

\part*{Supplementary Material}
\appendix
\setcounter{page}{1}
\global\long\def\thepage{S.\arabic{page}}%

\setcounter{equation}{0}
\global\long\def\theequation{A.\arabic{equation}}%

\setcounter{lemma}{0}
\global\long\def\thelemma{A.\arabic{lemma}}%

\section{Proofs}\label{sec:SupplementA}

\noindent \textbf{Proof of Proposition \ref{prop:Voigt}.} It suffices to prove the result for $\mu=0$; the general case follows by replacing $y$ with $y-\mu$. For $s\geq0$
we have $\varphi_{Y}(s)=\varphi_{Z}(s)\varphi_{X}(s)=e^{-\sigma^{2}s^{2}/2}e^{-\gamma s}$,
such that $f_{Y}(y)=\frac{1}{\pi}\int_{0}^{\infty}\cos(sy)e^{-\sigma^{2}s^{2}/2-\gamma s}ds$
follows by Euler's formula 
$$
\operatorname{Re}[e^{-isy}]=\operatorname{Re}[\cos(-sy)+i\sin(-sy)]=\cos(-sy)=\cos(sy).
$$
Next, we factorize the integral into two terms,
\begin{align*}
\frac{1}{\pi}\int_{0}^{\infty}\cos(sy)e^{-\frac{\sigma^{2}s^{2}}{2}-\gamma s}ds & =\frac{1}{2\pi}\int_{0}^{\infty}\left(e^{isy}+e^{-isy}\right)e^{-\frac{\sigma^{2}s^{2}}{2}-\gamma s}ds\\
 & =\underbrace{\frac{1}{2\pi}\int_{0}^{\infty}e^{isy}e^{-\frac{\sigma^{2}s^{2}}{2}-\gamma s}ds}_{I_{1}}+\underbrace{\frac{1}{2\pi}\int_{0}^{\infty}e^{-isy}e^{-\frac{\sigma^{2}s^{2}}{2}-\gamma s}ds}_{I_{2}}.
\end{align*}
Using $isy-\frac{\sigma^2s^{2}}{2}-\gamma s=\tfrac{1}{2}(\frac{\gamma-iy}{\sigma})^{2}-\tfrac{1}{2}(\frac{\gamma-iy}{\sigma}+\sigma s)^{2}$
we rewrite the first term as
\begin{align*}
I_{1} & =  \frac{e^{\tfrac{1}{2}(\frac{\gamma-iy}{\sigma})^{2}}}{2\pi}\int_{0}^{\infty}e^{-\tfrac{1}{2}(\frac{\gamma-iy}{\sigma}+\sigma s)^{2}}ds
=\frac{e^{\tfrac{1}{2}(\frac{\gamma-iy}{\sigma})^{2}}}{2\pi}\int_{\frac{\gamma-iy}{\sigma\sqrt{2}}}^{\infty}e^{-\zeta_{-}^{2}}\sqrt{\tfrac{2}{\sigma^{2}}}d\zeta_{-}\\
 &=  \frac{e^{\tfrac{1}{2}(\frac{\gamma-iy}{\sigma})^{2}}}{2\sqrt{2\pi\sigma^{2}}}\frac{2}{\sqrt{\pi}}\int_{\frac{\gamma-iy}{\sigma\sqrt{2}}}^{\infty}e^{-\zeta_{-}^{2}}d\zeta_{-}
 =\frac{1}{2}\frac{e^{\tfrac{1}{2}(\frac{\gamma-iy}{\sigma})^{2}}}{\sqrt{2\pi\sigma^{2}}}\left[\operatorname{erfc}\left(\frac{\gamma-iy}{\sigma\sqrt{2}}\right)\right],
\end{align*}
where we applied the substitution $\zeta_{-}=(\frac{\gamma-iy}{\sigma}+\sigma s)/\sqrt{2}$. 

For the second term, we use the related identity, $-isy-\frac{\sigma^{2}s^{2}}{2}-\gamma s=\tfrac{1}{2}(\frac{\gamma+iy}{\sigma})^{2}-\tfrac{1}{2}(\frac{\gamma+iy}{\sigma}+\sigma s)^{2}$,
and the substitution, $\zeta_{+}=(\frac{\gamma+iy}{\sigma}+\sigma s)/\sqrt{2}$,
to rewrite it as,
$$
I_{2}=\frac{e^{\tfrac{1}{2}(\frac{\gamma+iy}{\sigma})^{2}}}{2\pi}\int_{0}^{\infty}e^{-\tfrac{1}{2}(\frac{\gamma+iy}{\sigma}+\sigma s)^{2}}ds=\frac{e^{\tfrac{1}{2}(\frac{\gamma+iy}{\sigma})^{2}}}{2\sqrt{2\pi\sigma^{2}}}\operatorname{erfc}\left(\frac{\gamma+iy}{\sigma\sqrt{2}}\right).
$$
Combining the results we arrive at
\begin{align*}
f_{Y}\left(y\right) & =\frac{1}{\sqrt{2\pi\sigma^{2}}}\left[\frac{1}{2}e^{\tfrac{1}{2}(\frac{\gamma-iy}{\sigma})^{2}}\operatorname{erfc}\left(\tfrac{\gamma-iy}{\sigma\sqrt{2}}\right)+\frac{1}{2}e^{\tfrac{1}{2}(\frac{\gamma+iy}{\sigma})^{2}}\operatorname{erfc}\left(\tfrac{\gamma+iy}{\sigma\sqrt{2}}\right)\right]\\
 & =\frac{1}{\sqrt{2\pi\sigma^{2}}}\operatorname{Re}\left[e^{\tfrac{1}{2}(\frac{\gamma+iy}{\sigma})^{2}}\operatorname{erfc}\left(\tfrac{\gamma+iy}{\sigma\sqrt{2}}\right)\right]=\frac{1}{\sqrt{2\pi\sigma^{2}}}\operatorname{Re}\left[\operatorname{erfcx}\left(\tfrac{\gamma+iy}{\sigma\sqrt{2}}\right)\right].
\end{align*}
This completes the proof. \hfill$\square$\medskip{}

\noindent \textbf{Proof of Corollary \ref{cor:MillsRatio}.}
Let $t=\frac{\gamma+i(y-\mu)}{\sigma}$. From Proposition~\ref{prop:Voigt} we have the representation
$$
f_Y(y;\mu,\sigma,\gamma)=\frac{1}{\sqrt{2\pi\sigma^2}}\operatorname{Re}\left[\operatorname{erfcx}\left(\frac{t}{\sqrt{2}}\right)\right].
$$
Using
$m(t)=\frac{\Phi(-t)}{\phi(t)}=\sqrt{\frac{\pi}{2}}\operatorname{erfcx}\left(\frac{t}{\sqrt{2}}\right)$,
it follows that
$\operatorname{erfcx}\left(\frac{t}{\sqrt{2}}\right)=\sqrt{\frac{2}{\pi}}m(t)$.
Substituting this into the expression for $f_Y$ yields
$$
f_Y(y;\mu,\sigma,\gamma)=\frac{1}{\sqrt{2\pi\sigma^2}}\operatorname{Re}\left[\sqrt{\frac{2}{\pi}}m(t)\right]
=\frac{1}{\pi\sigma}\operatorname{Re}\left[m\left(\frac{\gamma+i(y-\mu)}{\sigma}\right)\right],
$$
which is the stated result. \hfill$\square$\medskip{}

\noindent{}\textbf{Proof of Lemma \ref{lem:erfcx}.} (i) First, $\mathsf{e}$
is entire because it is the composition of two entire functions, $\exp(w^{2})$ and $\operatorname{erfc}(w)$.
The identity
$\mathsf{e}^\prime(w)=2w\mathsf{e}(w)-\frac{2}{\sqrt{\pi}}$
follows by differentiating
$$
\mathsf{e}(w)=\exp(w^{2})\frac{2}{\sqrt{\pi}}\int_{w}^{\infty}e^{-t^{2}}dt.
$$
The expression for $\mathsf{e}^{(n)}(w)$ follows by induction. The induction starts from $p_0(w)=1$ and $q_0(w)=0$. Differentiating
$\mathsf{e}^{(n)}(w)=p_n(w)\mathsf{e}(w)+q_n(w)$,
yields
$$
\mathsf{e}^{(n+1)}(w) =
\{p_n^\prime(w)+2wp_n(w)\}\mathsf{e}(w)
+ \left\{q_n^\prime(w)-\frac{2}{\sqrt{\pi}}p_n(w)\right\}.
$$

The symmetry properties follow from the conjugation identity
$\mathsf{e}(\overline{w})=\overline{\mathsf{e}(w)}$.
Since
$$
w_{\mu+r,\theta}
=
\frac{\gamma+ir}{\sigma\sqrt{2}}
=
\overline{
\frac{\gamma-ir}{\sigma\sqrt{2}}}
=
\overline{w_{\mu-r,\theta}},
$$
we have
$\mathsf{u}(\mu+r;\theta)=\mathsf{u}(\mu-r;\theta)$ and $\mathsf{v}(\mu+r;\theta)=-\mathsf{v}(\mu-r;\theta)$.

Finally, set 
$$\alpha=\frac{\gamma}{\sigma\sqrt{2}}>0,
\qquad
\beta=\frac{y-\mu}{\sigma\sqrt{2}},
$$
such that $w_{y,\theta}=\alpha+i\beta$, and observe that
$$
\mathsf{u}(y;\theta) = \operatorname{Re}\{\mathsf{e}(\alpha+i\beta)\}
= \frac{\alpha}{\pi}\int_{-\infty}^{\infty} \frac{e^{-t^2}}{(t+\beta)^2+\alpha^2}dt >0.
$$
\hfill$\square$\medskip{}

\noindent{}{\bf Proof of Corollary \ref{cor:VoigtAlgebraicClosure}.}
By Lemma \ref{lem:erfcx}, every derivative of $e(w)=\operatorname{erfcx}(w)$ has the form
$$
e^{(n)}(w)=p_n(w)e(w)+q_n(w),
$$
for polynomials $p_n$ and $q_n$. The chain rule then implies that derivatives of $e(w_{y,\theta})$ with respect to $y$, $\mu$, $\sigma$, and $\gamma$ are obtained by algebraic operations involving $w_{y,\theta}$, its derivatives, and $e(w_{y,\theta})$. Taking real and imaginary parts expresses these derivatives in terms of $\mathsf{u}(y;\theta)$ and $\mathsf{v}(y;\theta)$. The prefactor $(2\pi\sigma^2)^{-1/2}$ contributes only algebraic derivatives in $\sigma$. Hence every finite-order derivative of $f_Y(y;\theta)$ has the stated form.

For the log-density, repeated differentiation of
$$
\log f_Y(y;\theta)
=
-\frac{1}{2}\log(2\pi)-\log\sigma+\log \mathsf{u}(y;\theta)
$$
produces rational expressions whose denominators involve powers of $\mathsf{u}(y;\theta)$ and $\sigma$. Lemma \ref{lem:erfcx} gives $\mathsf{u}(y;\theta)>0$, and $\sigma>0$ by assumption. Therefore these expressions are well defined for all real $y$. This proves the result.
\hfill$\square$\medskip{}

\noindent \textbf{Proof of Theorem \ref{thm:CondDensity}.} By Bayes'
rule we have $f_{Z|Y}(z|y)=f_{Y|Z}(y|z)\frac{f_{Z}(z)}{f_{Y}(y)}$.
The general results follow from the simple expressions, $f_{Y|Z}(y|z)=\tfrac{1}{\pi\gamma}[1+(\frac{y-\mu-z}{\gamma})^{2}]^{-1}$
and $f_{Z}(z)=e^{-z^{2}/(2\sigma^{2})}/\sqrt{2\pi\sigma^{2}}$, and
the expression for $f_{Y}(y)$ in (\ref{eq:Kendall}). 
$$
f_{Z|Y}(z|y)=\tfrac{1}{\pi\gamma}\left[1+\left(\tfrac{y-\mu-z}{\gamma}\right)^{2}\right]^{-1}\frac{\tfrac{1}{\sqrt{2\pi\sigma^{2}}}e^{-\frac{z^{2}}{2\sigma^{2}}}}{\tfrac{1}{\sqrt{2\pi\sigma^{2}}}\mathsf{u}(y;\theta)}=\frac{1}{\gamma\pi \mathsf{u}(y;\theta)}\frac{e^{-\frac{z^{2}}{2\sigma^{2}}}}{1+\left(\frac{y-\mu-z}{\gamma}\right)^{2}}.
$$
The limiting result follows by $y^{2}\mathsf{u}(y;\theta)=y^{2}\operatorname{Re} \left[\operatorname{erfcx}\left(\frac{\gamma+i(y-\mu)}{\sigma\sqrt{2}}\right)\right]\rightarrow\sqrt{2/\pi}\sigma\gamma$
as $|y|\rightarrow\pm\infty$, for any fixed $\mu$.

% Next, the expression for $y=\mu$ follows by $\lambda=\gamma/\sigma$
% \begin{align*}
% \operatorname{erfcx}\left(\tfrac{\lambda}{\sqrt{2}}\right) & =e^{\frac{\lambda^{2}}{2}}\frac{2}{\sqrt{\pi}}\int_{\lambda/\sqrt{2}}^{\infty}e^{-t^{2}}dt\\
%  & =e^{\frac{\lambda^{2}}{2}}\frac{2}{\sqrt{2\pi}}\int_{\lambda}^{\infty}e^{-s^{2}/2}ds,\quad s=\sqrt{2}t\\
%  & =2e^{\frac{\lambda^{2}}{2}}\Pr(Z\geq\lambda)
%   =e^{\frac{\lambda^{2}}{2}}\Pr(|Z|\geq\lambda)
% \end{align*}
% such that $$f_{Z|Y=\mu}(z)=\frac{1}{\gamma\pi\Pr(|Z|\geq\tfrac{\gamma}{\sigma})}\frac{e^{-\frac{\left(z^{2}+\gamma^{2}\right)}{2\sigma^{2}}}}{1+(\frac{z}{\gamma})^{2}}.$$ This completes the proof. 

Let $\tilde y=y-\mu$. As $|\tilde y|\to\infty$,
$\tilde y^2 \mathsf u(y;\theta)
\to \sqrt{2/\pi}\sigma\gamma$, and for fixed $z$,
\[
\frac{1}{1+\{(\tilde y-z)/\gamma\}^2}
\sim \frac{\gamma^2}{\tilde y^2}.
\]
Therefore
\[
\frac{1}{\gamma\pi\mathsf u(y;\theta)}
\frac{1}{1+\{(\tilde y-z)/\gamma\}^2}
\to
\frac{1}{\sqrt{2\pi}\sigma},
\]
which gives
\[
f_{Z |  Y}(z |  y)\to
\frac{1}{\sqrt{2\pi\sigma^2}}
\exp\{-z^2/(2\sigma^2)\}.
\]
This completes the proof. \hfill$\square$\medskip{}

\noindent{}\textbf{Proof of Proposition \ref{prop:TweedieGaussian}.}
Let $F_X$ denote the distribution function of $X$. Since $Z$ is Gaussian and
independent of $X$, the density of $Y=\mu+Z+X$ is the Gaussian convolution
$$
f_Y(y)=\int \phi_\sigma(y-\mu-x)dF_X(x),
$$
where $\phi_\sigma$ is the density of $\mathcal{N}(0,\sigma^2)$. This density is
strictly positive because $\phi_\sigma(y-\mu-x)>0$ for all real $x$ and $y$.
Moreover, $f_Y$ is smooth because all derivatives of the Gaussian density are
bounded for fixed $\sigma>0$, so differentiation under the integral sign is
valid.

Differentiating the marginal density gives
$$
f_Y^{\prime}(y) = \int \phi_\sigma'(y-\mu-x)dF_X(x) =
-\frac{1}{\sigma^2} \int (y-\mu-x)\phi_\sigma(y-\mu-x)dF_X(x).
$$
Since the conditional distribution of $X$ given $Y=y$ is proportional to
$\phi_\sigma(y-\mu-x)dF_X(x)$ and $Z=y-\mu-X$, the last integral equals $f_Y(y)\mathbb{E}[Z|Y=y]$.
Therefore
$$
\mathbb{E}[Z|Y=y] = -\sigma^2\frac{f_Y^{\prime}(y)}{f_Y(y)} = -\sigma^2\frac{\partial}{\partial y}\log f_Y(y).
$$

For the second moment, differentiating once more gives
$$
f_Y^{\prime\prime}(y) =
\int \phi_\sigma^{\prime\prime}(y-\mu-x)dF_X(x)=\int
\left[\frac{(y-\mu-x)^2}{\sigma^4}-\frac{1}{\sigma^2}
\right]\phi_\sigma(y-\mu-x)dF_X(x).
$$
Hence
$\mathbb{E}[Z^2|Y=y] = \sigma^2+\sigma^4\frac{f_Y^{\prime\prime}(y)}{f_Y(y)}$ and combining this with the conditional mean expression yields
$\operatorname{var}(Z|Y=y)=\sigma^2+\sigma^4\frac{\partial^2}{\partial y^2}\log f_Y(y)$. This completes the proof. 
\hfill$\square$\medskip{}

\noindent{}\textbf{Proof of Corollary \ref{cor:CondExpect}.}
Since the density $f_Y(y;\theta)$ depends on $y$ and $\mu$ only through $\tilde y=y-\mu$, we have $\partial_y\log f_Y(y;\theta)=-s_\mu(y;\theta)$. By Proposition \ref{prop:TweedieGaussian} and Lemma \ref{lem:ScoreHessian} we have
$$
\mathbb{E}[Z|Y=y]=\sigma^2s_\mu=\tilde y+\gamma\frac{\mathsf{v}(y;\theta)}{\mathsf{u}(y;\theta)}.
$$
For the conditional variance, we also use that $f_Y$ depends on $y$ and $\mu$ only through $\tilde y$, such that the second derivative with respect to $y$ equals $H_{\mu\mu}$. Hence, by Proposition \ref{prop:TweedieGaussian} and Lemma \ref{lem:ScoreHessian},
$$
\operatorname{var}(Z|Y=y)=\sigma^2+\sigma^4H_{\mu\mu}=\sigma^2+\sigma^3s_\sigma-(\sigma^2s_\mu)^2.
$$
Substituting
$$
\sigma^2s_\mu=\tilde y+\gamma\frac{\mathsf{v}}{\mathsf{u}},\qquad \sigma^3s_\sigma=\tilde y^2-\gamma^2-\sigma^2+2\gamma\tilde y\frac{\mathsf{v}}{\mathsf{u}}+\sqrt{\frac{2}{\pi}}\frac{\sigma\gamma}{\mathsf{u}},
$$
where $\mathsf{u}=\mathsf{u}(y;\theta)$ and $\mathsf{v}=\mathsf{v}(y;\theta)$, yields
$$
\operatorname{var}(Z|Y=y)=\sqrt{\frac{2}{\pi}}\frac{\sigma\gamma}{\mathsf{u}(y;\theta)}-\gamma^2\left(1+\frac{\mathsf{v}^2(y;\theta)}{\mathsf{u}^2(y;\theta)}\right).
$$
Finally, we obtain the first-order condition for the extrema of the conditional expectation. Let $m(y)=\mathbb{E}[Z|Y=y]$. Proposition \ref{prop:TweedieGaussian} implies
$\operatorname{var}(Z|Y=y)=\sigma^2\{1-m^{\prime}(y)\}$.
Thus, at any interior extremum $y^\ast$ of $m(y)$, we have $m^{\prime}(y^\ast)=0$ and therefore $\operatorname{var}(Z|Y=y^\ast)=\sigma^2$. Substituting the variance formula above and multiplying by $\mathsf{u}^2(y^\ast;\theta)$ gives
$$
(\gamma^2+\sigma^2)\mathsf{u}^2(y^\ast;\theta)+\gamma^2\mathsf{v}^2(y^\ast;\theta)-\sqrt{\frac{2}{\pi}}\sigma\gamma\mathsf{u}(y^\ast;\theta)=0.
$$
This completes the proof. \hfill$\square$\medskip{}

\noindent{}\textbf{Proof of Lemma \ref{lem:ScoreHessian}.}
Let $\tilde y=y-\mu$ and write
$$
w=w_{y,\theta}=\alpha+i\beta\quad\text{where}
\quad\alpha=\tfrac{\gamma}{\sigma\sqrt{2}},
\quad
\beta=\tfrac{\tilde y}{\sigma\sqrt{2}}, 
$$
such that $\mathsf{e}(w)=\mathsf{u}+i\mathsf{v}$ where $\mathsf{u}=\mathsf{u}(y;\theta)$ and
$\mathsf{v}=\mathsf{v}(y;\theta)$. The log-density is simply
$$
\ell(y;\theta)
=
\log \mathsf{u}-\log\sigma-\log\sqrt{2\pi}.
$$
By Lemma \ref{lem:erfcx} we now have
$$
\mathsf{e}^\prime(w)=2w\mathsf{e}(w)-\tfrac{2}{\sqrt{\pi}}=\left\{
2(\alpha\mathsf{u}-\beta\mathsf{v})-\tfrac{2}{\sqrt{\pi}}\right\}+i2(\alpha\mathsf{v}+\beta\mathsf{u}).
$$
Furthermore,
$$
\frac{\partial w}{\partial\mu}
=
-\frac{i}{\sigma\sqrt{2}},
\qquad
\frac{\partial w}{\partial\gamma}
=
\frac{1}{\sigma\sqrt{2}},
\qquad
\frac{\partial w}{\partial\sigma}
=
-\frac{w}{\sigma}.
$$
Taking real parts of
$\mathsf{e}^\prime(w)\partial w/\partial\theta_j$ gives
\begin{align*}
\frac{\partial\mathsf{u}}{\partial\mu}
&=
\frac{1}{\sigma^2}
\left(
\gamma\mathsf{v}+\tilde y\mathsf{u}
\right),
\\
\frac{\partial\mathsf{u}}{\partial\gamma}
&=
\frac{1}{\sigma^2}
\left(
\gamma\mathsf{u}-\tilde y\mathsf{v}
-\sigma\sqrt{\tfrac{2}{\pi}}
\right),
\\
\frac{\partial\mathsf{u}}{\partial\sigma}
&=
\frac{1}{\sigma^3}
\left[
(\tilde y^2-\gamma^2)\mathsf{u}
+
2\gamma\tilde y\mathsf{v}
+
\sigma\gamma\sqrt{\tfrac{2}{\pi}}
\right].
\end{align*}
Since
$
s_\mu=\frac{1}{\mathsf{u}}\frac{\partial\mathsf{u}}{\partial\mu}$,
$s_\gamma=\frac{1}{\mathsf{u}}\frac{\partial\mathsf{u}}{\partial\gamma}$,
and $s_\sigma=
\frac{1}{\mathsf{u}}\frac{\partial\mathsf{u}}{\partial\sigma}
-
\frac{1}{\sigma},
$
the score components are
\begin{align*}
s_\mu 
    & = \frac{\gamma\mathsf{v}+\tilde y\mathsf{u}}{\sigma^2\mathsf{u}},\\
s_\gamma
    &=\frac{\gamma\mathsf{u}-\tilde y\mathsf{v}-\sigma\sqrt{\frac{2}{\pi}}}{\sigma^2\mathsf{u}},\\
s_\sigma
    &=\frac{(\tilde y^2-\gamma^2-\sigma^2)\mathsf{u}+2\gamma\tilde y\mathsf{v}+\sigma\gamma\sqrt{\frac{2}{\pi}}}{\sigma^3\mathsf{u}}.
\end{align*}
We next derive the Hessian entries. Let
$$
H_{ij} = \frac{\partial^2\ell(y;\theta)}{\partial i\partial j} = \frac{\partial s_i}{\partial j}.
$$
We use three identities. The first is the \emph{heat identity}
$f_\sigma(y;\theta)=\sigma f_{\mu\mu}(y;\theta)$.
It follows directly from the Gaussian convolution structure. Writing
$$
f_Y(y;\mu,\sigma,\gamma)=\int \varphi_\sigma(y-\mu-x)c_\gamma(x)dx,
$$
where $\varphi_\sigma$ is the $\mathcal N(0,\sigma^2)$ density and $c_\gamma$ is the Cauchy density, the Gaussian kernel satisfies
$$
\frac{\partial}{\partial\sigma}\varphi_\sigma(y-\mu-x)
=
\sigma\frac{\partial^2}{\partial\mu^2}\varphi_\sigma(y-\mu-x).
$$
Differentiating under the integral sign gives the stated identity.

The second is the \emph{Laplace identity}
$f_{\mu\mu}(y;\theta)+f_{\gamma\gamma}(y;\theta)=0$, 
which follows by writing the Voigt density as the convolution of the Gaussian kernel with the Cauchy kernel and using the harmonicity relation
$c_{\gamma,xx}(x)+c_{\gamma,\gamma\gamma}(x)=0$. 

The third is \emph{scale homogeneity of the density},
$\sigma s_\sigma+\gamma s_\gamma-\tilde y s_\mu=-1$,
which follows from
$$
f_Y(\mu+a\tilde y;\mu,a\sigma,a\gamma)=a^{-1}f_Y(\mu+\tilde y;\mu,\sigma,\gamma),
$$
after differentiating with respect to $a$ at $a=1$, with $\tilde y$ held fixed.
The heat identity gives
$$
s_\sigma=\frac{f_\sigma}{f} = \sigma\frac{f_{\mu\mu}}{f} = \sigma(H_{\mu\mu}+s_\mu^2),
$$
and therefore
$H_{\mu\mu}=\frac{s_\sigma}{\sigma}-s_\mu^2$. 
Similarly, the Laplace identity gives
$$
0=\frac{f_{\mu\mu}+f_{\gamma\gamma}}{f}=H_{\mu\mu}+s_\mu^2+H_{\gamma\gamma}+s_\gamma^2.
$$
Using the preceding expression for $H_{\mu\mu}$, we obtain
$H_{\gamma\gamma}=-\frac{s_\sigma}{\sigma}-s_\gamma^2$.

Finally, we obtain the mixed derivative $H_{\mu\gamma}$ and the entries
involving $\sigma$. Since $\mathsf{e}$ is analytic, the Cauchy-Riemann
relations imply
$$
\frac{\partial\mathsf{v}}{\partial\gamma} = \frac{\partial\mathsf{u}}{\partial\mu} = \mathsf{u}s_\mu,
\qquad
\frac{\partial\mathsf{u}}{\partial\gamma}=\mathsf{u}s_\gamma.
$$
Differentiating
$
s_\mu=\frac{\tilde y\mathsf{u}+\gamma\mathsf{v}}{\sigma^2\mathsf{u}}
$
with respect to $\gamma$ gives
\begin{align*}
H_{\mu\gamma}
&=
\frac{1}{\sigma^2}
\frac{
\mathsf{u}
\left(
\tilde y\mathsf{u}s_\gamma
+
\mathsf{v}
+
\gamma\mathsf{u}s_\mu
\right)
-
(\tilde y\mathsf{u}+\gamma\mathsf{v})\mathsf{u}s_\gamma
}{
\mathsf{u}^2
}
\\
&=
\frac{1}{\sigma^2}
\left(
\tilde y s_\gamma
+
\frac{\mathsf{v}}{\mathsf{u}}
+
\gamma s_\mu
\right)
-
s_\mu s_\gamma.
\end{align*}

The remaining Hessian entries follow from differentiating the homogeneity
identity. Differentiating
$
\sigma s_\sigma+\gamma s_\gamma-\tilde y s_\mu=-1
$
with respect to $\mu$, $\gamma$, and $\sigma$, respectively, yields
\begin{align*}
\sigma H_{\mu\sigma}
+
\gamma H_{\mu\gamma}
-
\tilde y H_{\mu\mu}
+
s_\mu
&=
0,
\\
\sigma H_{\gamma\sigma}
+
\gamma H_{\gamma\gamma}
+
s_\gamma
-
\tilde y H_{\mu\gamma}
&=
0,
\\
\sigma H_{\sigma\sigma}
+
s_\sigma
+
\gamma H_{\gamma\sigma}
-
\tilde y H_{\mu\sigma}
&=
0.
\end{align*}
Solving these equations gives
\begin{align*}
H_{\mu\sigma}
&=
-\frac{1}{\sigma}
\left(
s_\mu+\gamma H_{\mu\gamma}-\tilde y H_{\mu\mu}
\right),
\\
H_{\gamma\sigma}
&=
-\frac{1}{\sigma}
\left(
s_\gamma+\gamma H_{\gamma\gamma}-\tilde y H_{\mu\gamma}
\right),
\\
H_{\sigma\sigma}
&=
-\frac{1}{\sigma}
\left(
s_\sigma+\gamma H_{\gamma\sigma}-\tilde y H_{\mu\sigma}
\right).
\end{align*}
This completes the proof.\hfill$\square$\medskip{}

\subsection{Proof of Theorem \ref{thm:MLE-consistent-asN}} 
We begin with some preliminary results.

\begin{lemma}[Identification]\label{lem:MLE-identification}
Let $Y\sim \mathcal{V}(\mu,\sigma,\gamma)$ 
with $\sigma\geq0$, $\gamma\geq0$, and not both $\sigma$ and $\gamma$ equal to zero. Then the mapping from $\theta=(\mu,\sigma,\gamma)^\prime$ to the distribution of $Y$ is one-to-one. In particular, if
$\mathcal{V}(\mu,\sigma,\gamma)=\mathcal{V}(\mu_0,\sigma_0,\gamma_0)$,
then
$(\mu,\sigma,\gamma)=(\mu_0,\sigma_0,\gamma_0)$.
\end{lemma}

\begin{proof}
The characteristic function of $Y\sim\mathcal{V}(\mu,\sigma,\gamma)$ is
$$
\varphi(t;\theta)
=
\exp\left(
i\mu t-\frac{1}{2}\sigma^2t^2-\gamma|t|
\right),
\qquad t\in\mathbb{R}.
$$
Suppose that two parameter vectors $\theta=(\mu,\sigma,\gamma)^\prime$ and
$\theta_0=(\mu_0,\sigma_0,\gamma_0)^\prime$ generate the same distribution.
Then their characteristic functions are equal for all $t\in\mathbb{R}$:
$$
\exp\left(
i\mu t-\frac{1}{2}\sigma^2t^2-\gamma|t|
\right)
=
\exp\left(
i\mu_0 t-\frac{1}{2}\sigma_0^2t^2-\gamma_0|t|
\right).
$$
Taking absolute values gives
$$
\exp\left(
-\frac{1}{2}\sigma^2t^2-\gamma|t|
\right)
=
\exp\left(
-\frac{1}{2}\sigma_0^2t^2-\gamma_0|t|
\right).
$$
Hence, for all $t>0$,
$\frac{1}{2}(\sigma^2-\sigma_0^2)t^2+(\gamma-\gamma_0)t=0$ and  dividing by $t$ gives
$\frac{1}{2}(\sigma^2-\sigma_0^2)t+(\gamma-\gamma_0)=0
$ for every $t>0$. Therefore
$\sigma^2=\sigma_0^2$ and $\gamma=\gamma_0$. 
Since $\sigma,\sigma_0\geq0$, this implies $\sigma=\sigma_0$.

With $\sigma=\sigma_0$ and $\gamma=\gamma_0$, equality of the characteristic
functions reduces to
$$
\exp(i\mu t)=\exp(i\mu_0 t),
\qquad t\in\mathbb{R}.
$$
Thus $\exp(i(\mu-\mu_0)t)=1$
for all $t\in\mathbb{R}$, which is only possible if $\mu=\mu_0$. Hence $\theta=\theta_0$, proving identification.
\end{proof}

\begin{lemma}[Log-likelihood envelope]\label{lem:MLE-envelope}
Let
$$
\Theta
=
\{(\mu,\sigma,\gamma):|\mu|\leq\mu_{\max},0<\sigma\leq\sigma_{\max},
0<\gamma_{\min}\leq\gamma\leq\gamma_{\max}\}.
$$
Then there exist finite constants $C_0,C_1,C_2$ such that, for all
$\theta\in\Theta$ and all $y\in\mathbb{R}$,
$$
|\ell(y;\theta)|
\leq
C_0+C_1\log(C_2+|y|),
$$
where $\ell(y;\theta)=\log f_Y(y;\theta)$. Moreover,
$\mathbb{E}_{\theta_0}\left[C_0+C_1\log(C_2+|Y|)
\right]<\infty$, and consequently,
$\mathbb{E}_{\theta_0}\left[\sup_{\theta\in\Theta}|\ell(Y;\theta)|\right]<\infty$.
\end{lemma}

\begin{proof}
Write the Voigt density as the convolution of a Gaussian density and a
Cauchy density:
$$
f_Y(y;\theta)
=
\int_{-\infty}^{\infty}
c_\gamma(y-\mu-z)\phi_\sigma(z)dz,
$$
where
$c_\gamma(x)=\frac{1}{\pi}\frac{\gamma}{x^2+\gamma^2}$
and $\phi_\sigma$ denotes the $\mathcal{N}(0,\sigma^2)$ density, with the
degenerate interpretation when $\sigma=0$.

We first obtain a uniform upper bound. Since
$$
c_\gamma(x)\leq \frac{1}{\pi\gamma}
\leq
\frac{1}{\pi\gamma_{\min}},
$$
we have $f_Y(y;\theta)\leq \frac{1}{\pi\gamma_{\min}}$
for all $y$ and all $\theta\in\Theta$. Therefore
$\ell(y;\theta)\leq C_U$ for a finite constant $C_U$.

Next, we obtain a uniform lower bound. Let $Z_\sigma\sim\mathcal{N}(0,\sigma^2)$.
For $0<\sigma\leq\sigma_{\max}$,
$$
\Pr(|Z_\sigma|\leq\sigma_{\max})
\geq
\Pr(|Z_{\sigma_{\max}}|\leq\sigma_{\max})
=
\Pr(|Z|\leq1)
\equiv p_0>0,
$$
where $Z\sim\mathcal{N}(0,1)$. Hence
$$
f_Y(y;\theta) = \mathbb{E}
\left[ c_\gamma(y-\mu-Z_\sigma) \right]
\geq
\mathbb{E}\left[c_\gamma(y-\mu-Z_\sigma)
\mathbf{1}\{|Z_\sigma|\leq\sigma_{\max}\}\right].
$$
On the event $|Z_\sigma|\leq\sigma_{\max}$,
$|y-\mu-Z_\sigma|\leq |y|+\mu_{\max}+\sigma_{\max}$. 
For $\gamma\in[\gamma_{\min},\gamma_{\max}]$, we therefore obtain
$$
c_\gamma(y-\mu-Z_\sigma)
\geq
\frac{\gamma_{\min}}
{\pi\left\{\gamma_{\max}^2+(|y|+\mu_{\max}+\sigma_{\max})^2\right\}}.
$$
It follows that
$$
f_Y(y;\theta)
\geq
\frac{p_0\gamma_{\min}}
{\pi\left\{\gamma_{\max}^2+(|y|+\mu_{\max}+\sigma_{\max})^2\right\}}.
$$
Thus, for finite constants $C_L$ and $C_2$,
$-\ell(y;\theta)\leq C_L+2\log(C_2+|y|)$.
Combining the upper and lower bounds gives
$|\ell(y;\theta)|
\leq
C_0+C_1\log(C_2+|y|)$ 
for suitable finite constants $C_0,C_1,C_2$, uniformly over
$\theta\in\Theta$.

It remains to verify integrability under the true distribution. Under
$\theta_0$, we may write
$$
Y=\mu_0+\sigma_0Z+\gamma_0X,
$$
where $Z\sim\mathcal{N}(0,1)$ and $X\sim\operatorname{Cauchy}(0,1)$ are
independent, with the obvious convention when one scale parameter is zero.
Since
$$
\log(C_2+|Y|)
\leq
\log(C_2+|\mu_0|+\sigma_0|Z|+\gamma_0|X|),
$$
and both Gaussian and Cauchy random variables have finite logarithmic moments,
we have
$\mathbb{E}_{\theta_0}\log(C_2+|Y|)<\infty$. 
Therefore the envelope is integrable under the true distribution, and
$$
\mathbb{E}_{\theta_0}
\left[
\sup_{\theta\in\Theta}|\ell(Y;\theta)|
\right]
<\infty.
$$
\end{proof}

\begin{lemma}[Smoothness and derivative envelopes]\label{lem:MLE-derivative-envelope}
Let $K\subset\Theta$ be a compact subset for which 
$0<\sigma_{\min}\leq\sigma\leq\sigma_{\max}<\infty$ and $0<\gamma_{\min}\leq\gamma\leq\gamma_{\max}<\infty$
for all $\theta=(\mu,\sigma,\gamma)^\prime\in K$.  Then
$\ell(y;\theta)=\log f_Y(y;\theta)$ is twice continuously differentiable in
$\theta$ on $K$. Moreover, there exists a finite constant
$C_{K}$ such that, for all $y\in\mathbb R$,
$$
\sup_{\theta\in K}\|s(y;\theta)\|
\leq C_{K},
\qquad
\sup_{\theta\in K}\|H(y;\theta)\|
\leq C_{K},
$$
where
$s(y;\theta)=\frac{\partial\ell(y;\theta)}{\partial\theta}$ and $H(y;\theta)=
\frac{\partial^2\ell(y;\theta)}{\partial\theta\partial\theta^\prime}$.
Moreover, for each pair of parameter indices $i,j$, there exists an
integrable function $M_K(y)$ such that
$$
\sup_{\theta\in K}\left|\frac{\partial f_Y(y;\theta)}{\partial\theta_i} \right|
+\sup_{\theta\in K}\left|
\frac{\partial^2 f_Y(y;\theta)}
{\partial\theta_i\partial\theta_j}\right|\leq M_K(y).
$$
Consequently, if $\theta_0\in K$ then
$$
\mathbb{E}_{\theta_0}\|s(Y;\theta_0)\|^2<\infty,
\qquad
\mathbb{E}_{\theta_0}
\left[
\sup_{\theta\in K}\|H(Y;\theta)\|
\right]
<\infty.
$$
\end{lemma}

\begin{proof}
Let $\tilde y=y-\mu$ and write
$$
w=w_{y,\theta}
=
\frac{\gamma+i\tilde y}{\sigma\sqrt{2}},
\qquad
\mathsf{e}(w)=\mathsf{u}(y;\theta)+i\mathsf{v}(y;\theta).
$$
By Lemma \ref{lem:erfcx}, $\mathsf{e}(w)$ is entire and
$\mathsf{u}(y;\theta)>0$ for all real $y$. Since $\sigma$ is bounded away
from zero on $K$, the map
$\theta\mapsto w_{y,\theta}$ is smooth on $K$. Hence
$\theta\mapsto f_Y(y;\theta)$ and $\theta\mapsto\ell(y;\theta)$ are twice
continuously differentiable on $K$.

It remains to establish the uniform bounds. The explicit score and Hessian
formulas in Lemma \ref{lem:ScoreHessian} express all entries of
$s(y;\theta)$ and $H(y;\theta)$ as continuous functions of
$\tilde y$, $\sigma$, $\gamma$, $\mathsf{u}(y;\theta)$, and
$\mathsf{v}(y;\theta)$, with denominators involving only powers of
$\sigma$ and $\mathsf{u}(y;\theta)$. On any set with $|\tilde y|\leq R$,
these expressions are uniformly bounded over $K$ because
$K$ is compact, $\sigma$ is bounded away from zero, and
$\mathsf{u}(y;\theta)>0$.

For the tails, use the standard large-argument expansion of
$\operatorname{erfcx}$ on the Voigt line. Uniformly for
$\theta\in K$, as $|\tilde y|\rightarrow\infty$,
$$
\mathsf{u}(y;\theta) = \sqrt{\tfrac{2}{\pi}}\sigma\gamma|\tilde y|^{-2}\{1+O(|\tilde y|^{-2})\},
$$
and
$$
\frac{\mathsf{v}(y;\theta)}{\mathsf{u}(y;\theta)}
= -\frac{\tilde y}{\gamma} + \frac{2\sigma^2}{\gamma\tilde y} + O(|\tilde y|^{-3}).
$$
Substituting these expansions into the score formulas in
Lemma \ref{lem:ScoreHessian} gives, uniformly over $K$,
$$
s_\mu(y;\theta)=\frac{2}{\tilde y}+O(|\tilde y|^{-3}),
\qquad 
s_\gamma(y;\theta)=\frac{1}{\gamma}-\frac{2\gamma}{\tilde y^2}+O(|\tilde y|^{-4}),
$$
and
$s_\sigma(y;\theta)=6\sigma\tilde y^{-2}+O(|\tilde y|^{-4})$.
Since $\gamma$ and $\sigma$ are bounded above and away from zero on
$K$, the score is uniformly bounded in the tails.

The Hessian bounds follow in the same way from the Hessian formulas in
Lemma \ref{lem:ScoreHessian}. The corresponding tail expansions are, uniformly
over $K$,
$$
H_{\mu\mu}(y;\theta)=O(|\tilde y|^{-2}),
\qquad
H_{\gamma\gamma}(y;\theta)=-\frac{1}{\gamma^2}+O(|\tilde y|^{-2}),
$$
$$
H_{\mu\gamma}(y;\theta)=O(|\tilde y|^{-3}),
\qquad
H_{\mu\sigma}(y;\theta)=O(|\tilde y|^{-2}),
$$
$$
H_{\gamma\sigma}(y;\theta)=O(|\tilde y|^{-2}),
\qquad
H_{\sigma\sigma}(y;\theta)=O(|\tilde y|^{-1}).
$$
Thus all Hessian entries are uniformly bounded in the tails. Combining the
boundedness on compact sets in $\tilde y$ with the tail bounds proves that
there is a finite constant $C_{K}$ such that
$$
\sup_{\theta\in K}\|s(y;\theta)\|
\leq C_{K},
\qquad
\sup_{\theta\in K}\|H(y;\theta)\|
\leq C_{K},\quad \text{for all}\quad y\in\mathbb R.
$$

Finally, we establish the density-derivative envelopes. From the convolution
representation and the compactness of $K$, there is a constant $C_K<\infty$
such that
$$
\sup_{\theta\in K} f_Y(y;\theta)
\leq
\frac{C_K}{1+y^2}.
$$
Indeed, the Cauchy component controls the tails uniformly over $K$, while the
Gaussian component is uniformly bounded in scale on $K$. Since $\frac{\partial f_Y(y;\theta)}{\partial\theta_j}=f_Y(y;\theta)s_j(y;\theta)$, 
and
$$
\frac{\partial^2 f_Y(y;\theta)}
{\partial\theta_i\partial\theta_j}
=
f_Y(y;\theta)
\left[
H_{ij}(y;\theta)+s_i(y;\theta)s_j(y;\theta)
\right],
$$
the uniform bounds on $s$ and $H$ imply that
$$
\sup_{\theta\in K}
\left|
\frac{\partial f_Y(y;\theta)}{\partial\theta_i}
\right|
+
\sup_{\theta\in K}
\left|
\frac{\partial^2 f_Y(y;\theta)}
{\partial\theta_i\partial\theta_j}
\right|
\leq
\frac{C_K}{1+y^2}.
$$
The right-hand side is integrable on $\mathbb R$, so differentiation under the
integral sign is valid on $K$. The stated moment bounds follow
from the uniform bounds on $s$ and $H$.
\end{proof}

\begin{lemma}[Positive definiteness of the Fisher information]
\label{lem:non-singular-I}
Suppose Assumption \ref{assu:Compact} is satisfied and
$\theta_0=(\mu_0,\sigma_0,\gamma_0)^\prime\in\Theta$. Let
$\mathcal{I}_{\theta_0}=\mathbb{E}_{\theta_0}
\left[s(Y;\theta_0)s(Y;\theta_0)^\prime\right]$ 
be the per-observation Fisher information matrix. 
Then $\mathcal{I}_{\theta_0}$ is positive definite.
\end{lemma}

\begin{proof}
Since $a^\prime\mathcal{I}_{\theta_0}a=\mathbb{E}_{\theta_0}
\left[\{a^\prime s(Y;\theta_0)\}^2\right]\geq0$ for any $a\in\mathbb{R}^3$, it follows that $\mathcal{I}_{\theta_0}$ is positive semidefinite, and to prove positive
definiteness, it suffices to show that
$a^\prime s(y;\theta_0)=0$ for 
$\mathbb{P}_{\theta_0}$-almost all $y$ implies $a=0$.

Consider
$$ 
L(y) = a_1 s_\mu(y;\theta_0) +a_2 s_\sigma(y;\theta_0) + a_3 s_\gamma(y;\theta_0).
$$
Because the score components are continuous in $y$ and the Voigt density is
strictly positive on $\mathbb{R}$, $L(Y)=0$ a.s.-$\mathbb{P}_{\theta_0}$
implies $L(y)=0$ for all $y\in\mathbb{R}$. Otherwise, continuity would imply
that $L$ is nonzero on an open interval with positive probability.

Write $\tilde y=y-\mu_0$. From the score formulas in Lemma
\ref{lem:ScoreHessian}, $s_\mu(y;\theta_0)$ is odd in $\tilde y$, whereas
$s_\sigma(y;\theta_0)$ and $s_\gamma(y;\theta_0)$ are even.
Evaluating $L(y)=0$ at $y_+=\mu_0+r$ and $y_-=\mu_0-r$ gives
\begin{align*}
a_1s_\mu(y_+;\theta_0)
+
a_2s_\sigma(y_+;\theta_0)
+
a_3s_\gamma(y_+;\theta_0)
&=0,
\\
-a_1s_\mu(y_+;\theta_0)
+
a_2s_\sigma(y_+;\theta_0)
+
a_3s_\gamma(y_+;\theta_0)
&=0.
\end{align*}
Subtracting the two expressions yields
$2a_1s_\mu(y_+;\theta_0)=0$ for all $r$. Since $s_\mu(y;\theta_0)$ is not identically zero, it follows
that $a_1=0$.

With $a_1=0$, we have
$ a_2s_\sigma(y;\theta_0)+a_3s_\gamma(y;\theta_0)=0$
for all $y$. From the tail expansions in Lemma
\ref{lem:MLE-derivative-envelope}, 
$ s_\gamma(y;\theta_0)\rightarrow 1/\gamma_0$ and 
$s_\sigma(y;\theta_0)\rightarrow 0$, as $|\tilde y|\rightarrow\infty$.
Taking the limit, 
$\frac{a_3}{\gamma_0}=0$, proves $a_3=0$ since $\gamma_0>0$ under Assumption \ref{assu:Compact}.

The reduced identity is now
$a_2s_\sigma(y;\theta_0)=0$ for all $y$, and since 
$s_\sigma(y;\theta_0) = 6\sigma_0\tilde y^{-2} +O(|\tilde y|^{-4})$ is not identically zero ($\sigma_0>0$) under Assumption \ref{assu:Compact}, we can also conclude $a_2=0$. This completes the proof.
\end{proof}

\noindent{}\textbf{Proof of Theorem \ref{thm:MLE-consistent-asN}.}
Let
$\ell(y;\theta)=\log f_Y(y;\theta)$,
$Q_n(\theta)=\frac{1}{n}\sum_{i=1}^n\ell(Y_i;\theta)$, and
$Q(\theta)=\mathbb{E}_{\theta_0}[\ell(Y;\theta)]$. 
We first prove consistency. By Lemma \ref{lem:MLE-envelope},
$\ell(Y;\theta)$ has an integrable envelope uniformly over
$\theta\in\Theta$, and by continuity of the Voigt density,
$\ell(y;\theta)$ is continuous in $\theta$ on $\Theta$ for each fixed $y$.
Since $\Theta$ is compact, the uniform law of large numbers gives
$
\sup_{\theta\in\Theta}|Q_n(\theta)-Q(\theta)|
\overset{p}{\rightarrow}0$.
Moreover,
$$
Q(\theta)-Q(\theta_0)=\mathbb{E}_{\theta_0}
\left[\log\frac{f_Y(Y;\theta)}{f_Y(Y;\theta_0)}\right]=-\operatorname{KL}
\left(f_Y(\cdot;\theta_0),f_Y(\cdot;\theta)\right)\leq0.
$$
By Lemma \ref{lem:MLE-identification}, equality holds only when
$\theta=\theta_0$. Hence $Q(\theta)$ is uniquely maximized at
$\theta_0$. The argmax theorem therefore implies
$\hat\theta_n\overset{p}{\rightarrow}\theta_0$.

We next prove asymptotic normality. Suppose that $\theta_0$ is an interior
point of $\Theta$. Let
$K$ be a compact convex neighborhood of $\theta_0$ contained in the interior
of the parameter space. By consistency, $\Pr(\hat\theta_n\in K)\rightarrow 1$. 
On this event, the first-order condition for the MLE is
$0 = \frac{1}{n}\sum_{i=1}^n s(Y_i;\hat\theta_n)$, 
where 
$s(y;\theta)=\frac{\partial\ell(y;\theta)}{\partial\theta}$.
A Taylor expansion around $\theta_0$ gives
$$
0
=
\frac{1}{\sqrt n}\sum_{i=1}^n s(Y_i;\theta_0)
+
A_n\sqrt n(\hat\theta_n-\theta_0),
$$
where
$$
A_n
=
\int_0^1
\frac{1}{n}\sum_{i=1}^n
H\{Y_i;\theta_0+r(\hat\theta_n-\theta_0)\}
dr
$$
and $H(y;\theta)=\frac{\partial^2\ell(y;\theta)}{\partial\theta\partial\theta^\prime}$. 
By Lemma \ref{lem:MLE-derivative-envelope}, the score has finite second
moments and the Hessian is uniformly bounded on $K$. The same lemma provides
integrable envelopes for the first and second derivatives of $f_Y(y;\theta)$
with respect to $\theta$, uniformly on $K$. Therefore differentiation under
the integral sign is valid on $K$. Differentiating
$\int f_Y(y;\theta)dy=1$ at $\theta=\theta_0$ gives
$\mathbb{E}_{\theta_0}[s(Y;\theta_0)]=0$. Differentiating a second time gives the information identity
$$
\mathbb{E}_{\theta_0}\left[H(Y;\theta_0)+s(Y;\theta_0)s(Y;\theta_0)^\prime\right]=0,
$$
such that $\mathcal{I}_{\theta_0}=\mathbb{E}_{\theta_0}\left[s(Y;\theta_0)s(Y;\theta_0)^\prime\right]=-\mathbb{E}_{\theta_0}\left[H(Y;\theta_0)\right]$.
The multivariate central limit theorem yields
$$
\frac{1}{\sqrt n}\sum_{i=1}^n s(Y_i;\theta_0)
\overset{d}{\rightarrow}
\mathcal{N}(0,\mathcal{I}_{\theta_0}).
$$
The uniform law of large numbers applied to the Hessian, together with
$\hat\theta_n\overset{p}{\rightarrow}\theta_0$, gives
$A_n\overset{p}{\rightarrow} \mathbb{E}_{\theta_0}[H(Y;\theta_0)] = -\mathcal{I}_{\theta_0}$,
which is invertible by Lemma \ref{lem:non-singular-I}. Therefore,
$$
\sqrt n(\hat\theta_n-\theta_0)
= -A_n^{-1}\frac{1}{\sqrt n}\sum_{i=1}^n s(Y_i;\theta_0) + o_p(1),
$$
and Slutsky's theorem implies
$\sqrt n(\hat\theta_n-\theta_0)
\overset{d}{\rightarrow}
\mathcal{N}(0,\mathcal{I}_{\theta_0}^{-1})$.

It remains only to note the stated zero restrictions in the Fisher information.
At $\theta_0$, write $\tilde y=y-\mu_0$. The Voigt density is even in
$\tilde y$. From the score formulas in Lemma \ref{lem:ScoreHessian},
$s_\mu(y;\theta_0)$ is odd in $\tilde y$, while
$s_\sigma(y;\theta_0)$ and $s_\gamma(y;\theta_0)$ are even in $\tilde y$.
Hence the products $s_\mu s_\sigma$ and $s_\mu s_\gamma$ are odd, and their
expectations under the symmetric density are zero. Therefore
$\mathcal{I}_{\mu_0\sigma_0} = \mathcal{I}_{\mu_0\gamma_0}=0$. 
This completes the proof.\hfill$\square$

%%%%%%%%%%%%%%%%%%%%%%%%%%%%%%%%%%%%%%%%%%%%
\subsection{GCC Filter Results}

\noindent{}\textbf{Proof of Theorem \ref{thm:GCCfilter}.} 
Let $\xi_t=x_t-x_{t|t-1}$. Under Assumption \ref{ass:Masreliez}, conditionally on $\mathcal F_{t-1}$, $\xi_t\sim\mathcal{N}(0,h_{t|t-1})$ and is independent of the measurement error $\eta_t$. Since $\eta_t\sim\mathcal{V}(0,\sigma,\gamma)$, we may write $\eta_t=Z_t+C_t$, where $Z_t\sim\mathcal{N}(0,\sigma^2)$ and $C_t\sim\operatorname{Cauchy}(0,\gamma)$ are independent. Hence
$e_t = \xi_t+\eta_t = (\xi_t+Z_t)+C_t$,
where $\xi_t+Z_t\sim\mathcal{N}(0,h_{t|t-1}+\sigma^2)$. Therefore
$$
e_t|\mathcal{F}_{t-1}\sim\mathcal{V}(0,\delta_t,\gamma),
\qquad
\delta_t^2=h_{t|t-1}+\sigma^2.
$$
We now derive the moment update. Fix $t$ and condition on $\mathcal F_{t-1}$. Under Assumption \ref{ass:Masreliez}, the prediction error has the Gaussian-convolution form
$e_t=\xi_t+\eta_t$ 
with $\xi_t$ independent of $\eta_t$ and $\xi_t|\mathcal F_{t-1}\sim\mathcal N(0,h_{t|t-1})$. Let $f_{t-1}$ denote the conditional density of $e_t$ given $\mathcal F_{t-1}$. Applying Proposition \ref{prop:TweedieGaussian} conditionally, with $Z=\xi_t$, $X=\eta_t$, $Y=e_t$, and $\sigma^2=h_{t|t-1}$, gives
$$
\mathbb{E}[\xi_t|e_t,\mathcal F_{t-1}]
=
-h_{t|t-1}\frac{\partial}{\partial e_t}\log f_{t-1}(e_t).
$$
Since $e_t=y_t-x_{t|t-1}$ and $x_{t|t-1}$ is $\mathcal F_{t-1}$-measurable, conditioning on $(e_t,\mathcal F_{t-1})$ is the same as conditioning on $\mathcal F_t$, and we obtain
$$
x_{t|t}=x_{t|t-1}+h_{t|t-1}\psi_t\quad\text{with}\quad 
\psi_t\equiv-\frac{\partial}{\partial e_t}\log f_{t-1}(e_t).
$$
Similarly, the variance identity in Proposition \ref{prop:TweedieGaussian} gives
$$
\operatorname{var}(\xi_t|e_t,\mathcal F_{t-1})
= h_{t|t-1} + h_{t|t-1}^2
\frac{\partial^2}{\partial e_t^2}\log f_{t-1}(e_t),
$$
and from $x_t=x_{t|t-1}+\xi_t$, it follows that
$$
h_{t|t} = h_{t|t-1} - h_{t|t-1}^2\psi_t^{\prime}
\quad\text{where}\quad
\psi_t^{\prime}\equiv\frac{\partial\psi_t}{\partial e_t}
=-\frac{\partial^2}{\partial e_t^2}\log f_{t-1}(e_t).
$$
It remains only to identify the explicit expressions for $\psi_t$ and $\psi_t^\prime$. Since $e_t|\mathcal F_{t-1}\sim \mathcal{V}(0,\delta_t,\gamma)$, Lemma \ref{lem:ScoreHessian} applied to the prediction-error density gives
$$
\psi_t
=
-\frac{\partial\log f_{t-1}(e_t)}{\partial e_t}
=
\frac{1}{\delta_t^2}
\left(
e_t+\gamma\frac{\mathsf{v}_t}{\mathsf{u}_t}
\right),
$$
where $\mathsf{u}_t=\mathsf{u}(e_t;0,\delta_t,\gamma)$ and $\mathsf{v}_t=\mathsf{v}(e_t;0,\delta_t,\gamma)$. Differentiating this expression with respect to $e_t$ yields
$$
\psi_t^\prime
=
\frac{1}{\delta_t^4}
\left[
\delta_t^2
+
\gamma^2
\left(
1+\frac{\mathsf{v}_t^2}{\mathsf{u}_t^2}
\right)
-
\sqrt{\frac{2}{\pi}}
\frac{\gamma\delta_t}{\mathsf{u}_t}
\right].
$$
This proves the stated recursions.
\hfill$\square$\medskip{}

\noindent{}\textbf{Proof of Lemma \ref{lem:GCCvariance}.}
By Theorem \ref{thm:GCCfilter}, under the Gaussian prediction approximation the variance update is $h_{t|t}=\operatorname{var}(x_t-x_{t|t-1}|e_t,\mathcal F_{t-1})$. Equivalently, applying Proposition \ref{prop:TweedieGaussian} conditionally with Gaussian variance $h_{t|t-1}>0$ gives $h_{t|t}=h_{t|t-1}-h_{t|t-1}^2\psi_t^\prime$. The conditional density of $x_t-x_{t|t-1}$ given $e_t$ is proportional to the product of a nondegenerate Gaussian density and a strictly positive Voigt density, and therefore is not concentrated at a point. Hence $h_{t|t}>0$.

The prediction recursion then gives
$h_{t+1|t}=\phi^2h_{t|t}+\tau^2\geq\tau^2$, with strict inequality whenever
$\phi\neq0$. Since $\tau^2>0$, we have $h_{t+1|t}>0$, and therefore
$\delta_{t+1}^2=h_{t+1|t}+\sigma^2>0$. Thus the Voigt prediction-error density
remains well defined throughout the recursion.
\hfill$\square$\medskip{}

%%%%%%%%%%%%%%%%%%%%%%%%%%%%%%%%%%
\setcounter{figure}{0}
\global\long\def\thefigure{B.\arabic{figure}}%
\setcounter{table}{0}
\global\long\def\thetable{B.\arabic{table}}%

\section{Additional Details}
\label{app:additional-details}
\subsection{Conditional Density}

Figure \ref{fig:Conditional-density} shows the conditional density
for selected values of $y$ with $\sigma=\gamma=1$. The vertical
lines along the $x$-axis indicate the conditional expected value
for each conditional density.

For the special case $y=\mu$, the conditional density is
 $$
 f_{Z|Y}(z|y=\mu)=\frac{1}{\gamma\pi}\frac{1}{\Pr(|Z|\geq\tfrac{\gamma}{\sigma})}\frac{\exp(-\frac{1}{2}\frac{z^{2}+\gamma^{2}}{\sigma^{2}})}{1+(\frac{z}{\gamma})^{2}}.
 $$

\begin{figure}[t]
\spacingset{1.15}
\begin{centering}
\includegraphics[width=0.8\textwidth]{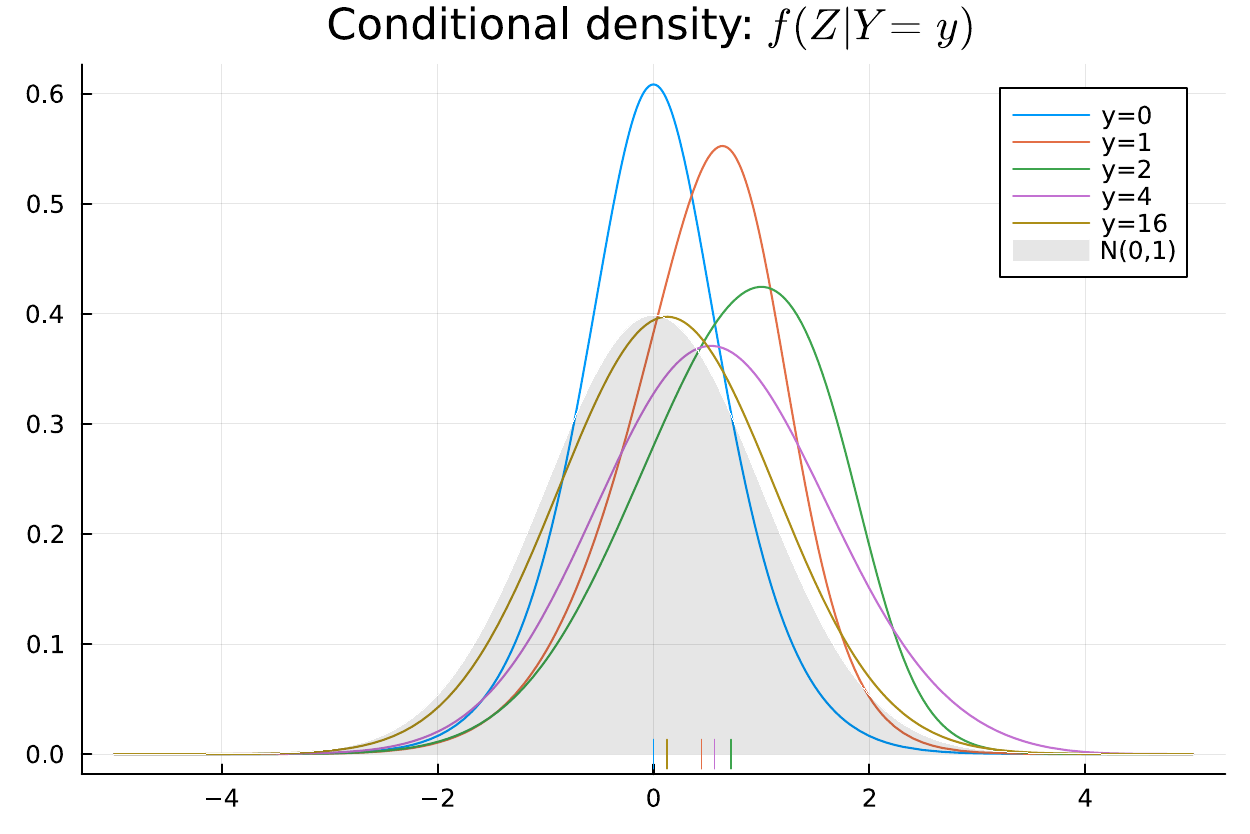}
\par\end{centering}
\caption{{\footnotesize Conditional density of $Z|Y=y$, for selected values
of $y$, where $Y=Z+X$ with $Z\sim \mathcal{N}(0,1)$ and $X\sim \operatorname{Cauchy}(0,1)$
independent. The corresponding conditional expectations, $\mathbb{E}[Z|Y=y]$,
are indicated with vertical lines. The unconditional density of $Z$
(standard normal) is represented with the shaded area. }\label{fig:Conditional-density}}
\end{figure}

\subsection{Higher-Order Conditional Cumulants}\label{sec:higher-order-cumulants}

Higher-order conditional cumulants can be obtained in the same way. The first step is not specific to the Voigt model, but holds for any Gaussian convolution. Let $K(s|y)=\log\mathbb{E}\left[\exp(sZ)|Y=y\right]$
denote the cumulant generating function of $Z|Y=y$. If $Y=\mu+Z+X$, where $Z\sim\mathcal{N}(0,\sigma^2)$ is independent of an arbitrary proper random variable $X$, then
$$
K(s|y)=\frac{\sigma^2s^2}{2}+\log f_Y(y-\sigma^2s)-\log f_Y(y).
$$
The conditional cumulants are defined by $\kappa_r(y)=\left.\frac{\partial^r}{\partial s^r}K(s|y)\right|_{s=0}$. Therefore
$$
\kappa_1(y)=-\sigma^2\frac{\partial}{\partial y}\log f_Y(y),
\qquad
\kappa_2(y)=\sigma^2+\sigma^4\frac{\partial^2}{\partial y^2}\log f_Y(y).
$$
These are the conditional mean and variance identities in Proposition \ref{prop:TweedieGaussian}. For $r\geq3$, $\kappa_r(y)=(-\sigma^2)^r\frac{\partial^r}{\partial y^r}\log f_Y(y)$, 
and equivalently,  $\kappa_{r+1}(y)=-\sigma^2\frac{\partial}{\partial y}\kappa_r(y)$, for all $r\geq2$.

For the Voigt Gauss-Cauchy convolution, Proposition \ref{prop:Voigt} gives $f_Y$ explicitly in terms of $\mathsf{u}(y;\theta)$. Hence all higher-order conditional cumulants are obtained by differentiating the closed-form variance expression in Corollary \ref{cor:CondExpect}. By Lemma \ref{lem:erfcx}, these derivatives are again algebraic functions of $\mathsf{u}(y;\theta)$, $\mathsf{v}(y;\theta)$, and the parameters.

The tractability is therefore not limited to the conditional mean and variance. Although the Voigt random variable $Y$ has Cauchy tails and hence does not have finite positive integer moments when $\gamma>0$, the conditional distribution of the latent Gaussian component $Z|Y=y$ has all finite moments. Thus conditional skewness, kurtosis, and higher conditional moments of the latent Gaussian component can be evaluated without numerical integration.

\subsection{Relative precision of the Gaussian and Cauchy scale estimates}\label{app:relative_precision}

The asymptotic covariance matrix for $(\sigma,\gamma)$ is homogeneous of degree two. Hence the ratio of their asymptotic standard deviations depends only on $\lambda=\gamma/\sigma$. We define\footnote{The ratio $\lambda$ is proportional to the Voigt damping parameter used in astrophysics, $a=\lambda/\sqrt{2}$.}
$$
R(\lambda)\equiv \frac{\operatorname{aStd}(\sigma;\mu,\sigma,\lambda\sigma)}{\operatorname{aStd}(\gamma;\mu,\sigma,\lambda\sigma)}.
$$
Figure \ref{fig:seratio} reports the resulting asymptotic standard deviations and their ratio. The left panel shows the asymptotic standard deviations for $\sigma$ and $\gamma$ as functions of $\lambda=\gamma/\sigma$, under the normalization $(\sigma,\gamma)=(1,\lambda)$. Both axes are log transformed. The relationship is close to log-linear for $\gamma$, corresponding to an approximately constant elasticity, whereas the relationship for $\sigma$ is convex.
\begin{figure}[htb]
\spacingset{1.15}
\begin{centering}
\includegraphics[width=1\textwidth]{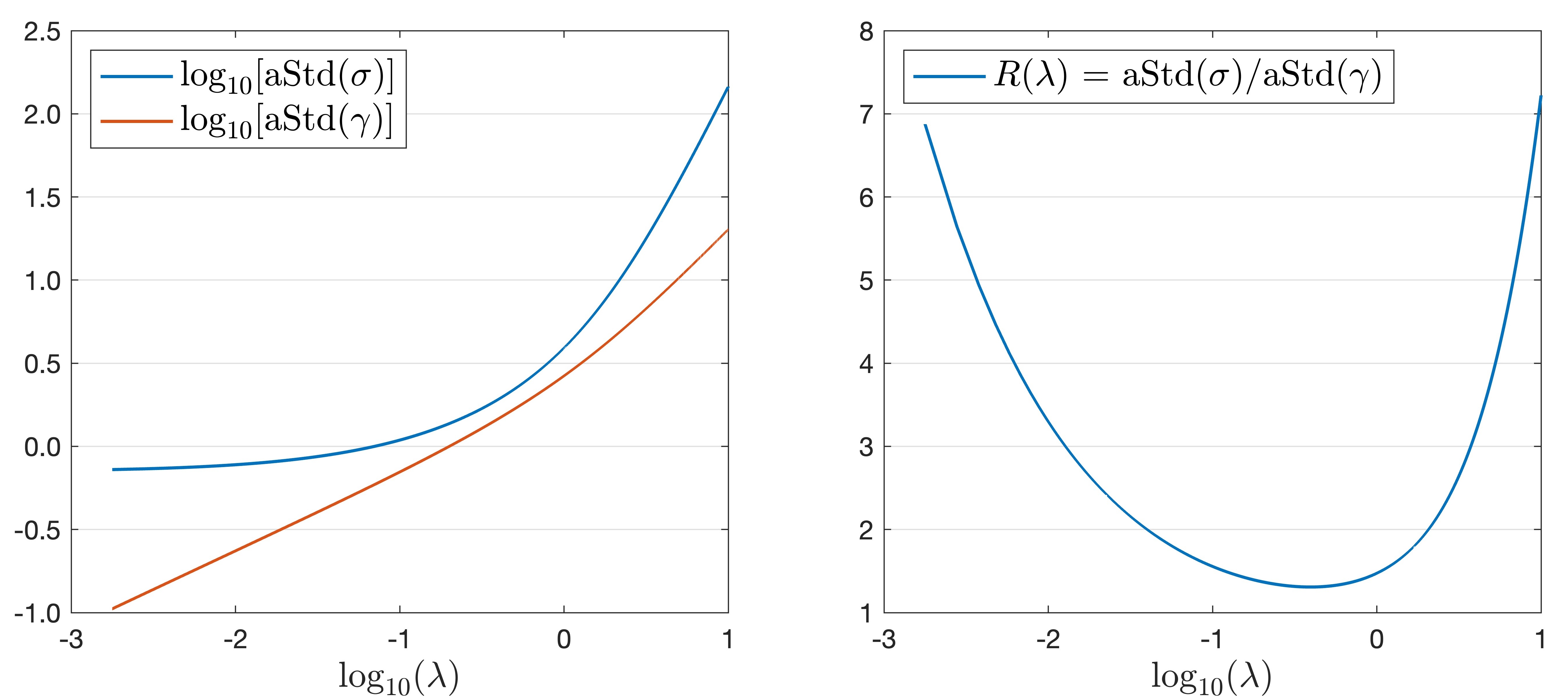}
\par \end{centering}
\caption{\footnotesize The asymptotic standard deviations for $\sigma$ and $\gamma$ are plotted against the ratio $\lambda=\gamma/\sigma$ in the left panel for $(\sigma,\gamma)=(1,\lambda)$. The right panel shows $R(\lambda)$, the ratio of the asymptotic standard deviation of $\hat\sigma$ to that of $\hat\gamma$. Values above one imply that the Cauchy scale $\gamma$ is estimated more precisely than the Gaussian scale $\sigma$.}
\label{fig:seratio}
\end{figure}

The right panel plots $R(\lambda)$ against $\log_{10}\lambda$. The ratio is larger than one throughout the range considered, so the Cauchy scale parameter $\gamma$ is asymptotically more precisely estimated than the Gaussian scale parameter $\sigma$. This may seem counterintuitive because $\gamma$ governs the heavy-tailed component. The reason is that the Cauchy scale is strongly identified by the frequency and magnitude of extreme observations, where the Gaussian component contributes little. By contrast, the Gaussian scale is identified primarily from the central part of the distribution, where changes in $\sigma$ and $\gamma$ can partially offset one another. The U-shape of $R(\lambda)$ shows that the relative precision is especially favorable to $\gamma$ when the Cauchy scale is either small or large relative to the Gaussian scale.

%%%%%%%%%%%%%%%%%%%%%%%%%%%%

\setcounter{table}{0}
\global\long\def\thetable{C.\arabic{table}}%

\section{Simulation Details and Additional Monte Carlo Results}

\subsection{Finite-sample behavior of the Voigt MLE}\label{app:voigt_mle_mc}

We examine the finite-sample behavior of the Voigt maximum likelihood estimator using $N=100,000$ independent datasets for sample sizes ranging from $n=100$ to $n=10,000$. The simulations cover three values of the Cauchy-to-Gaussian scale ratio, $\lambda=\gamma_0/\sigma_0$, corresponding to $\lambda=0.01$, $\lambda=0.1$, and $\lambda=1$. Thus the designs range from a nearly Gaussian case to a balanced Voigt case in which the Gaussian and Cauchy scales are equal.

Table \ref{tab:simulation} reports the mean of the parameter estimates, their empirical standard deviation (Std), and the asymptotic standard deviation (aStd) derived from the Fisher information matrix. For each parameter component, we also report the empirical left and right tail probabilities
$$
\alpha_{L,j}=\frac{1}{N}\sum_{i=1}^{N}\left[\frac{\hat{\theta}_{n,j}^{(i)}-\theta_{0,j}}{\operatorname{Std}(\hat{\theta}_{n,j})}<-1.96\right],\qquad
\alpha_{R,j}=\frac{1}{N}\sum_{i=1}^{N}\left[\frac{\hat{\theta}_{n,j}^{(i)}-\theta_{0,j}}{\operatorname{Std}(\hat{\theta}_{n,j})}>1.96\right],
$$
where $j$ indexes the components of $\theta=(\mu,\sigma,\gamma)^\prime$. Under asymptotic normality, both probabilities should approach $0.025$ as $n\rightarrow\infty$.
\begin{table}[htbp]
\spacingset{1.15}
\caption{Simulation results for the Voigt density MLE}\label{tab:simulation}
\begin{centering}
\begin{footnotesize}
\begin{tabularx}{\textwidth}{Yp{-0.5cm}YYYp{-0.5cm}YYYp{-0.5cm}YYY @{}}
\toprule
    \midrule
          &       & $\mu$    & $\sigma$ & $\gamma$ &       & $\mu$    & $\sigma$ & $\gamma$ &       & $\mu$    & $\sigma$ & $\gamma$ \\
    \midrule     
    \multicolumn{13}{c}{Panel A: $(\mu_0,\sigma_0, \gamma_0)= (1,1,0.01)$} \\   
          &       & \multicolumn{3}{c}{$n=100$} &       & \multicolumn{3}{c}{$n=200$} &       & \multicolumn{3}{c}{$n=400$} \\
    Mean  &       & 0.9999 & 0.9894 & 0.0112 &       & 1.0001 & 0.9952 & 0.0102 &       & 1.0001 & 0.9980 & 0.0100 \\
    Std   &       & 0.1016 & 0.0791 & 0.0239 &       & 0.0715 & 0.0552 & 0.0164 &       & 0.0509 & 0.0388 & 0.0117 \\
    aStd  &       & 0.1013 & 0.0775 & 0.0235 &       & 0.0716 & 0.0548 & 0.0166 &       & 0.0507 & 0.0388 & 0.0118 \\
    $\alpha_{L}$ &       & 0.0250 & 0.0166 & 0.0588 &       & 0.0250 & 0.0199 & 0.0570 &       & 0.0258 & 0.0223 & 0.0487 \\
    $\alpha_{R}$ &       & 0.0249 & 0.0353 & 0.0000 &       & 0.0253 & 0.0307 & 0.0000 &       & 0.0244 & 0.0290 & 0.0000 \\
          &       & \multicolumn{3}{c}{$n=1,000$} &       & \multicolumn{3}{c}{$n=4,000$} &       & \multicolumn{3}{c}{$n=10,000$} \\
    Mean  &       & 1.0001 & 0.9995 & 0.0099 &       & 0.9999 & 0.9998 & 0.0100 &       & 0.9999 & 0.9999 & 0.0100 \\
    Std   &       & 0.0320 & 0.0247 & 0.0075 &       & 0.0160 & 0.0122 & 0.0037 &       & 0.0101 & 0.0077 & 0.0024 \\
    aStd  &       & 0.0320 & 0.0245 & 0.0074 &       & 0.0160 & 0.0123 & 0.0037 &       & 0.0101 & 0.0078 & 0.0024 \\
    $\alpha_{L}$ &       & 0.0250 & 0.0233 & 0.0384 &       & 0.0252 & 0.0241 & 0.0327 &       & 0.0246 & 0.0237 & 0.0301 \\
    $\alpha_{R}$ &       & 0.0241 & 0.0261 & 0.0000 &       & 0.0246 & 0.0258 & 0.0146 &       & 0.0254 & 0.0258 & 0.0189 \\
\\[-0.4cm]    
\hdashline
\\[-0.4cm]
    \multicolumn{13}{c}{Panel B: $(\mu_0,\sigma_0, \gamma_0)= (1,1,0.1)$} \\  
          &       & \multicolumn{3}{c}{$n=100$} &       & \multicolumn{3}{c}{$n=200$} &       & \multicolumn{3}{c}{$n=400$} \\
    Mean  &       & 1.0002 & 0.9935 & 0.0972 &       & 0.9999 & 0.9970 & 0.0986 &       & 1.0003 & 0.9983 & 0.0994 \\
    Std   &       & 0.1122 & 0.1156 & 0.0717 &       & 0.0793 & 0.0795 & 0.0505 &       & 0.0558 & 0.0552 & 0.0351 \\
    aStd  &       & 0.1112 & 0.1090 & 0.0701 &       & 0.0786 & 0.0771 & 0.0496 &       & 0.0556 & 0.0545 & 0.0350 \\
    $\alpha_{L}$ &       & 0.0249 & 0.0202 & 0.0376 &       & 0.0252 & 0.0227 & 0.0346 &       & 0.0250 & 0.0223 & 0.0322 \\
    $\alpha_{R}$ &       & 0.0254 & 0.0315 & 0.0000 &       & 0.0252 & 0.0294 & 0.0282 &       & 0.0247 & 0.0287 & 0.0161 \\
          &       & \multicolumn{3}{c}{$n=1,000$} &       & \multicolumn{3}{c}{$n=4,000$} &       & \multicolumn{3}{c}{$n=10,000$} \\
    Mean  &       & 1.0000 & 0.9994 & 0.0998 &       & 1.0000 & 0.9999 & 0.0999 &       & 1.0001 & 0.9999 & 0.1000 \\
    Std   &       & 0.0351 & 0.0347 & 0.0223 &       & 0.0176 & 0.0173 & 0.0111 &       & 0.0111 & 0.0109 & 0.0070 \\
    aStd  &       & 0.0352 & 0.0345 & 0.0222 &       & 0.0176 & 0.0172 & 0.0111 &       & 0.0111 & 0.0109 & 0.0070 \\
    $\alpha_{L}$ &       & 0.0257 & 0.0235 & 0.0297 &       & 0.0249 & 0.0244 & 0.0272 &       & 0.0251 & 0.0238 & 0.0272 \\
    $\alpha_{R}$ &       & 0.0244 & 0.0270 & 0.0194 &       & 0.0254 & 0.0260 & 0.0231 &       & 0.0253 & 0.0260 & 0.0234 \\
\\[-0.4cm]
\hdashline
\\[-0.4cm]
    \multicolumn{13}{c}{Panel C: $(\mu_0,\sigma_0, \gamma_0)= (1,1,1)$} \\  
          &       & \multicolumn{3}{c}{$n=100$} &       & \multicolumn{3}{c}{$n=200$} &       & \multicolumn{3}{c}{$n=400$} \\
    Mean  &       & 0.9996 & 0.9422 & 0.9769 &       & 0.9992 & 0.9642 & 0.9929 &       & 0.9992 & 0.9846 & 0.9966 \\
    Std   &       & 0.2138 & 0.4571 & 0.2576 &       & 0.1492 & 0.3250 & 0.1883 &       & 0.1046 & 0.2133 & 0.1338 \\
    aStd  &       & 0.2088 & 0.3910 & 0.2653 &       & 0.1477 & 0.2765 & 0.1876 &       & 0.1044 & 0.1955 & 0.1326 \\
    $\alpha_{L}$ &       & 0.0249 & 0.0103 & 0.0206 &       & 0.0245 & 0.0086 & 0.0263 &       & 0.0248 & 0.0118 & 0.0272 \\
    $\alpha_{R}$ &       & 0.0255 & 0.0910 & 0.0275 &       & 0.0258 & 0.0491 & 0.0232 &       & 0.0256 & 0.0357 & 0.0228 \\
          &       & \multicolumn{3}{c}{$n=1,000$} &       & \multicolumn{3}{c}{$n=4,000$} &       & \multicolumn{3}{c}{$n=10,000$} \\
    Mean  &       & 1.0003 & 0.9952 & 0.9985 &       & 1.0000 & 0.9992 & 0.9994 &       & 1.0000 & 0.9993 & 0.9999 \\
    Std   &       & 0.0662 & 0.1265 & 0.0841 &       & 0.0330 & 0.0620 & 0.0418 &       & 0.0209 & 0.0393 & 0.0266 \\
    aStd  &       & 0.0660 & 0.1236 & 0.0839 &       & 0.0330 & 0.0618 & 0.0419 &       & 0.0209 & 0.0391 & 0.0265 \\
    $\alpha_{L}$ &       & 0.0251 & 0.0180 & 0.0261 &       & 0.0249 & 0.0218 & 0.0254 &       & 0.0244 & 0.0224 & 0.0253 \\
    $\alpha_{R}$ &       & 0.0253 & 0.0328 & 0.0244 &       & 0.0252 & 0.0278 & 0.0244 &       & 0.0251 & 0.0277 & 0.0251 \\
\midrule
\bottomrule
\end{tabularx}
\end{footnotesize}

\par\end{centering}
{\footnotesize Note: The simulation results are based on $N=100,000$ Monte Carlo replications. We report the mean and standard deviation (Std) of the estimated parameters. For comparison, the asymptotic standard deviations (aStd), computed from the Fisher information, are also included. We also report the empirical left and right tail probabilities, $\alpha_{L}$ and $\alpha_{R}$; the nominal benchmark for each one-sided tail probability is $0.025$.}{\small\par}
\end{table}

The simulation results support the asymptotic approximation. The estimator of $\mu$ is essentially unbiased across all designs and sample sizes, and its empirical standard deviation is close to the asymptotic standard deviation even for small samples. The scale parameters also behave well once the sample size is moderate. The main finite-sample distortions occur in the smaller samples and involve the scale parameters. In the near-Gaussian design, $\lambda=0.01$, the Cauchy scale is close to the lower boundary of the parameter space, so the finite-sample distribution of $\hat\gamma$ is visibly skewed even though the magnitude of its standard deviation is well approximated by the Fisher information. In the balanced design, $\lambda=1$, the estimator of $\sigma$ exhibits some small-sample skewness, reflecting the difficulty of separating the Gaussian core from the Cauchy tails with limited data. These distortions diminish rapidly as $n$ increases.

Overall, the table indicates that the Voigt MLE is well centered and that the Fisher-information approximation provides an accurate description of sampling variability. The remaining discrepancies are finite-sample effects associated with weak identification of one of the two scale components, either because $\gamma_0$ is close to zero or because the sample is too small to clearly distinguish the Gaussian and Cauchy contributions.

\subsection{Finite-sample behavior of the GCC QMLE}\label{app:gcc_qmle_mc}

We examine the finite-sample behavior of the GCC QMLE through a Monte Carlo study. The purpose is to assess whether the inverse-information approximation gives an accurate description of the sampling distribution of the estimator under the Masreliez Gaussian prediction approximation. The simulations use $N=100,000$ replications and two designs: a small-Cauchy design with $\gamma_0/\sigma_0=0.1$ and a balanced design with $\gamma_0/\sigma_0=1$.

Table \ref{tab:GCC_QMLE_simulation} reports the mean and empirical standard deviation (Std) of the parameter estimates, together with the corresponding asymptotic standard deviations (aStd). It also reports the empirical left and right tail frequencies, $\alpha_L$ and $\alpha_R$, based on standardized estimates below $-1.96$ and above $1.96$.

\begin{table}[htbp]
\caption{Finite-Sample Behavior of the GCC QMLE}\label{tab:GCC_QMLE_simulation}
\spacingset{1.15}
\begin{centering}
\begin{footnotesize}
\begin{tabularx}{\textwidth}{Yp{-0.1cm}YYYYYp{-0.1cm}YYYYY}
\toprule
    \midrule
          &       & $\sigma$ & $\gamma$ & $\mu$    & $\phi$   & $\tau$ &       & $\sigma$ & $\gamma$ & $\mu$    & $\phi$   & $\tau$ \\
    \midrule    
    \multicolumn{13}{c}{Panel A: $(\sigma_0, \gamma_0,\mu_0,\phi_0,\tau_0)= (1,0.1,1,0.95,1)$} \\   
          &       & \multicolumn{5}{c}{$T=500$}             &       & \multicolumn{5}{c}{$T=1,000$} \\
    Mean  &       & 0.9889 & 0.0994 & 1.0007 & 0.9389 & 1.0055 &       & 0.9945 & 0.0998 & 0.9987 & 0.9446 & 1.0027 \\
    Std   &       & 0.0986 & 0.0376 & 0.8401 & 0.0190 & 0.0870 &       & 0.0684 & 0.0264 & 0.6151 & 0.0123 & 0.0611 \\
    aStd  &       & 0.0947 & 0.0368 & 0.8927 & 0.0158 & {0.0876} &       & 0.0670 & 0.0260 & 0.6313 & {0.0112} & {0.0620} \\
     $\alpha_{L}$ &       & 0.0128 & 0.0326 & 0.0244 & 0.0000 & 0.0298 &       & 0.0167 & 0.0312 & 0.0248 & 0.0005 & 0.0282 \\
     $\alpha_{R}$ &       & 0.0383 & 0.0149 & 0.0251 & 0.0928 & 0.0208 &       & 0.0349 & 0.0181 & 0.0247 & 0.0751 & 0.0213 \\
          &       & \multicolumn{5}{c}{$T=2,000$}             &       & \multicolumn{5}{c}{$T=4,000$} \\
    Mean  &       & 0.9973 & 0.0998 & 0.9997 & 0.9474 & 1.0013 &       & 0.9985 & 0.1000 & 1.0001 & 0.9487 & 1.0006 \\
    Std   &       & 0.0478 & 0.0187 & 0.4416 & 0.0083 & 0.0432 &       & 0.0335 & 0.0132 & 0.3133 & 0.0057 & 0.0303 \\
    aStd  &       & 0.0474 & 0.0184 & 0.4464 & 0.0079 & {0.0438} &       & 0.0335 & 0.0130 & 0.3156 & {0.0056} & {0.0310} \\
    $\alpha_{L}$ &       & 0.0189 & 0.0293 & 0.0251 & 0.0028 & 0.0267 &       & 0.0210 & 0.0277 & 0.0256 & 0.0069 & 0.0262 \\
    $\alpha_{R}$ &       & 0.0322 & 0.0205 & 0.0245 & 0.0607 & 0.0229 &       & 0.0303 & 0.0213 & 0.0250 & 0.0499 & 0.0240 \\
          &       & \multicolumn{5}{c}{$T=8,000$}             &       & \multicolumn{5}{c}{$T=16,000$} \\
    Mean  &       & 0.9991 & 0.1001 & 1.0000 & 0.9494 & 1.0003 &       & 0.9993 & 0.1001 & 0.9993 & 0.9497 & 1.0001 \\
    Std   &       & 0.0236 & 0.0093 & 0.2216 & 0.0040 & 0.0214 &       & 0.0167 & 0.0066 & 0.1572 & 0.0028 & 0.0151 \\
    aStd  &       & 0.0237 & 0.0092 & 0.2232 & 0.0039 & {0.0219} &       & 0.0167 & 0.0065 & 0.1578 & {0.0028} & {0.0155} \\    
    $\alpha_{L}$ &       & 0.0214 & 0.0273 & 0.0254 & 0.0114 & 0.0259 &       & 0.0213 & 0.0277 & 0.0250 & 0.0152 & 0.0255 \\
   $\alpha_{R}$ &       & 0.0281 & 0.0219 & 0.0244 & 0.0425 & 0.0243 &       & 0.0281 & 0.0229 & 0.0252 & 0.0377 & 0.0247 \\
%%%%%%%%%%%%%%%%%%%%%%%%%%%%%%%%%%%%%%%%%%%%%%%%%%%%%%%%%%%%%%%%%%%%%%%%%%%%%%%%%%%%
\\[-0.4cm]
\hdashline
\\[-0.4cm]
    \multicolumn{13}{c}{Panel B: $(\sigma_0, \gamma_0,\mu_0,\phi_0,\tau_0)= (1,1,1,0.95,1)$} \\   
          &       & \multicolumn{5}{c}{$T=500$}             &       & \multicolumn{5}{c}{$T=1,000$} \\
    Mean  &       & 0.9295 & 0.9938 & 1.0013 & 0.9377 & {1.0152} &       & 0.9620 & 0.9988 & 0.9989 & {0.9441} & {1.0084} \\
    Std   &       & 0.3602 & 0.1358 & 0.8455 & 0.0223 & {0.1308} &       & 0.2462 & 0.0970 & 0.6170 & {0.0141} & {0.0912} \\
    aStd  &       & 0.3067 & 0.1378 & 0.8929 & 0.0177 & {0.1301} &       & 0.2169 & 0.0974 & 0.6314 & {0.0125} & {0.0920} \\   
    $\alpha_{L}$ &       & 0.0024 & 0.0233 & 0.0246 & 0.0000 & {0.0360} &       & 0.0043 & 0.0263 & 0.0250 & {0.0003} & {0.0333} \\
    $\alpha_{R}$ &       & 0.0692 & 0.0258 & 0.0252 & 0.0882 & {0.0164} &       & 0.0472 & 0.0239 & 0.0248 & {0.0743} & {0.0182} \\
          &       & \multicolumn{5}{c}{$T=2,000$}             &       & \multicolumn{5}{c}{$T=4,000$} \\
    Mean  &       & 0.9801 & 0.9998 & 0.9993 & 0.9471 & {1.0049} &       & 0.9878 & 1.0004 & 1.0001 & {0.9485} & {1.0031} \\
    Std   &       & 0.1617 & 0.0688 & 0.4410 & 0.0094 & {0.0637} &       & 0.1108 & 0.0488 & 0.3113 & {0.0065} & {0.0448} \\
    aStd  &       & 0.1534 & 0.0689 & 0.4465 & 0.0088 & {0.0651} &       & 0.1084 & 0.0487 & 0.3157 & {0.0062} & {0.0460} \\   
    $\alpha_{L}$ &       & 0.0094 & 0.0264 & 0.0255 & 0.0024 & {0.0312} &       & 0.0132 & 0.0273 & 0.0262 & {0.0064} & {0.0312} \\
    $\alpha_{R}$ &       & 0.0423 & 0.0234 & 0.0253 & 0.0620 & {0.0199} &       & 0.0396 & 0.0232 & 0.0260 & {0.0516} & {0.0213} \\
          &       & \multicolumn{5}{c}{$T=8,000$}             &       & \multicolumn{5}{c}{$T=16,000$} \\
    Mean  &       & 0.9913 & 1.0006 & 0.9998 & 0.9493 & {1.0023} &       & 0.9928 & 1.0005 & 0.9997 & {0.9496} & {1.0020} \\
    Std   &       & 0.0774 & 0.0343 & 0.2178 & 0.0045 & {0.0315} &       & 0.0540 & 0.0242 & 0.1517 & {0.0031} & {0.0221} \\
    aStd  &       & 0.0767 & 0.0344 & 0.2232 & 0.0044 & {0.0325} &       & 0.0542 & 0.0244 & 0.1579 & {0.0031} & {0.0230} \\   
    $\alpha_{L}$ &       & 0.0141 & 0.0281 & 0.0268 & 0.0108 & {0.0320} &       & 0.0156 & 0.0280 & 0.0275 & {0.0146} & {0.0326} \\
    $\alpha_{R}$ &       & 0.0375 & 0.0237 & 0.0265 & 0.0431 & {0.0203} &       & 0.0382 & 0.0244 & 0.0275 & {0.0394} & {0.0202} \\
    \midrule
\bottomrule
\end{tabularx}
\end{footnotesize}

\par\end{centering}
{\footnotesize \textit{Note:} Monte Carlo results for the GCC QMLE based on $N=100,000$ replications. Std denotes the empirical standard deviation, and aStd is computed from the Fisher information using a simulated sample with $T=10,000,000$. The tail frequencies $\alpha_L$ and $\alpha_R$ are based on standardized estimates below $-1.96$ and above $1.96$; the nominal benchmark is $0.025$ for each tail.}{\par}
\end{table}

Across both designs, the QMLE is well centered and the empirical standard deviations closely match the inverse-information counterparts. In the small-Cauchy design, estimates of $\sigma$, $\gamma$, $\mu$, and $\tau$ are nearly unbiased even at moderate sample sizes. The main finite-sample distortion is a downward bias in the persistence parameter $\phi$, with right-tail rejection probabilities above nominal benchmarks. This distortion shrinks as $T$ increases and reflects a dynamic-parameter finite-sample effect rather than an evident failure of the Masreliez approximation.

The balanced design is more challenging because the Gaussian and Cauchy measurement-error scales are equal. The estimator of $\gamma$ remains well centered, while $\hat\sigma$ displays a short-sample downward bias, reflecting the difficulty of separating comparable Gaussian and Cauchy components in short samples. This bias decreases monotonically, and the empirical standard deviations closely match their asymptotic counterparts by $T=8{,}000$ or $T=16{,}000$.

The tail frequencies provide a diagnostic for asymptotic normality. For $\mu$, empirical probabilities are close to the nominal 0.025 benchmark across both panels. Scale and persistence parameters exhibit more finite-sample skewness, especially $\sigma$ in the balanced design and $\phi$ overall, but these deviations diminish with sample size. Overall, the Monte Carlo evidence indicates that the pseudo-true parameter, $\vartheta_\star$, is close to the true parameter, $\vartheta_0$, in the designs considered, and that the inverse-information approximation provides a useful description of sampling variability.

\subsection{Exact benchmark filter and diagnostic definitions}\label{app:simulation-details}

This section provides details on the benchmark filter and diagnostics used in Section \ref{sec:Simulations}. The objective is to assess the numerical consequences of the Gaussian prediction approximation in Assumption \ref{ass:Masreliez}, holding the model parameters fixed at their true values.

Let $\pi_{t|t-1}^\ast$ denote the exact predictive density of $x_t$ conditional
on $\mathcal F_{t-1}$. In the benchmark filter, this density is propagated
numerically without imposing Gaussianity. Given $\pi_{t|t-1}^\ast$, the exact
predictive density of $y_t$ is
$$
f_{t|t-1}^\ast(y)=\int f_\eta(y-x)\pi_{t|t-1}^\ast(x)dx,
$$
where $f_\eta$ is the Voigt measurement-error density. The exact posterior mean
associated with a generic observation value $y$ is
$$
x_{t|t}^\ast(y)=
\frac{\int xf_\eta(y-x)\pi_{t|t-1}^\ast(x)dx}
{\int f_\eta(y-x)\pi_{t|t-1}^\ast(x)dx},
$$
and the corresponding exact correction is
$$
\Delta_t^\ast(y)=x_{t|t}^\ast(y)-x_{t|t-1}^\ast,
\qquad
x_{t|t-1}^\ast=\int x\pi_{t|t-1}^\ast(x)dx.
$$
After observing $y_t$, the exact filtering density is updated by Bayes' rule:
$$
\pi_{t|t}^\ast(x)=
\frac{f_\eta(y_t-x)\pi_{t|t-1}^\ast(x)}
{\int f_\eta(y_t-z)\pi_{t|t-1}^\ast(z)dz}.
$$
The next predictive density is obtained by propagating this filtering density
through the Gaussian state transition,
$$
\pi_{t+1|t}^\ast(x)=
\int
\varphi_{\tau^2}\{x-(1-\phi)\mu-\phi z\}
\pi_{t|t}^\ast(z)dz.
$$
Thus the benchmark filter is exact up to numerical integration error.

We compare this benchmark with two Gaussian approximations. The first is the
moment-matched Gaussian approximation,
$$
\tilde\pi_{t|t-1}(x)=
\varphi_{h_{t|t-1}^\ast}(x-x_{t|t-1}^\ast),
\qquad
h_{t|t-1}^\ast=
\int (x-x_{t|t-1}^\ast)^2\pi_{t|t-1}^\ast(x)dx.
$$
This approximation has the same predictive mean and variance as the exact
predictive density but imposes Gaussian shape. The second is the operational
GCC approximation,
$$
\pi_{t|t-1}(x)=
\varphi_{h_{t|t-1}}(x-x_{t|t-1}),
$$
where $x_{t|t-1}$ and $h_{t|t-1}$ are generated by the GCC recursion in
Theorem \ref{thm:GCCfilter}. The moment-matched approximation isolates the
shape error from Gaussianizing the predictive density, whereas the operational
approximation measures the full discrepancy of the implemented filter.

The corresponding predictive densities of the observation are
$$
\tilde f_{t|t-1}(y)=\int f_\eta(y-x)\tilde\pi_{t|t-1}(x)dx,
\qquad
f_{t|t-1}(y)=\int f_\eta(y-x)\pi_{t|t-1}(x)dx.
$$
We define the density-level diagnostics
$$
\operatorname{KL}_{t}^{x,\operatorname{shape}}
=
\int
\pi_{t|t-1}^\ast(x)
\log
\left(
\frac{\pi_{t|t-1}^\ast(x)}
{\tilde\pi_{t|t-1}(x)}
\right)dx,
\qquad
\operatorname{KL}_{t}^{x,\operatorname{op}}
=
\int
\pi_{t|t-1}^\ast(x)
\log
\left(
\frac{\pi_{t|t-1}^\ast(x)}
{\pi_{t|t-1}(x)}
\right)dx.
$$
The analogous observation-density diagnostics are
$$
\operatorname{KL}_{t}^{y,\operatorname{shape}}
=
\int
f_{t|t-1}^\ast(y)
\log
\left(
\frac{f_{t|t-1}^\ast(y)}
{\tilde f_{t|t-1}(y)}
\right)dy,
\qquad
\operatorname{KL}_{t}^{y,\operatorname{op}}
=
\int
f_{t|t-1}^\ast(y)
\log
\left(
\frac{f_{t|t-1}^\ast(y)}
{f_{t|t-1}(y)}
\right)dy.
$$

For the correction-level diagnostics, let $\tilde x_{t|t}(y)$ be the
posterior mean obtained by replacing $\pi_{t|t-1}^\ast$ with
$\tilde\pi_{t|t-1}$ in the exact posterior-mean formula, and define
$\tilde\Delta_t(y)=\tilde x_{t|t}(y)-x_{t|t-1}^\ast$. The operational
GCC correction is $\Delta_t(y)=x_{t|t}(y)-x_{t|t-1}$, and at the realized
observation it satisfies $\Delta_t(y_t)=h_{t|t-1}\psi_t$. The shape and
operational correction distortions are
$$
\tilde D_t(y)=\Delta_t^\ast(y)-\tilde\Delta_t(y),
\qquad
D_t(y)=\Delta_t^\ast(y)-\Delta_t(y).
$$
The reported correction diagnostics summarize the realized distortions
$\tilde D_t(y_t)$ and $D_t(y_t)$ across simulated paths. Specifically, the
tables report mean absolute error, root mean square error, and selected
quantiles of $|D_t(y_t)|$.

The main text reports diagnostics averaged over the nine designs for each value of $\lambda$. The full design-by-design diagnostics are reported in Supplement~\ref{app:simulation-tables}.

\clearpage
\subsection{Design-by-design approximation diagnostics}\label{app:simulation-tables}

The simulation design exploits scale equivariance. The relevant dimensionless
parameters are $\lambda=\gamma/\sigma$, $\tau/\sigma$, and $\phi$. For each
value of $\lambda$, the main tables average over the nine designs generated by
$\phi\in\{0.90,0.97,0.99\}$ and $\tau/\sigma\in\{0.25,0.50,1.00\}$.

\begingroup
\small
\renewcommand{\arraystretch}{0.75}
\setlength{\tabcolsep}{4pt}
\setlength{\LTleft}{\fill}
\setlength{\LTright}{\fill}

\begin{longtable}{@{}ccc*{6}{r}@{}}
\caption{KL diagnostics by simulation design}
\label{tab:gcc_density_diagnostics_appendix} \\
\toprule
    \midrule
$\lambda$ & $\phi$ & $\tau/\sigma$ & $\overline{\mathrm{KL}}^{x,\operatorname{shape}}$ & $\overline{\mathrm{KL}}^{x,\operatorname{op}}$ & $\overline{\mathrm{KL}}^{y,\operatorname{shape}}$ & $\overline{\mathrm{KL}}^{y,\operatorname{op}}$ & $q_{.95}\mathrm{KL}^{x,\operatorname{op}}$ & $\max\mathrm{KL}^{x,\operatorname{op}}$ \\
\midrule
\endfirsthead
\caption[]{KL diagnostics by simulation design (continued)}\\
\toprule
    \midrule
$\lambda$ & $\phi$ & $\tau/\sigma$ & $\overline{\mathrm{KL}}^{x,\operatorname{shape}}$ & $\overline{\mathrm{KL}}^{x,\operatorname{op}}$ & $\overline{\mathrm{KL}}^{y,\operatorname{shape}}$ & $\overline{\mathrm{KL}}^{y,\operatorname{op}}$ & $q_{.95}\mathrm{KL}^{x,\operatorname{op}}$ & $\max\mathrm{KL}^{x,\operatorname{op}}$ \\
\midrule
\endhead
\midrule
\multicolumn{9}{r}{\emph{continued on next page}}\\
\endfoot
\midrule
\bottomrule
\multicolumn{9}{@{}p{\linewidth}@{}}{\footnotesize\textit{Note:} The table reports time averages after discarding the diagnostic burn-in. The parameter grid is $\lambda=\gamma/\sigma\in\{0,0.01,0.05,0.10,0.50,1.00\}$, $\phi\in\{0.90,0.97,0.99\}$, and $\tau/\sigma\in\{0.25,0.50,1.00\}$. The superscript $x$ refers to the predictive density of the latent state, while $y$ refers to the predictive density of the observation. The ``shape" columns compare the exact predictive density with its moment-matched Gaussian approximation; the ``op" columns compare the exact predictive density with the operational Gaussian approximation used by the GCC filter. The last two columns report the 95th percentile and maximum of the operational state-density KL over time. Because Kullback-Leibler divergence is invariant to common one-to-one rescalings of the compared densities, normalizing $\sigma=1$ entails no loss of generality for fixed values of $\lambda$ and $\tau/\sigma$.}\\
\endlastfoot
$0.00$ & $0.90$ & $0.25$ & $<10^{-15}$ & $<10^{-15}$ & $<10^{-15}$ & $<10^{-15}$ & $<10^{-15}$ & $<10^{-15}$ \\
$0.00$ & $0.90$ & $0.50$ & $<10^{-15}$ & $<10^{-15}$ & $<10^{-15}$ & $<10^{-15}$ & $<10^{-15}$ & $<10^{-15}$ \\
$0.00$ & $0.90$ & $1.00$ & $<10^{-15}$ & $<10^{-15}$ & $<10^{-15}$ & $<10^{-15}$ & $<10^{-15}$ & $<10^{-15}$ \\
$0.00$ & $0.97$ & $0.25$ & $<10^{-15}$ & $<10^{-15}$ & $<10^{-15}$ & $<10^{-15}$ & $<10^{-15}$ & $<10^{-15}$ \\
$0.00$ & $0.97$ & $0.50$ & $<10^{-15}$ & $<10^{-15}$ & $<10^{-15}$ & $<10^{-15}$ & $<10^{-15}$ & $<10^{-15}$ \\
$0.00$ & $0.97$ & $1.00$ & $<10^{-15}$ & $<10^{-15}$ & $<10^{-15}$ & $<10^{-15}$ & $<10^{-15}$ & $<10^{-15}$ \\
$0.00$ & $0.99$ & $0.25$ & $<10^{-15}$ & $<10^{-15}$ & $<10^{-15}$ & $<10^{-15}$ & $<10^{-15}$ & $<10^{-15}$ \\
$0.00$ & $0.99$ & $0.50$ & $<10^{-15}$ & $<10^{-15}$ & $<10^{-15}$ & $<10^{-15}$ & $<10^{-15}$ & $<10^{-15}$ \\
$0.00$ & $0.99$ & $1.00$ & $<10^{-15}$ & $<10^{-15}$ & $<10^{-15}$ & $<10^{-15}$ & $<10^{-15}$ & $<10^{-15}$ \\
$0.01$ & $0.90$ & $0.25$ & $2.63\times 10^{-6}$ & $3.00\times 10^{-6}$ & $9.14\times 10^{-9}$ & $3.79\times 10^{-8}$ & $9.32\times 10^{-6}$ & $1.84\times 10^{-4}$ \\
$0.01$ & $0.90$ & $0.50$ & $4.08\times 10^{-5}$ & $7.78\times 10^{-5}$ & $1.58\times 10^{-6}$ & $1.11\times 10^{-5}$ & $3.92\times 10^{-4}$ & $3.48\times 10^{-3}$ \\
$0.01$ & $0.90$ & $1.00$ & $2.06\times 10^{-4}$ & $2.67\times 10^{-4}$ & $6.19\times 10^{-5}$ & $9.81\times 10^{-5}$ & $2.37\times 10^{-4}$ & $2.78\times 10^{-2}$ \\
$0.01$ & $0.97$ & $0.25$ & $2.75\times 10^{-6}$ & $3.33\times 10^{-6}$ & $1.57\times 10^{-8}$ & $9.62\times 10^{-8}$ & $9.07\times 10^{-6}$ & $2.06\times 10^{-4}$ \\
$0.01$ & $0.97$ & $0.50$ & $4.34\times 10^{-5}$ & $4.80\times 10^{-5}$ & $1.99\times 10^{-6}$ & $3.10\times 10^{-6}$ & $1.23\times 10^{-4}$ & $2.74\times 10^{-3}$ \\
$0.01$ & $0.97$ & $1.00$ & $3.76\times 10^{-5}$ & $3.89\times 10^{-5}$ & $5.71\times 10^{-6}$ & $6.41\times 10^{-6}$ & $7.94\times 10^{-5}$ & $3.69\times 10^{-3}$ \\
$0.01$ & $0.99$ & $0.25$ & $4.66\times 10^{-5}$ & $7.28\times 10^{-5}$ & $6.75\times 10^{-7}$ & $4.49\times 10^{-6}$ & $2.11\times 10^{-4}$ & $4.58\times 10^{-3}$ \\
$0.01$ & $0.99$ & $0.50$ & $7.43\times 10^{-5}$ & $1.12\times 10^{-4}$ & $5.42\times 10^{-6}$ & $1.83\times 10^{-5}$ & $1.18\times 10^{-4}$ & $1.15\times 10^{-2}$ \\
$0.01$ & $0.99$ & $1.00$ & $1.80\times 10^{-4}$ & $1.96\times 10^{-4}$ & $4.37\times 10^{-5}$ & $5.36\times 10^{-5}$ & $3.32\times 10^{-4}$ & $2.09\times 10^{-2}$ \\
$0.05$ & $0.90$ & $0.25$ & $1.11\times 10^{-5}$ & $1.37\times 10^{-5}$ & $3.33\times 10^{-8}$ & $1.97\times 10^{-7}$ & $6.25\times 10^{-5}$ & $3.95\times 10^{-4}$ \\
$0.05$ & $0.90$ & $0.50$ & $8.86\times 10^{-5}$ & $1.31\times 10^{-4}$ & $4.20\times 10^{-6}$ & $1.51\times 10^{-5}$ & $5.68\times 10^{-4}$ & $7.51\times 10^{-3}$ \\
$0.05$ & $0.90$ & $1.00$ & $1.60\times 10^{-4}$ & $2.02\times 10^{-4}$ & $2.76\times 10^{-5}$ & $4.95\times 10^{-5}$ & $7.11\times 10^{-4}$ & $1.20\times 10^{-2}$ \\
$0.05$ & $0.97$ & $0.25$ & $4.54\times 10^{-5}$ & $9.13\times 10^{-5}$ & $3.40\times 10^{-7}$ & $5.97\times 10^{-6}$ & $5.67\times 10^{-4}$ & $1.55\times 10^{-3}$ \\
$0.05$ & $0.97$ & $0.50$ & $2.18\times 10^{-4}$ & $3.12\times 10^{-4}$ & $1.23\times 10^{-5}$ & $3.91\times 10^{-5}$ & $1.67\times 10^{-3}$ & $7.31\times 10^{-3}$ \\
$0.05$ & $0.97$ & $1.00$ & $4.70\times 10^{-4}$ & $5.83\times 10^{-4}$ & $1.13\times 10^{-4}$ & $1.76\times 10^{-4}$ & $1.64\times 10^{-3}$ & $1.98\times 10^{-2}$ \\
$0.05$ & $0.99$ & $0.25$ & $5.19\times 10^{-5}$ & $8.71\times 10^{-5}$ & $5.17\times 10^{-7}$ & $4.41\times 10^{-6}$ & $3.37\times 10^{-4}$ & $5.20\times 10^{-3}$ \\
$0.05$ & $0.99$ & $0.50$ & $1.75\times 10^{-4}$ & $2.51\times 10^{-4}$ & $9.87\times 10^{-6}$ & $3.08\times 10^{-5}$ & $1.21\times 10^{-3}$ & $8.08\times 10^{-3}$ \\
$0.05$ & $0.99$ & $1.00$ & $3.84\times 10^{-4}$ & $4.25\times 10^{-4}$ & $8.57\times 10^{-5}$ & $1.07\times 10^{-4}$ & $1.16\times 10^{-3}$ & $2.13\times 10^{-2}$ \\
$0.10$ & $0.90$ & $0.25$ & $1.74\times 10^{-5}$ & $2.22\times 10^{-5}$ & $4.50\times 10^{-8}$ & $3.43\times 10^{-7}$ & $1.17\times 10^{-4}$ & $5.25\times 10^{-4}$ \\
$0.10$ & $0.90$ & $0.50$ & $1.50\times 10^{-4}$ & $1.95\times 10^{-4}$ & $5.09\times 10^{-6}$ & $1.40\times 10^{-5}$ & $1.09\times 10^{-3}$ & $3.28\times 10^{-3}$ \\
$0.10$ & $0.90$ & $1.00$ & $3.08\times 10^{-4}$ & $3.52\times 10^{-4}$ & $5.10\times 10^{-5}$ & $7.19\times 10^{-5}$ & $1.86\times 10^{-3}$ & $9.08\times 10^{-3}$ \\
$0.10$ & $0.97$ & $0.25$ & $5.28\times 10^{-5}$ & $9.60\times 10^{-5}$ & $2.98\times 10^{-7}$ & $5.63\times 10^{-6}$ & $5.92\times 10^{-4}$ & $1.58\times 10^{-3}$ \\
$0.10$ & $0.97$ & $0.50$ & $3.17\times 10^{-4}$ & $5.98\times 10^{-4}$ & $1.86\times 10^{-5}$ & $1.02\times 10^{-4}$ & $2.24\times 10^{-3}$ & $1.90\times 10^{-2}$ \\
$0.10$ & $0.97$ & $1.00$ & $4.55\times 10^{-4}$ & $5.09\times 10^{-4}$ & $7.19\times 10^{-5}$ & $1.00\times 10^{-4}$ & $2.50\times 10^{-3}$ & $1.29\times 10^{-2}$ \\
$0.10$ & $0.99$ & $0.25$ & $9.01\times 10^{-5}$ & $1.68\times 10^{-4}$ & $6.93\times 10^{-7}$ & $9.66\times 10^{-6}$ & $7.62\times 10^{-4}$ & $3.84\times 10^{-3}$ \\
$0.10$ & $0.99$ & $0.50$ & $3.05\times 10^{-4}$ & $5.42\times 10^{-4}$ & $1.61\times 10^{-5}$ & $9.69\times 10^{-5}$ & $2.16\times 10^{-3}$ & $2.03\times 10^{-2}$ \\
$0.10$ & $0.99$ & $1.00$ & $9.53\times 10^{-4}$ & $1.22\times 10^{-3}$ & $1.90\times 10^{-4}$ & $3.53\times 10^{-4}$ & $6.90\times 10^{-3}$ & $3.50\times 10^{-2}$ \\
$0.50$ & $0.90$ & $0.25$ & $3.11\times 10^{-5}$ & $3.99\times 10^{-5}$ & $3.41\times 10^{-8}$ & $3.91\times 10^{-7}$ & $1.44\times 10^{-4}$ & $3.70\times 10^{-4}$ \\
$0.50$ & $0.90$ & $0.50$ & $2.60\times 10^{-4}$ & $3.51\times 10^{-4}$ & $2.98\times 10^{-6}$ & $1.55\times 10^{-5}$ & $1.20\times 10^{-3}$ & $4.83\times 10^{-3}$ \\
$0.50$ & $0.90$ & $1.00$ & $9.67\times 10^{-4}$ & $1.27\times 10^{-3}$ & $6.95\times 10^{-5}$ & $1.72\times 10^{-4}$ & $4.59\times 10^{-3}$ & $1.50\times 10^{-2}$ \\
$0.50$ & $0.97$ & $0.25$ & $1.12\times 10^{-4}$ & $2.16\times 10^{-4}$ & $2.76\times 10^{-7}$ & $8.08\times 10^{-6}$ & $8.58\times 10^{-4}$ & $3.29\times 10^{-3}$ \\
$0.50$ & $0.97$ & $0.50$ & $6.24\times 10^{-4}$ & $8.34\times 10^{-4}$ & $1.17\times 10^{-5}$ & $4.30\times 10^{-5}$ & $3.27\times 10^{-3}$ & $8.09\times 10^{-3}$ \\
$0.50$ & $0.97$ & $1.00$ & $1.69\times 10^{-3}$ & $2.28\times 10^{-3}$ & $1.42\times 10^{-4}$ & $3.70\times 10^{-4}$ & $1.14\times 10^{-2}$ & $3.56\times 10^{-2}$ \\
$0.50$ & $0.99$ & $0.25$ & $1.90\times 10^{-4}$ & $3.02\times 10^{-4}$ & $6.74\times 10^{-7}$ & $9.35\times 10^{-6}$ & $1.02\times 10^{-3}$ & $4.01\times 10^{-3}$ \\
$0.50$ & $0.99$ & $0.50$ & $8.74\times 10^{-4}$ & $1.67\times 10^{-3}$ & $2.41\times 10^{-5}$ & $2.15\times 10^{-4}$ & $4.68\times 10^{-3}$ & $6.47\times 10^{-2}$ \\
$0.50$ & $0.99$ & $1.00$ & $1.82\times 10^{-3}$ & $2.30\times 10^{-3}$ & $1.60\times 10^{-4}$ & $3.43\times 10^{-4}$ & $9.15\times 10^{-3}$ & $3.64\times 10^{-2}$ \\
$1.00$ & $0.90$ & $0.25$ & $2.22\times 10^{-5}$ & $2.67\times 10^{-5}$ & $7.47\times 10^{-9}$ & $1.09\times 10^{-7}$ & $8.64\times 10^{-5}$ & $1.89\times 10^{-4}$ \\
$1.00$ & $0.90$ & $0.50$ & $2.24\times 10^{-4}$ & $3.24\times 10^{-4}$ & $1.19\times 10^{-6}$ & $1.07\times 10^{-5}$ & $1.13\times 10^{-3}$ & $3.35\times 10^{-3}$ \\
$1.00$ & $0.90$ & $1.00$ & $9.71\times 10^{-4}$ & $1.22\times 10^{-3}$ & $3.27\times 10^{-5}$ & $8.42\times 10^{-5}$ & $4.40\times 10^{-3}$ & $9.17\times 10^{-3}$ \\
$1.00$ & $0.97$ & $0.25$ & $1.42\times 10^{-4}$ & $2.24\times 10^{-4}$ & $1.88\times 10^{-7}$ & $3.39\times 10^{-6}$ & $7.48\times 10^{-4}$ & $2.82\times 10^{-3}$ \\
$1.00$ & $0.97$ & $0.50$ & $7.82\times 10^{-4}$ & $1.38\times 10^{-3}$ & $8.55\times 10^{-6}$ & $8.50\times 10^{-5}$ & $4.68\times 10^{-3}$ & $1.88\times 10^{-2}$ \\
$1.00$ & $0.97$ & $1.00$ & $2.00\times 10^{-3}$ & $2.82\times 10^{-3}$ & $8.84\times 10^{-5}$ & $2.86\times 10^{-4}$ & $9.44\times 10^{-3}$ & $4.11\times 10^{-2}$ \\
$1.00$ & $0.99$ & $0.25$ & $2.30\times 10^{-4}$ & $4.03\times 10^{-4}$ & $5.32\times 10^{-7}$ & $8.25\times 10^{-6}$ & $1.25\times 10^{-3}$ & $8.91\times 10^{-3}$ \\
$1.00$ & $0.99$ & $0.50$ & $8.97\times 10^{-4}$ & $1.49\times 10^{-3}$ & $1.01\times 10^{-5}$ & $9.65\times 10^{-5}$ & $5.12\times 10^{-3}$ & $3.59\times 10^{-2}$ \\
$1.00$ & $0.99$ & $1.00$ & $2.32\times 10^{-3}$ & $3.28\times 10^{-3}$ & $1.18\times 10^{-4}$ & $3.84\times 10^{-4}$ & $1.15\times 10^{-2}$ & $3.43\times 10^{-2}$ \\
\end{longtable}
\endgroup

\begingroup
\small
\renewcommand{\arraystretch}{0.75}
\setlength{\tabcolsep}{5pt}
\setlength{\LTleft}{\fill}
\setlength{\LTright}{\fill}

\begin{longtable}{ *{3}{>{\centering\arraybackslash}p{0.09\linewidth}}  *{4}{>{\centering\arraybackslash}p{0.14\linewidth}}  }
\caption[]{Correction diagnostics by simulation design}\\
\toprule
    \midrule
$\lambda$ & $\phi$ & $\tau/\sigma$
& $\operatorname{MAE}_{s}$
& $\operatorname{MAE}_{o}$
& $\operatorname{RMSE}_{o}$
& $q_{.95}|D_{o}|$ \\
\midrule
\endfirsthead

\caption[]{Correction diagnostics by simulation design (continued)}\\
\toprule
    \midrule
$\lambda$ & $\phi$ & $\tau/\sigma$
& $\operatorname{MAE}_{s}$
& $\operatorname{MAE}_{o}$
& $\operatorname{RMSE}_{o}$
& $q_{.95}|D_{o}|$ \\
\midrule
\endhead
\midrule
\multicolumn{7}{r}{\emph{continued on next page}}\\
\endfoot
\midrule
\bottomrule
\multicolumn{7}{@{}p{\linewidth}@{}}{\footnotesize\textit{Note:} The table reports correction-level diagnostics for each simulation design. The distortion is $D_t(y_t)=\Delta_t^\ast(y_t)-\Delta_t(y_t)$. Subscript $s$ denotes the shape diagnostic, based on moment-matched Gaussian approximation to exact predictive density; subscript $o$ denotes operational GCC diagnostic. Entries are computed after discarding the diagnostic burn-in period. The case $\lambda=0$ corresponds to the purely Gaussian measurement-error benchmark, for which the approximation is exact up to numerical integration error.}\\
\endlastfoot
$0.00$ & $0.90$ & $0.25$ & $<10^{-15}$          & $<10^{-15}$          & $<10^{-15}$          & $<10^{-15}$          \\
$0.00$ & $0.90$ & $0.50$ & $<10^{-15}$          & $<10^{-15}$          & $<10^{-15}$          & $<10^{-15}$          \\
$0.00$ & $0.90$ & $1.00$ & $<10^{-15}$          & $<10^{-15}$          & $<10^{-15}$          & $<10^{-15}$          \\
$0.00$ & $0.97$ & $0.25$ & $<10^{-15}$          & $<10^{-15}$          & $<10^{-15}$          & $<10^{-15}$          \\
$0.00$ & $0.97$ & $0.50$ & $<10^{-15}$          & $<10^{-15}$          & $<10^{-15}$          & $<10^{-15}$          \\
$0.00$ & $0.97$ & $1.00$ & $<10^{-15}$          & $<10^{-15}$          & $<10^{-15}$          & $1.55\times 10^{-15}$ \\
$0.00$ & $0.99$ & $0.25$ & $<10^{-15}$          & $<10^{-15}$          & $<10^{-15}$          & $<10^{-15}$          \\
$0.00$ & $0.99$ & $0.50$ & $<10^{-15}$          & $<10^{-15}$          & $<10^{-15}$          & $<10^{-15}$          \\
$0.00$ & $0.99$ & $1.00$ & $<10^{-15}$          & $<10^{-15}$          & $1.03\times 10^{-15}$ & $2.27\times 10^{-15}$ \\
$0.01$ & $0.90$ & $0.25$ & $3.94\times 10^{-5}$ & $6.03\times 10^{-5}$ & $1.36\times 10^{-4}$ & $2.14\times 10^{-4}$ \\
$0.01$ & $0.90$ & $0.50$ & $6.88\times 10^{-4}$ & $1.17\times 10^{-3}$ & $4.17\times 10^{-3}$ & $5.67\times 10^{-3}$ \\
$0.01$ & $0.90$ & $1.00$ & $1.80\times 10^{-3}$ & $2.88\times 10^{-3}$ & $1.82\times 10^{-2}$ & $4.81\times 10^{-3}$ \\
$0.01$ & $0.97$ & $0.25$ & $7.82\times 10^{-5}$ & $1.27\times 10^{-4}$ & $2.44\times 10^{-4}$ & $4.75\times 10^{-4}$ \\
$0.01$ & $0.97$ & $0.50$ & $3.47\times 10^{-4}$ & $6.10\times 10^{-4}$ & $1.37\times 10^{-3}$ & $2.62\times 10^{-3}$ \\
$0.01$ & $0.97$ & $1.00$ & $6.10\times 10^{-4}$ & $8.51\times 10^{-4}$ & $1.93\times 10^{-3}$ & $4.37\times 10^{-3}$ \\
$0.01$ & $0.99$ & $0.25$ & $3.15\times 10^{-4}$ & $5.36\times 10^{-4}$ & $1.66\times 10^{-3}$ & $2.20\times 10^{-3}$ \\
$0.01$ & $0.99$ & $0.50$ & $4.46\times 10^{-4}$ & $1.13\times 10^{-3}$ & $6.25\times 10^{-3}$ & $3.37\times 10^{-3}$ \\
$0.01$ & $0.99$ & $1.00$ & $1.55\times 10^{-3}$ & $2.50\times 10^{-3}$ & $8.52\times 10^{-3}$ & $1.07\times 10^{-2}$ \\
$0.05$ & $0.90$ & $0.25$ & $1.55\times 10^{-4}$ & $2.51\times 10^{-4}$ & $4.88\times 10^{-4}$ & $1.19\times 10^{-3}$ \\
$0.05$ & $0.90$ & $0.50$ & $1.19\times 10^{-3}$ & $1.74\times 10^{-3}$ & $4.89\times 10^{-3}$ & $8.62\times 10^{-3}$ \\
$0.05$ & $0.90$ & $1.00$ & $2.43\times 10^{-3}$ & $3.47\times 10^{-3}$ & $1.10\times 10^{-2}$ & $1.21\times 10^{-2}$ \\
$0.05$ & $0.97$ & $0.25$ & $5.17\times 10^{-4}$ & $1.01\times 10^{-3}$ & $2.27\times 10^{-3}$ & $4.20\times 10^{-3}$ \\
$0.05$ & $0.97$ & $0.50$ & $1.93\times 10^{-3}$ & $3.16\times 10^{-3}$ & $7.91\times 10^{-3}$ & $1.57\times 10^{-2}$ \\
$0.05$ & $0.97$ & $1.00$ & $4.61\times 10^{-3}$ & $6.81\times 10^{-3}$ & $2.10\times 10^{-2}$ & $2.44\times 10^{-2}$ \\
$0.05$ & $0.99$ & $0.25$ & $4.62\times 10^{-4}$ & $9.18\times 10^{-4}$ & $2.60\times 10^{-3}$ & $2.96\times 10^{-3}$ \\
$0.05$ & $0.99$ & $0.50$ & $1.92\times 10^{-3}$ & $3.08\times 10^{-3}$ & $7.91\times 10^{-3}$ & $1.12\times 10^{-2}$ \\
$0.05$ & $0.99$ & $1.00$ & $3.33\times 10^{-3}$ & $4.78\times 10^{-3}$ & $1.13\times 10^{-2}$ & $1.52\times 10^{-2}$ \\
$0.10$ & $0.90$ & $0.25$ & $1.89\times 10^{-4}$ & $3.43\times 10^{-4}$ & $6.31\times 10^{-4}$ & $1.16\times 10^{-3}$ \\
$0.10$ & $0.90$ & $0.50$ & $1.83\times 10^{-3}$ & $2.50\times 10^{-3}$ & $5.10\times 10^{-3}$ & $1.17\times 10^{-2}$ \\
$0.10$ & $0.90$ & $1.00$ & $3.88\times 10^{-3}$ & $6.07\times 10^{-3}$ & $1.38\times 10^{-2}$ & $2.35\times 10^{-2}$ \\
$0.10$ & $0.97$ & $0.25$ & $5.63\times 10^{-4}$ & $1.03\times 10^{-3}$ & $2.13\times 10^{-3}$ & $4.01\times 10^{-3}$ \\
$0.10$ & $0.97$ & $0.50$ & $3.69\times 10^{-3}$ & $5.35\times 10^{-3}$ & $1.21\times 10^{-2}$ & $2.20\times 10^{-2}$ \\
$0.10$ & $0.97$ & $1.00$ & $4.88\times 10^{-3}$ & $6.59\times 10^{-3}$ & $1.13\times 10^{-2}$ & $2.66\times 10^{-2}$ \\
$0.10$ & $0.99$ & $0.25$ & $9.81\times 10^{-4}$ & $1.99\times 10^{-3}$ & $5.02\times 10^{-3}$ & $8.08\times 10^{-3}$ \\
$0.10$ & $0.99$ & $0.50$ & $3.56\times 10^{-3}$ & $6.26\times 10^{-3}$ & $1.79\times 10^{-2}$ & $2.04\times 10^{-2}$ \\
$0.10$ & $0.99$ & $1.00$ & $1.11\times 10^{-2}$ & $1.43\times 10^{-2}$ & $4.73\times 10^{-2}$ & $5.92\times 10^{-2}$ \\
$0.50$ & $0.90$ & $0.25$ & $3.19\times 10^{-4}$ & $5.27\times 10^{-4}$ & $8.00\times 10^{-4}$ & $1.91\times 10^{-3}$ \\
$0.50$ & $0.90$ & $0.50$ & $2.93\times 10^{-3}$ & $4.60\times 10^{-3}$ & $7.52\times 10^{-3}$ & $1.46\times 10^{-2}$ \\
$0.50$ & $0.90$ & $1.00$ & $1.53\times 10^{-2}$ & $1.93\times 10^{-2}$ & $3.51\times 10^{-2}$ & $7.27\times 10^{-2}$ \\
$0.50$ & $0.97$ & $0.25$ & $1.01\times 10^{-3}$ & $1.99\times 10^{-3}$ & $3.38\times 10^{-3}$ & $7.50\times 10^{-3}$ \\
$0.50$ & $0.97$ & $0.50$ & $5.96\times 10^{-3}$ & $8.74\times 10^{-3}$ & $1.44\times 10^{-2}$ & $2.71\times 10^{-2}$ \\
$0.50$ & $0.97$ & $1.00$ & $1.86\times 10^{-2}$ & $2.61\times 10^{-2}$ & $4.77\times 10^{-2}$ & $8.18\times 10^{-2}$ \\
$0.50$ & $0.99$ & $0.25$ & $1.50\times 10^{-3}$ & $2.80\times 10^{-3}$ & $4.88\times 10^{-3}$ & $1.02\times 10^{-2}$ \\
$0.50$ & $0.99$ & $0.50$ & $8.40\times 10^{-3}$ & $1.42\times 10^{-2}$ & $3.09\times 10^{-2}$ & $4.33\times 10^{-2}$ \\
$0.50$ & $0.99$ & $1.00$ & $2.02\times 10^{-2}$ & $2.47\times 10^{-2}$ & $4.40\times 10^{-2}$ & $8.39\times 10^{-2}$ \\
$1.00$ & $0.90$ & $0.25$ & $2.32\times 10^{-4}$ & $3.72\times 10^{-4}$ & $5.29\times 10^{-4}$ & $1.00\times 10^{-3}$ \\
$1.00$ & $0.90$ & $0.50$ & $2.86\times 10^{-3}$ & $4.12\times 10^{-3}$ & $6.70\times 10^{-3}$ & $1.50\times 10^{-2}$ \\
$1.00$ & $0.90$ & $1.00$ & $1.34\times 10^{-2}$ & $1.80\times 10^{-2}$ & $2.75\times 10^{-2}$ & $6.45\times 10^{-2}$ \\
$1.00$ & $0.97$ & $0.25$ & $1.08\times 10^{-3}$ & $2.07\times 10^{-3}$ & $3.17\times 10^{-3}$ & $6.62\times 10^{-3}$ \\
$1.00$ & $0.97$ & $0.50$ & $7.33\times 10^{-3}$ & $1.17\times 10^{-2}$ & $2.17\times 10^{-2}$ & $4.57\times 10^{-2}$ \\
$1.00$ & $0.97$ & $1.00$ & $2.34\times 10^{-2}$ & $3.12\times 10^{-2}$ & $5.32\times 10^{-2}$ & $9.23\times 10^{-2}$ \\
$1.00$ & $0.99$ & $0.25$ & $1.92\times 10^{-3}$ & $3.82\times 10^{-3}$ & $6.77\times 10^{-3}$ & $1.27\times 10^{-2}$ \\
$1.00$ & $0.99$ & $0.50$ & $7.27\times 10^{-3}$ & $1.17\times 10^{-2}$ & $2.05\times 10^{-2}$ & $3.51\times 10^{-2}$ \\
$1.00$ & $0.99$ & $1.00$ & $2.51\times 10^{-2}$ & $3.56\times 10^{-2}$ & $5.51\times 10^{-2}$ & $1.17\times 10^{-1}$ \\
\end{longtable}
\endgroup

%%%%%%%%%%%%%%%%%%%%%%%%%%%%%%%%%%%%%%%%%%%%%%

\setcounter{table}{0}
\global\long\def\thetable{D.\arabic{table}}%

\section{Details on Competing Filters}
\label{app:filter_details}

This appendix provides additional details on the competing filters used in the empirical comparison. All specifications use the same Gaussian AR(1) state equation,
$$
x_t=(1-\phi)\mu+\phi x_{t-1}+\varepsilon_t,
\qquad
\varepsilon_t\sim\mathcal{N}(0,\tau^2),
$$
and differ only in the measurement-error distribution in
$$
y_t=x_t+\eta_t.
$$
The comparison separates two issues: whether the measurement-error distribution is flexible enough to accommodate realized-volatility outliers, and whether convolution with Gaussian state-prediction uncertainty preserves an analytically tractable prediction-error density.

\begin{table}[tbp]
\caption{Measurement-error specifications used in the empirical comparison}
\label{tab:filter_comparison}
\spacingset{1.15}
\begin{centering}
\begin{footnotesize}
\begin{tabularx}{\textwidth}{@{} l >{\hsize=0.9\hsize}Y >{\hsize=1.1\hsize}Y c @{}}
\toprule
\midrule
Filter & Measurement error $\eta_t$ & Prediction-error density for $e_t$ & Closed form \\
\midrule
Gaussian & $\mathcal{N}(0,\sigma^2)$
& Gaussian with variance $h_{t|t-1}+\sigma^2$ & Yes \\

Cauchy & $\operatorname{Cauchy}(0,\gamma)$
& Voigt convolution of $\mathcal{N}(0,h_{t|t-1})$ and Cauchy error & Yes \\

GCC & $\mathcal{V}(0,\sigma,\gamma)$
& Voigt with Gaussian scale $\delta_t=\sqrt{h_{t|t-1}+\sigma^2}$ and Cauchy scale $\gamma$ & Yes \\

Normal-Laplace & $\mathcal{N}(0,\sigma^2)+\operatorname{Laplace}(0,\gamma)$
& Normal-Laplace with Gaussian scale $\sqrt{h_{t|t-1}+\sigma^2}$ & Yes \\

Student-$t$ & $t_\nu(0,\sigma)$
& Gaussian-$t$ convolution; not generally a Student-$t$ density & No \\

Huber & Huber density with Gaussian core and exponential tails
& Gaussian-Huber convolution; no elementary closed form & No \\
\midrule
\bottomrule
\end{tabularx}
\end{footnotesize}
\par\end{centering}
{\footnotesize \textit{Note:} All filters use the same Gaussian AR(1) state equation. The distinction is the measurement-error distribution and whether convolution with the Gaussian prediction error preserves an analytically tractable prediction-error density.}
\end{table}

The Gaussian and pure Cauchy filters represent the two limiting cases. The GCC filter combines the Gaussian core with a Cauchy component, while the Normal-Laplace filter provides a light-tailed robust alternative whose Gaussian convolution remains available in closed form. The Student-$t$ and Huber filters are included as standard robust benchmarks, but their prediction-error densities are not closed under Gaussian state uncertainty. They are therefore implemented using same-family prediction-error approximations.

\begin{enumerate}
\item \textbf{Gaussian Kalman filter.}
The Gaussian filter assumes
$$
\eta_t\sim\mathcal{N}(0,\sigma^2).
$$
Under the Gaussian prediction approximation, the prediction error is Gaussian,
$$
e_t|\mathcal{F}_{t-1}\sim\mathcal{N}(0,h_{t|t-1}+\sigma^2),
$$
which gives the standard Kalman filter.

\item \textbf{Cauchy filter.}
The Cauchy filter assumes
$$
\eta_t\sim\operatorname{Cauchy}(0,\gamma).
$$
This is the limiting case of the GCC measurement-error specification with no Gaussian measurement-error component. Under the Gaussian prediction approximation, the prediction error is the convolution of $\mathcal{N}(0,h_{t|t-1})$ and $\operatorname{Cauchy}(0,\gamma)$. Hence
$$
e_t|\mathcal{F}_{t-1}\sim\mathcal{V}(0,\sqrt{h_{t|t-1}},\gamma).
$$

\item \textbf{Gauss-Cauchy convolution filter.}
The GCC filter assumes
$$
\eta_t\sim\mathcal{V}(0,\sigma,\gamma),
$$
or equivalently
$$
\eta_t=Z_t+C_t,
\qquad
Z_t\sim\mathcal{N}(0,\sigma^2),
\qquad
C_t\sim\operatorname{Cauchy}(0,\gamma),
$$
with $Z_t$ and $C_t$ independent. Under the Gaussian prediction approximation,
$$
e_t|\mathcal{F}_{t-1}\sim\mathcal{V}(0,\delta_t,\gamma),
\qquad
\delta_t^2=h_{t|t-1}+\sigma^2.
$$
Thus the prediction-error density, score, and filtering update are available in closed form.

\item \textbf{Normal-Laplace filter.}
The Normal-Laplace filter assumes that the measurement error is the convolution of a Gaussian and a Laplace random variable,
$$
\eta_t=Z_t+L_t,
\qquad
Z_t\sim\mathcal{N}(0,\sigma^2),
\qquad
L_t\sim\operatorname{Laplace}(0,\gamma),
$$
as in \citet{Reed:2006}. Its density is
$$
f_{\mathrm{NL}}(\eta;\sigma,\gamma)
=
\frac{1}{2\gamma}
\exp\left(\frac{\sigma^2}{2\gamma^2}\right)
\left[
\exp\left(-\frac{\eta}{\gamma}\right)
\Phi\left(\frac{\eta}{\sigma}-\frac{\sigma}{\gamma}\right)
+
\exp\left(\frac{\eta}{\gamma}\right)
\Phi\left(-\frac{\eta}{\sigma}-\frac{\sigma}{\gamma}\right)
\right],
$$
where $\Phi$ is the standard normal distribution function. Because the state-prediction uncertainty is Gaussian, the prediction error remains Normal-Laplace under the Gaussian prediction approximation, with Gaussian scale
$$
\delta_t=\sqrt{h_{t|t-1}+\sigma^2}
$$
and Laplace scale $\gamma$. Thus the likelihood contribution and score can be evaluated in closed form.

\item \textbf{Huber filter.}
The Huber filter uses the density generated by the Huber loss,
$$
f_{\mathrm{H}}(\eta;\sigma,k)
=
\frac{1}{\sigma c(k)}
\exp\left[-\rho_k\left(\frac{\eta}{\sigma}\right)\right],
$$
where
$$
\rho_k(r)=
\begin{cases}
r^2/2, & |r|\leq k,\\
k|r|-k^2/2, & |r|>k,
\end{cases}
$$
and
$$
c(k)=
\sqrt{2\pi}\operatorname{erf}\left(\frac{k}{\sqrt{2}}\right)
+
\frac{2}{k}\exp\left(-\frac{k^2}{2}\right).
$$
The parameter $k$ is the standardized threshold separating the Gaussian core from the exponential tails. The density is continuous at $|\eta|=k\sigma$, but its convolution with Gaussian state-prediction uncertainty does not have an elementary closed form. We therefore use the same-family approximation
$$
f_{e_t}(e|\mathcal{F}_{t-1})
\approx
f_{\mathrm{H}}(e;\sigma_{e,t},k),
\qquad
\sigma_{e,t}^2=h_{t|t-1}+\sigma^2.
$$
This approximation preserves the Huber score shape but is not the exact Gaussian-Huber convolution.

\item \textbf{Student-$t$ filter.}
The Student-$t$ filter assumes
$$
\eta_t\sim t_\nu(0,\sigma),
$$
where $\sigma$ is the scale parameter. The sum of a Gaussian random variable and a Student-$t$ random variable is not generally Student-$t$, except in special limiting cases. Hence the prediction-error density is not available in the same closed form as in the Gaussian, Cauchy, GCC, and Normal-Laplace specifications. We use the same-family approximation
$$
f_{e_t}(e|\mathcal{F}_{t-1})
\approx
t_\nu(e;0,\sigma_{e,t}),
\qquad
\sigma_{e,t}^2=h_{t|t-1}+\sigma^2.
$$
This approximation is computationally convenient, but it is not the exact Gaussian-$t$ convolution and may distort the likelihood contribution and score when $\nu$ is small.
\end{enumerate}

\end{document}